\documentclass[12pt]{article}
\usepackage[utf8]{inputenc}
\usepackage[T1]{fontenc}
\usepackage[english]{babel}
\usepackage{csquotes}

\usepackage{geometry}
\usepackage{setspace}
\usepackage{graphicx}
\usepackage{xcolor}
\definecolor{KUblue}{RGB}{0, 32, 91}
\usepackage{caption}
\usepackage{subcaption}

\usepackage{comment}

\usepackage{booktabs}
\usepackage{threeparttable}
\usepackage{float}

\usepackage{dsfont}

\usepackage{mathtools,amssymb,mathrsfs,bm}
\allowdisplaybreaks
\DeclareMathOperator*{\argmax}{arg\,max}
\DeclareMathOperator*{\argmin}{arg\,min}
\DeclareMathOperator{\expect}{\mathbb{E}}

\DeclareMathOperator{\ind}{1}
\DeclareMathOperator{\imr}{IMR}
\DeclareMathOperator{\lam}{LAM}

\DeclareMathOperator{\plugin}{plug-in}
\DeclareMathOperator{\rimr}{RIMR}

\DeclarePairedDelimiter{\abs}{\lvert}{\rvert}
\DeclarePairedDelimiter{\dangle}{\langle}{\rangle}
\DeclarePairedDelimiter{\dbrack}{\lbrack}{\rbrack}
\DeclarePairedDelimiter{\dbrace}{\lbrace}{\rbrace}
\DeclarePairedDelimiter{\dparen}{\lparen}{\rparen}
\DeclarePairedDelimiter{\norm}{\lVert}{\rVert}
\DeclarePairedDelimiter{\normB}{\lVert}{\rVert_{\mathbb{B}}}

\usepackage[thinc]{esdiff} 

\usepackage{amsthm}
\theoremstyle{plain}
\newtheorem{theorem}{Theorem}
\newtheorem{lemma}{Lemma}

\theoremstyle{definition}
\newtheorem{assumption}{Assumption}

\newtheorem{example}{Example}
\newtheorem*{excont}{Example~\continuation}
\newcommand{\continuation}{??}
\newenvironment{examplecontinued}[1]
 {\renewcommand{\continuation}{\ref*{#1}}\excont[continued]}
 {\endexcont}

\theoremstyle{remark}
\newtheorem{remark}{Remark}

\usepackage{enumitem}
\setlist[enumerate]{label=(\roman*)}
\newlist{lemmaenum}{enumerate}{1} 
\setlist[lemmaenum]{label=(\roman*), ref=\thelemma.(\roman*)}
\newlist{assumptionenum}{enumerate}{1} 
\setlist[assumptionenum]{label=(\roman*), ref=\theassumption.(\roman*),nosep}

\usepackage{hyperref}
\hypersetup{colorlinks=true,allcolors=blue,bookmarksnumbered}

\usepackage[backend = biber,style = apa,sorting = nyt,sortcites = true,maxnames = 8,minnames = 3,hyperref = true,date = year,natbib = true,dashed = true,giveninits = true]{biblatex}
\usepackage{url}
\usepackage[capitalize,noabbrev,sort]{cleveref}
\crefname{theorem}{Theorem}{Theorems}
\Crefname{theorem}{Theorem}{Theorems}
\crefname{lemma}{Lemma}{Lemmas}
\Crefname{lemma}{Lemma}{Lemmas}
\crefname{corollary}{Corollary}{Corollaries}
\Crefname{corollary}{Corollary}{Corollaries}
\crefname{proposition}{Proposition}{Propositions}
\Crefname{proposition}{Proposition}{Propositions}
\crefname{assumption}{Assumption}{Assumptions}
\Crefname{assumption}{Assumption}{Assumptions}
\crefname{example}{Example}{Examples}
\Crefname{example}{Example}{Examples}
\crefname{equation}{}{}
\Crefname{equation}{}{}
\creflabelformat{equation}{(#2#1#3)}
\crefalias{lemmaenumi}{lemma}
\crefalias{assumptionenumi}{assumption}

\usepackage{subfiles}

\title{
    Locally Asymptotically Minimax Statistical Treatment Rules Under Partial Identification%
    \thanks{I acknowledge Takanori Ida for his continuous support. I also acknowledge Tim Armstrong, Keisuke Hirano, Hidehiko Ichimura, Hiroaki Kaido, Toru Kitagawa, Ryo Okui, Yoshiyasu Rai, Aleksey Tetenov Takahide Yanagi, and Kohei Yata for their helpful comments. Finally, I acknowledge seminar participants at the Bravo Center/JEA/SNSF Workshop at Brown University and the Summer Econometrics Forum at the University of Tokyo.}
}
\author{Daido Kido%
    \thanks{Graduate School of Economics, Kyoto University. E-mail: \href{mailto:daido.kido@gmail.com}{\texttt{daido.kido@gmail.com}}}}
\date{\today}

\begin{document}

\maketitle

\begin{abstract}
    Policymakers often desire a \emph{statistical treatment rule (STR)} that determines a treatment assignment rule deployed in a future population from available data. With the true knowledge of the data generating process, the average treatment effect (ATE) is the key quantity characterizing the optimal treatment rule. Unfortunately, the ATE is often not point identified but partially identified. Presuming the partial identification of the ATE, this study conducts a local asymptotic analysis and develops the \emph{locally asymptotically minimax (LAM)} STR. The analysis does not assume the full differentiability but the directional differentiability of the boundary functions of the identification region of the ATE. Accordingly, the study shows that the LAM STR differs from the \emph{plug-in} STR. A simulation study also demonstrates that the LAM STR outperforms the plug-in STR.
\end{abstract}

\section{Introduction}\label{sec:intro}

Consider the problem of choosing a future treatment for the population of interest from two alternatives, say ``treat'' and ``do not treat''. The objective is to maximize the population's mean outcome, such as average income and cure rate from a specific disease. In making the decision, data that is somewhat informative about the performance of each treatment is available. A data-dependent procedure that effectively chooses the future treatment based on the obtained information from the data is, thus, desirable. \textcite{Manski2004} calls such data-dependent procedures \emph{statistical treatment rules (STRs)} and proposes to analyze this problem within the framework of statistical decision theory to derive the optimal STR. Since then, studies on STRs have developed under various settings \parencite[e.g.,][]{Stoye2009,Hirano2009,Tetenov2012,Kitagawa2018,Mbakop2021,Athey2021}.
\par
A common presumption in these studies is the \emph{point identification} of the (conditional) \emph{average treatment effect (ATE)} in the target population. In other words, if the true data generating process (dgp), the distribution from which the data is drawn, was known, the target population's ATE could be uniquely determined. It would then be optimal to choose the future treatment depending on the sign of the ATE. Although the true dgp is not known in practice, the preceding argument suggests that one of the most intuitive STRs is the empirical success rule that chooses ``treat'' if and only if the estimated ATE is positive \parencite{Manski2004}. This rule has several desirable properties. Specifically, \textcite{Stoye2009} shows that this STR is a nearly finite-sample minimax regret STR under several circumstances. Moreover, \textcite{Hirano2009} shows that it is asymptotically minimax optimal.
\par
However, in practice, the ATE is frequently not point identified but \emph{partially identified}. For example, one might not assume the random assignment of the treatments in an observational study. Even in experimental studies, the experiment participants may not be the random sample from the population of interest. Moreover, non-compliance with assigned treatment may be allowed. In these cases, the ATE of the target population cannot be point identified but partially identified. Namely, even the full knowledge of the true dgp cannot uniquely pin down the value of the ATE. Rather, the knowledge only gives the identification region, the set to which the ATE belongs. The identification region can take various forms depending on the identifying assumption researchers impose on dgps. Still, the set is not generally a singleton. 
\par
When the ATE is partially identified, it is troubling for two reasons. First, as the identification region depends on the dgp, it must be estimated from the data. Thus, there is the problem of statistical uncertainty. This is the fundamental feature any statistical decision problem has in common. Second, superior treatment might be ambiguous even with the true knowledge of the dgp. Suppose the true identification region contains positive and negative values. That is, ``treat'' may be strictly better than ``do not treat'' in the target population, and vice versa. Hence, it is unclear which treatment should be implemented in the population, even with the true knowledge of the identification region. This is the challenge inherent in the treatment decision problem under partial identification of the ATE.
\par
Although the partial identification of the ATE adds ambiguity to the treatment decision problem, the framework of statistical decision theory still applies. In particular, presuming the partial identification of the ATE, several studies have derived finite-sample minimax regret STRs by restricting the class of dgps to some analytically tractable classes or focusing on a specific identifying assumption. Specifically, \textcite{Ishihara2021,Stoye2012,Yata2021} solve the statistical decision problem exactly by restricting the class of dgps to Gaussian experiments, where the data comes from a normal distribution with unknown mean and known covariance matrix. \textcite{Manski2007} also solves the finite-sample problem exactly when the partial identification is conducted with empirical evidence alone.
\par
Unfortunately, it is challenging to solve the statistical treatment decision problem for a general class of dgps induced from arbitrary identification assumptions. Instead, in the same spirit as \textcite{Hirano2009}, this study conducts a local asymptotic analysis to derive asymptotically optimal STRs. 
A brief discussion of the setup and results are as follows. The setup comprises semi- or non-parametric models. Suppose the lower and upper bounds of the identification region can be written as $\tau_L(\theta)$ and $\tau_U(\theta)$. Here, $\theta$ is the possibly infinite-dimensional intermediate parameter that summarizes some features of the dgp relevant to the region. Denote by $\theta_0$ the value of the intermediate parameter at the true dgp. The boundary functions $\tau_L$ and $\tau_U$ are known real-valued functions from theory. Assume that these functions are \emph{Hadamard directionally differentiable}, a version of directional differentiability for maps between two normed spaces. This is the key assumption that determines the scope of this study. Do not assume full differentiability because $\tau_L$ and $\tau_U$ are often not fully differentiable. For instance, $\max$ and $\min$ are not fully differentiable but directionally differentiable, although they often appear in the bounds \parencite[e.g.,][]{Manski1990,Manski2000a,Balke1997}.
\par
Allow for nondeterministic treatment rules: ``treat'' and ``do not treat'' can be implemented with positive probabilities in the target population. As noted, when the ATE is partially identified, the optimal treatment cannot be uniquely determined even if the true dgp is known. Thus, to address the ambiguity arising from the partial identification of the ATE, use the minimax regret criterion to uniquely determine the optimal treatment rule that should be implemented if the true dgp were known. \textcite{Manski2007,Manski2009} have shown that the optimal treatment rule is generally nondeterministic. Specifically, if $\theta_0$ were known, then the optimal treatment rule would be
\begin{equation}
    \kappa(\theta_0) \coloneqq \frac{\tau_U^+(\theta_0)}{\tau_U^+(\theta_0) + \tau_L^-(\theta_0)},
    \label{eq:optimal-treatment-rule}
\end{equation}
where $f^+(\cdot) = \max\dbrace{f(\cdot),0}$ and $f^-(\cdot) = \max\dbrace{-f(\cdot),0}$ for arbitrary real-valued function $f$. Concretely, ``treat'' would be implemented with probability $\kappa(\theta_0)$, while ``do not treat'' would be implemented with probability $1-\kappa(\theta_0)$. In particular, if the identification region contains positive and negative values, this rule is certainly nondeterministic. When the boundary functions, $\tau_L$ and $\tau_U$, are only directionally differentiable, the optimal treatment rule is also only directionally differentiable.
\par
The study seeks an asymptotically optimal STR. Formally, an STR is a function that maps data to a treatment rule. The asymptotic optimality criterion investigated here is the \emph{local asymptotic minimaxity (LAM)}. Given a risk function that measures the loss from the repeated use of an STR for a fixed dgp, the local asymptotic maximum risk of the STR is the worst-case risk over neighborhoods of the true dgp. Then, the LAM STR is the STR that minimizes the local asymptotic maximum risk. 
\par
The main result of this study implies that the LAM STR takes the form of 
\begin{equation*}
    \kappa\dparen*{\hat{\theta}_n + \frac{w_{\theta_0}^*}{\sqrt{n}}},
\end{equation*}
where $\hat{\theta}_n$ is the efficient estimator for $\theta_0$, and $w_{\theta_0}^*$ are the optimal adjustment terms at $\theta_0$. The adjustment term $w_{\theta_0}^*$ is not necessarily zero. Hence, the \emph{plug-in} STR (i.e., $\kappa(\hat{\theta}_n)$) is not generally LAM. The optimal adjustment term generally depends on $\theta_0$; hence, it is unknown in practice unless $\theta_0$ is known. Instead, this study employs a procedure to determine a data-dependent adjustment term $\hat{w}_n$ and show that the feasible STR $\kappa(\hat{\theta}_n + \hat{w}_n/\sqrt{n})$ is LAM. A simulation study exemplifies the proposed STR relative to the plug-in STR.
\paragraph{Related Literature.}{
This study mainly contributes to a line of research that studies the statistical treatment choice problem with partially identified ATE. As noted, \textcite{Manski2007,Stoye2012,Yata2021,Ishihara2021} have derived exact finite-sample minimax regret STRs by restricting the class of data generating process such that the finite-sample problem becomes analytically tractable. In a similar split of this study, \textcite{Christensen2022} have also conducted the local asymptotic analysis to derive asymptotically optimal STR under partial identification of the ATE. The major difference from their analysis is the adopted optimality criterion: while they analyze the asymptotic average risk optimality, this study analyzes the asymptotic minimax optimality. Further, they restrict themselves to deterministic treatment rules, while this study allows for nondeterministic treatment rules. Several studies examine \emph{individualized} treatment assignment problems using observable covariates \parencite{Adjaho2022,Kido2022,DAdamo2022,Pu2021,Mo2021,Zhao2019,Si2020,Siforthcoming,Kallus2018}. However, such studies focus on deterministic treatment rules; this study posits that the treatment assignment in the future population can be probabilistic. Furthermore, the result of this study applies to such individualized treatment assignment problems as long as it considers local asymptotics pointwise in covariates, and there is no constraint on feasible individualized treatment rules. 
\par
This study is also closely related to studies of efficient estimation of parameters expressed as $f(\theta_0)$ for some known function $f$. When $f$ is Hadamard differentiable, the plug-in estimator $f(\hat{\theta}_n)$ is the efficient estimator \parencite{Vaart1991b}. When $f$ is not Hadamard differentiable but Hadamard \emph{directionally} differentiable, then regular estimators, which are robust to small perturbations of dgps, do not exist \parencite{Hirano2012,Vaart1991}. Thus, the classical notion of efficiency based on best regularity breaks down for such nondifferentiable functions. However, the efficiency based on LAM still works \parencite{Song2014a,Fang2016,Ponomarev2022}. The relevant assume the (Hadamard) directional differentiability of $f$ and derive the LAM point estimators for $f(\theta)$. In such studies, the LAM point estimators also require some adjustment terms. In particular, \textcite{Ponomarev2022} shows that the LAM point estimators generally take the form of $f(\hat{\theta}_n + w_{\theta_0}^\dagger/\sqrt{n}) + v_{\theta_0}^\dagger/\sqrt{n}$. Accordingly, it is possible to construct the LAM point estimator for $\kappa(\theta_0)$. However, the optimal adjustment terms $w_{\theta_0}^*$ and $w_{\theta_0}^\dagger$ generally differ because the objective functions differ in the estimation problem and treatment assignment problem; the point estimation problem aims at minimizing the symmetric loss function such as mean squared error, while the treatment assignment problem aims at minimizing the asymmetric loss function such as welfare regret loss. The study will exemplify this difference using an example. The Hadamard directional differentiability has also been adopted in \textcite{Fang2019,Hong2018} to develop asymptotically valid inference procedures on $f(\theta_0)$. 
}
\paragraph{Structure of the Paper.}{
\cref{sec:stat-treat-decide-prob-under-pi} formally sets up the statistical treatment decision problem under the partial identification of the ATE. It also introduces the precise requirement for the identification region of the ATE with some motivating examples. \cref{sec:framework-of-local-asymptotic-analysis} discusses the framework of the local asymptotic analysis employed in this study. \cref{sec:asymptotically-minimax-optimal-statistical-treatment-rules} derives the LAM STRs, and \cref{sec:conclusion} concludes the paper. All proofs are relegated to \cref{sec:proofs-of-results-in-texts}.
}

\section{Setup}\label{sec:stat-treat-decide-prob-under-pi}
\subsection{Statistical Treatment Choice Problem}\label{sec:setup}
Suppose there is a binary treatment encoded as 0 and 1. Let $Y(a) \in \mathcal{Y}$ denote the potential outcome that would be realized if an individual was exposed to treatment $a \in \{0,1\}$. The set $\mathcal{Y}$ is assumed to be a bounded subset of $\mathbb{R}$, and let $y_L \coloneqq \inf \mathcal{Y}$ and $y_U \coloneqq \sup \mathcal{Y}$. Consider a social planner who wants to assign treatment 0 or 1 to each individual of a population of interest. The planner is allowed to implement probabilistic assignments independently of potential outcomes. Specifically, the planner can assign treatment 1 to an individual with probability $\delta \in [0,1]$, which will be called a \emph{treatment rule}. When the planner adopts a treatment rule $\delta$, the \emph{welfare} at the population of interest is given by 
\begin{equation}
    \expect \dbrack{Y(1)\cdot \delta + Y(0)\cdot (1-\delta)} = \expect[Y(1) - Y(0)]\cdot \delta + \expect[Y(0)],
    \label{eq:welfare}
\end{equation}
where the expectation is taken with respect to the target population's joint distribution of $(Y(1),Y(0))$. Suppose that higher welfare is more desirable. Then, if the planner knew the true marginal distributions of potential outcomes, she would search for the treatment rule that maximizes the welfare. The right-hand side of \cref{eq:welfare} shows that the important quantity for her optimization is the \emph{average treatment effect (ATE)}, $\expect \dbrack{Y(1) - Y(0)}$. The optimal treatment rules assign treatment 1 with probability one (zero) if and only if the ATE is positive (negative). Therefore, if she knew the signs of ATE, she could obtain an optimal treatment rule. Unfortunately, this rule is generally infeasible, as the ATE is unknown.
\par
Suppose instead a sample $\mathbf{Z}_n = (Z_1, \cdots, Z_n)$ is available to the planner. It is assumed that each observation $Z_i$ belongs to a complete separable metric space $\mathcal{Z}$ equipped with the Borel $\sigma$-algebra $\mathcal{B}(\mathcal{Z})$. Further, it is independent and identically distributed (i.i.d.) according to distribution $P_0 \in \mathcal{P}$, where $\mathcal{P}$ is a known class of Borel probability measures.\footnote{The completeness and separability of the space $\mathcal{Z}$ is put to ensure the separability of $L_2(P)$ space for every $P \in \mathcal{P}$. Indeed, every Borel probability measure $P$ is Radon \parencite[][Theorem~7.1.7]{Bogachev2007}. This then implies that every $P$ is separable \parencite[][Proposition~7.14.12]{Bogachev2007}. The separability of $L_2(P)$ follows from the separability of $P$ \Parencite[][Exercise~4.7.63]{Bogachev2007}.} Given this i.i.d. assumption, the sample is generated following $\mathbf{Z}_n \sim P_0^n$, which is the $n$-fold product of $P_0$. It is common to often refer each $P \in \mathcal{P}$ to a \emph{data generating process (dgp)} and also refer $P_0$ to the true dgp. If an infinite number of observations of $Z_i$ were available, the dgp $P_0$ would be learnable with arbitrary precision.
\par
This study assumes that the ATE cannot be point identified but \emph{partially identified}. Formally, even with knowledge of the true dgp $P_0$, it is only possible to conclude that the ATE belongs to the identification region $\tau(P_0)$, which is a subset of $[y_L - y_U,y_U-y_L]$ by the boundedness of $\mathcal{Y}$. Note that $\tau(P_0)$ may or may not contain positive and negative values, in which case the better treatment cannot be determined. The study places some restrictions on $\tau(P_0)$, which will be discussed in \cref{sec:identification-region-ATE}.
\par
Given the availability of a sample $\mathbf{Z}_n$, it is natural that the planner decides a treatment rule, utilizing the knowledge obtained from $\mathbf{Z}_n$. From this perspective, what the planner needs is the procedure that outputs a good treatment rule for each possible realization of $\mathbf{Z}_n$. The planner's problem can then be interpreted as one of the statistical decision problems \citep{Manski2004}. A statistical decision problem mainly comprises two components: a statistical decision rule and a risk. In the current context, a statistical decision rule is a mapping $\delta_n:\mathcal{Z}^n \to [0,1]$, which is called a \emph{statistical treatment rule (STR)}. Here, $\delta_n(\mathbf{z}_n) \in [0,1]$ is interpreted as the probability of assigning treatment 1 in the target population when the observed sample is $\mathbf{z}_n$. For the risk function, it is important to first specify a loss function that penalizes an arbitrary treatment rule $\delta$ when ATE is $\tau$. Following the literature, this study focuses on the welfare regret loss given by 
\begin{align*}
    L(\delta,\tau) 
    &\coloneqq \sup_{d \in [0,1]} \expect\dbrack{Y(1)\cdot d + Y(0)\cdot (1-d)} - \expect\dbrack{Y(1)\cdot \delta + Y(0)\cdot (1-\delta)}\\
    &= \tau(1-\delta)\cdot\ind\dbrace{\tau > 0} - \tau\delta\cdot \ind\dbrace{\tau \leq 0}. 
\end{align*}
This measures the loss the planner incurs by not knowing $\tau$. Given this loss function, for every $P \in \mathcal{P}$ and $\tau \in \tau(P)$, the risk of a STR $\delta_n$ is defined by $\expect_{P^n} \dbrack*{ L \dparen*{ \delta_n \dparen*{\mathbf{Z}_n}, \tau } }$.
\par
The definition clearly shows that the risk depends on the true dgp $P_0$ and ATE $\tau$. As neither $P_0$ nor $\tau$ are known with a sample of finite size, the planner would seek the STR that minimizes the risk for every possible $P \in \mathcal{P}$ and $\tau \in \tau(P)$. However, such STRs do not necessarily exist. To alleviate this difficulty, the literature on statistical decision theory has considered several criteria. Among those criteria, this study employs the \emph{minimax regret criterion}, as suggested by \textcite{Manski2004}. Specifically, the minimax regret STR minimizes 
\begin{equation}
    \adjustlimits \sup_{P \in \mathcal{P}} \sup_{\tau \in \tau(P)} \expect_{P^n} \dbrack*{ L \dparen*{ \delta_n \dparen*{\mathbf{Z}_n}, \tau } }.
    \label{eq:finite-sample-maximum-regret}
\end{equation}
From here, suppress the dependence of an STR $\delta_n(\mathbf{Z}_n)$ on the data $\mathbf{Z}_n$ to simplify the notation.
\subsection{Identification Region of ATE}\label{sec:identification-region-ATE}
Among the features of the identification region $\tau(P)$, the important quantities in the current context are its lower and upper bounds. Indeed, when one fixes $P \in \mathcal{P}$ and consider the inner maximization problem in \cref{eq:finite-sample-maximum-regret}, the direct calculation yields
\begin{align*}
    \sup_{\tau \in \tau(P)} \expect_{P^n} \dbrack*{ L \dparen*{ \delta_n, \tau } }
    = \max\dbrace*{ \dparen*{ \inf\tau(P) }^-\expect_{P^n}[\delta_n],~ \dparen*{ \sup\tau(P) }^+\dparen*{ 1 - \expect_{P^n}[\delta_n] } },
\end{align*}
where for any real value $x \in \mathbb{R}$, $(x)^+ = \max\{x,0\}$, and $(x)^- = \max\{-x,0\}$. Note that the boundedness of the outcome space implies that the bounds of $\tau(P)$ are bounded by a constant uniformly in $P$. The equation reveals that whatever the form of $\tau(P)$ is, the value of the inner maximization in \cref{eq:finite-sample-maximum-regret} depends only on $\inf\tau(P)$ and $\sup\tau(P)$. This study requires that the bounds, $\inf\tau(P)$ and $\sup\tau(P)$, are sufficiently smooth in $P$.
\subsubsection{Motivating Examples}\label{sec:motivating-examples}
To motivate the concrete smoothness conditions, the study first presents at some examples of partial identification of ATE.
\begin{example}
    \label{ex:identification-region-with-empirical-evidence-alone}
    Imagine an observational study in which treatments are not randomly assigned. Formally, suppose a researcher draws $n$ i.i.d. copies of $(Y(1),Y(0),D)$ from a population of interest and observes $Z_i = (Y_i,D_i),~i=1,\cdots,n$. Here, $D_i \in \{0,1\}$ denotes the treatment $i$-th subject received, and $Y_i$ denotes the $i$-th subject's observed outcome satisfying $Y_i = Y_i(D_i)$. Let $p = P(D = 1)$ and $\mu_d = \expect_P[Y|D=d]$ for $d \in \{0,1\}$. With this empirical evidence alone, the ATE is not point identified but partially identified. Specifically, \textcite{Manski1990} gives the sharp bounds as
    \begin{align}
        \inf\tau(P) = (\mu_1-y_U)p + (y_L - \mu_0)(1-p),
        \quad
        \sup\tau(P) = (\mu_1-y_L)p + (y_U - \mu_0)(1-p).
        \label{eq:example1-bound}
    \end{align}
    The width of $\tau(P)$ is always one irrespective of $P$. This implies that $\tau(P)$ always contains zero, whence the sign of ATE is never point identified.
    \qed
\end{example}
\begin{example}
    \label{ex:experimental-study-with-non-random-sampling}
    Imagine an experimental study in which an experimenter recruits participants non-randomly. More concretely, suppose a researcher has access to $n$ i.i.d. copies of $(Y(1),Y(0),D,S)$ from a population of interest and observes $Z_i = (S_iY_i,S_iD_i,S_i)$, where $Y_i$ and $D_i$ are defined in the same way as \cref{ex:identification-region-with-empirical-evidence-alone}, and $S_i \in \{0,1\}$ denotes the $i$-th subject's participation in the experiment ($S_i = 1$ if and only if subject $i$ participates in the experiment). Thus, for subjects with $S_i = 0$, the experimenter only observes that these subjects do not participate. However, among subjects with $S_i = 1$, the experimenter randomly assigns treatments with known probability, where one can assume $(Y(1),Y(0)) \perp D |S=1$. This setting has often been studied in the literature on external validity \parencite[e.g.,][]{Hotz2005,Cole2010}. Let $p = P(S_i = 1)$ and $\mu_d = \expect_P\dbrack{Y_i | D_i = d, S_i = 1}$ for $d \in \{0,1\}$. Without any further assumptions such as $(Y(1),Y(0))\perp S$, the ATE is partially identified. In particular, \textcite{Manski1989}'s result implies
    \begin{align}
        \inf\tau(P) = (\mu_1 - \mu_0)p - (y_U - y_L)(1-p),
        \quad
        \sup\tau(P) = (\mu_1 - \mu_0)p + (y_U - y_L)(1-p).
        \label{eq:example2-bound}
    \end{align}
    The width of $\tau(P)$ is $2(1-p)$, which depends on the dgp $P$. Moreover, the set $\tau(P)$ may or may not contain zero depending on $P$.
    \qed
\end{example}
\begin{example}
    \label{ex:identification-region-with-mean-independent-IV}
    In \cref{ex:identification-region-with-empirical-evidence-alone}, suppose the researcher has access to a binary instrument $V \in \{v_0,v_1\}$. More precisely, suppose the researcher draws $n$ i.i.d. copies of $(Y(1),Y(0),D,V)$ from a population of interest and observes $Z_i = (Y_i,D_i,V_i),~i=1,\cdots,n$, where $Y_i$ and $D_i$ are defined in the same way as \cref{ex:identification-region-with-empirical-evidence-alone}. Here, suppose the instrument satisfies the exclusion restriction, which rules out the possibility of the direct effect of the instrument on outcomes. Assume the instrument $V$ is mean independent of potential outcomes, that is, 
    \begin{equation*}
        \expect[Y(d)|V=v_j] = \expect[Y(d)]\quad\text{for every}\quad (d,j) \in \{0,1\} \times \{0,1\}.
        \label{eq:mean-independent-iv}
    \end{equation*}
    Let $p_j = P(D = 1|V = v_j)$ for $j \in \{0,1\}$ and $\mu_{d,j} = \expect_P[Y|D = d, V = v_j]$ for each $(d,j)$. Then, \textcite{Manski1990}'s result gives the sharp bounds as
    \begin{align}
        \begin{aligned}
        \inf\tau(P) 
        &= \max\dbrace*{
            \begin{aligned}
                & \mu_{1,1}p_1 + y_L(1-p_1) - \mu_{0,1}(1-p_1) - y_Up_1\\
                & \mu_{1,1}p_1 + y_L(1-p_1) - \mu_{0,0}(1-p_0) - y_Up_0\\
                & \mu_{1,0}p_0 + y_L(1-p_0) - \mu_{0,1}(1-p_1) - y_Up_1\\
                & \mu_{1,0}p_0 + y_L(1-p_0) - \mu_{0,0}(1-p_0) - y_Up_0
            \end{aligned}
        },\\
        \sup\tau(P)
        &= \min\dbrace*{
            \begin{aligned}
                & \mu_{1,1}p_1 + y_U(1-p_1) - \mu_{0,1}(1-p_1) - y_Lp_1\\
                & \mu_{1,1}p_1 + y_U(1-p_1) - \mu_{0,0}(1-p_0) - y_Lp_0\\
                & \mu_{1,0}p_0 + y_U(1-p_0) - \mu_{0,1}(1-p_1) - y_Lp_1\\
                & \mu_{1,0}p_0 + y_U(1-p_0) - \mu_{0,0}(1-p_0) - y_Lp_0
            \end{aligned} 
        }.
        \end{aligned}
        \label{eq:example3-bound}
    \end{align}
    In the binary outcome setting (i.e., $\mathcal{Y} = \{0,1\}$), \textcite{Balke1997} showed that the bounds given in \cref{eq:example3-bound} can be tightened with the assumption, $(Y(1),Y(0)) \perp V$, which is stronger than the mean-independence. Although the exact form of their bounds is not shown to save space, note that they still contain $\max$ and $\min$ functions. The set $\tau(P)$ may or may not contain zero depending on $P$.
    \qed
\end{example}
\subsubsection{Directional Differentiability of the Bounds}
The inspection of \cref{ex:identification-region-with-empirical-evidence-alone,ex:experimental-study-with-non-random-sampling,ex:identification-region-with-mean-independent-IV} reveals that the lower and upper bounds of $\tau(P)$ can often be written as a composition of maps. The relevant features of the data-generating process $P$ are first summarized by intermediate parameters $\theta$. Then, $\theta$ is mapped by certain functions, say $\tau_L$ and $\tau_U$, to the lower and upper bounds of $\tau(P)$. For the smoothness of $\tau_L$ and $\tau_U$, we must deal with the possible nondifferentiability of these functions. For instance, max and min functions, which appear in \cref{ex:identification-region-with-mean-independent-IV}, are not differentiable at some points. Rather, the max and min functions are directionally differentiable. Thus, to cover such instances, it is adequate to assume the directional differentiability of $\tau_L$ and $\tau_U$ in some strength. In the following, this study formalizes these features as the restrictions on the identification region.
\par
Suppose there exist $\theta:\mathcal{P} \to \Theta \subset \mathbb{B}$ for a Banach space $\mathbb{B}$ with norm $\normB{\cdot}$ and known functions $\tau_j:\Theta \to \mathbb{R}$ for $j \in \{L,U\}$ such that
\begin{equation}
    \inf\tau(P) = \tau_L\circ \theta (P) \quad \text{and} \quad \sup\tau(P) = \tau_U\circ\theta(P).
    \label{eq:form-of-bounds}
\end{equation}
One can take $\mathbb{B} = \mathbb{R}^k$ for an integer $k$ in most applications, but the formulation allows infinite-dimensional parameters, such as $\mathbb{B} = L_2(P)$. The intermediate parameter $\theta(P)$ is required to be pathwise differentiable, which will be discussed in \cref{sec:regularity-of-intermediate-parameter}. On the other hand, the functions $\tau_L(\theta)$ and $\tau_U(\theta)$ are assumed to be directionally differentiable in a certain strength, which is clarified in \cref{asum:bounds-of-identificaiton-region} below. In that statement, let $P_0 \in \mathcal{P}$ be the true data generating process and let $\theta_0 \coloneqq \theta(P_0)$ be the value of the intermediate parameter under true dgp. 
\begin{assumption}\label{asum:bounds-of-identificaiton-region}
    The functions $\tau_L:\Theta \to \mathbb{R}$ and $\tau_U:\Theta \to \mathbb{R}$ are continuous in $\theta$ and satisfies $\tau_L(\theta) \leq \tau_U(\theta)$ for all $\theta \in \Theta$ and $\tau_L(\theta_0) < \tau_U(\theta_0)$ at $\theta_0 \in \Theta$. Moreover, the following properties hold.
    \begin{assumptionenum}
        \item\label{assumptionenum:directional-differentiability} For each $j \in \{L,U\}$, the function $\tau_j:\Theta \to \mathbb{R}$ is \emph{Hadamard directionally differentiable} at $\theta_0 \in \Theta$. That is, there is a continuous map $\tau_{j,\theta_0}':\mathbb{B} \to \mathbb{R}$ such that
        \begin{equation*}
            \lim_{n\to\infty} \abs*{ \frac{\tau_j(\theta_0 + t_nb_n) - \tau_j(\theta_0)}{t_n} - \tau_{j,\theta_0}'(b) } = 0
        \end{equation*}
        for all sequences $\{b_n\} \subset \mathbb{B}$ and $\{t_n\} \subset \mathbb{R}_+$ such that $t_n \downarrow 0$, $b_n \to b \in \mathbb{B}$ as $n \to \infty$ and $\theta_0 + t_nb_n \in \Theta$ for all $n$.
        \item\label{assumptionenum:lipschitz-continuity} For each $j \in \{L,U\}$, the Hadamard directional derivative $\tau_{j,\theta_0}':\mathbb{B} \to \mathbb{R}$ is Lipschitz continuous; that is, there exists $C_{\tau_{j,\theta_0}'} > 0$ such that 
        \begin{equation*}
            \abs*{\tau_{j,\theta_0}'(b) - \tau_{j,\theta_0}'(\tilde{b})} \leq C_{\tau_{j,\theta_0}'}\normB*{b - \tilde{b}}
        \end{equation*}
        for all $b,\tilde{b} \in \mathbb{B}$.
        \item\label{assumptionenum:translation-equivariance} For every $v \in \mathbb{R}$, there exists $\tilde{v} \in \mathbb{B}$ such that $\tau_{j,\theta_0}'(b) + v = \tau_{j,\theta_0}'(b + \tilde{v})$ for all $b \in \mathbb{B}$ and $j \in \{L,U\}$.
    \end{assumptionenum}
\end{assumption}
\Cref{assumptionenum:directional-differentiability} gives the precise strength of directional differentiability needed in this paper. It is important to note that the Hadamard directional differentiability does not demand the linearity of the derivative $\tau_{j,\theta_0}':\mathbb{B} \to \mathbb{R}$. When it is linear, the function $\tau_{j}:\mathbb{B} \to \mathbb{R}$ is said to be \emph{Hadamard differentiable} at $\theta_0 \in \Theta$ \parencite[see][Proposition~2.1]{Fang2019}. Conventionally, Hadamard differentiability has often been utilized to derive the asymptotic distribution of statistics by applying the delta method. Concretely, if $\theta_n$ is an estimator for $\theta_0$ such that $\sqrt{n}(\theta_n - \theta_0) \leadsto \mathbb{G}$ and $\tau_j$ is Hadamard differentiable at $\theta_0$, then
\begin{equation}
    \sqrt{n}\dparen*{\tau_j(\theta_n) - \tau_j(\theta_0)} \leadsto \tau_{j,\theta_0}'(\mathbb{G}),
    \label{eq:delta-method}
\end{equation}
where $\leadsto$ indicates the weak convergence. Unfortunately, the Hadamard differentiability is too strong to cover max and min functions that often appear in the bounds of identification region as in \cref{ex:identification-region-with-mean-independent-IV}. Instead, this study only assumes the directional differentiability.
\par
Among the several notions of directional differentiability \parencite[see, e.g.,][]{Shapiro1990}, the Hadamard directional differentiability is the most suitable for this study. Especially, as emphasized in \textcite{Fang2019}, this directional differentiability allows for the direct generalization of the standard delta method: as long as $\tau_{j}(\theta)$ is Hadamard directionally differentiable at $\theta_0$, the weak convergence in \cref{eq:delta-method} continues to hold. Based on this generalized delta-method, \textcite{Fang2016,Ponomarev2022} developed the locally asymptotically minimax estimator for $\tau_j(\theta_0)$ itself. This study utilizes the generalized delta method to develop the LAM STR. In addition to the generalization of the delta method, the Hadamard directional differentiability is also useful, as it admits the chain rule \parencite[][Proposition~3.6]{Shapiro1990}.
\par
\cref{assumptionenum:lipschitz-continuity} requires that the Hadamard directional derivatives are Lipschitz continuous. \cref{assumptionenum:translation-equivariance} is imposed to simplify the analysis. When $\mathbb{B} = \mathbb{R}^k$, \cref{assumptionenum:translation-equivariance} is satisfied if $\tau_{j,\theta_0}'$ is translation equivariant. In this case, $\tau_{j,\theta_0}'(b + v\mathbf{1}) = \tau_{j,\theta_0}'(b) + v$ for every $b$ and $v$, where $\mathbf{1} = (1,\cdots,1) \in \mathbb{R}^k$. Particularly, $\max$ and $\min$ functions, which often appear in the bounds of partial identification, satisfy the translation equivariance.
\begin{examplecontinued}{ex:identification-region-with-empirical-evidence-alone}\label{ex:example1-second}
    In this example, let $\theta(P) = (\mu_1,\mu_0,p)^\intercal \in \Theta = [y_L,y_U]^2\times[0,1] \subset \mathbb{R}^3$, and let $\tau_L(\theta)$ and $\tau_U(\theta)$ be specified in the right-hand sides of \cref{eq:example1-bound}. As long as $\theta_0$ is an interior point of $\Theta$, \cref{asum:bounds-of-identificaiton-region} is satisfied. In particular, $\tau_L(\theta)$ and $\tau_U(\theta)$ are Hadamard differentiable, and the corresponding derivative can be calculated from the usual partial derivatives. For instance,
    \begin{equation*}
        \tau_{L,\theta_0}'(b) = \dparen*{\frac{\partial \tau_L}{\partial \theta}\Bigr|_{\theta = \theta_0}}^\intercal b
    \end{equation*}
    for $b \in \mathbb{R}^3$.
\end{examplecontinued}
\begin{examplecontinued}
{ex:experimental-study-with-non-random-sampling}
    Let $\theta(P) = (\mu_1,\mu_0,p)^\intercal \in \Theta = [y_L,y_U]^2\times[0,1] \subset \mathbb{R}^3$, and let $\tau_L(\theta)$ and $\tau_U(\theta)$ be given in the right-hand side of \cref{eq:example2-bound}. Again, \cref{asum:bounds-of-identificaiton-region} is satisfied as long as $\theta_0$ is an interior point. The functions $\tau_L(\theta)$ and $\tau_U(\theta)$ are Hadamard differentiable with the derivatives calculated from the partial derivative as in \cref{ex:identification-region-with-empirical-evidence-alone}.
\end{examplecontinued}
\begin{examplecontinued}{ex:identification-region-with-mean-independent-IV}
    Let $\theta(P) = (\mu_{1,1},\mu_{1,0},\mu_{0,1},\mu_{0,0},p_1,p_0)^\intercal \in \Theta = [y_L,y_U]^4\times[0,1]^2 \subset \mathbb{R}^6$. Moreover, specify $\tau_L(\theta)$ and $\tau_U(\theta)$ by the right-hand sides of \cref{eq:example3-bound}. Then, \cref{asum:bounds-of-identificaiton-region} holds as long as $\theta_0$ is in the interior of $\Theta$. To see the Hadamard directionally differentiability of $\tau_L(\theta)$, it is convenient to express the function as the composition of the max function and
    \begin{equation*}
        \psi_L(\theta) 
        = (\psi_{L,1}(\theta),\cdots,\psi_{L,4}(\theta))^\intercal 
        = \begin{pmatrix}
            \mu_{1,1}p_1 + y_L(1-p_1) - \mu_{0,1}(1-p_1) - y_Up_1\\
            \mu_{1,1}p_1 + y_L(1-p_1) - \mu_{0,0}(1-p_0) - y_Up_0\\
            \mu_{1,0}p_0 + y_L(1-p_0) - \mu_{0,1}(1-p_1) - y_Up_1\\
            \mu_{1,0}p_0 + y_L(1-p_0) - \mu_{0,0}(1-p_0) - y_Up_0
        \end{pmatrix}
        \in\mathbb{R}^4.
    \end{equation*}
    As in \cref{ex:identification-region-with-empirical-evidence-alone,ex:experimental-study-with-non-random-sampling}, one can easily observe that $\psi_L(\theta)$ is Hadamard differentiable with the derivative given by $\psi_{L,\theta_0}'(b) = \dparen*{\frac{\partial \psi_L}{\partial \theta^\intercal }|_{\theta = \theta_0}}^\intercal b$ for $b = (b_1,\cdots,b_6) \in \mathbb{R}^6$. Moreover, the $\max$ function is Hadamard directionally differentiable; the directional derivative of the $\max$ function at $\psi_L(\theta_0)$ is given by
    \begin{equation*}
        \max_{j \in B_L(\theta_0)} \psi_j \quad\text{for every}\quad (\psi_1,\dots,\psi_4) \in \mathbb{R}^4,
    \end{equation*}
    where $B_L(\theta_0) = \argmax_{j = 1,\cdots,4} \psi_{L,j}(\theta_0)$. By the chain rule for the Hadamard directionally differentiable maps, $\tau_L(\theta)$ is Hadamard directionally differentiable with derivative 
    \begin{equation*}
        \tau_{L,\theta_0}'(b) = \max_{j \in B_L(\theta_0)} \psi_{L,\theta_0}'(b)
    \end{equation*}
    Similarly, $\tau_U(\theta)$ is also Hadamard directionally differentiable.
\end{examplecontinued}

\subsection{Optimal Treatment Rule with Knowledge of True DGP}\label{sec:optimal-treatment-rule-with-knowledge-of-true-dgp}
If the social planner knew the true dgp $P_0$, there would be no need to consider statistical uncertainty, and hence, they would solve $\inf_{\delta \in [0,1]}\sup_{\tau \in \tau(P_0)} L(\delta,\tau)$. As has been observed in \textcite{Manski2007,Manski2009}, the minimax regret value is 
\begin{equation}
    \adjustlimits{\inf}_{\delta \in [0,1]}{\sup}_{\tau \in \tau(P_0)} L(\delta,\tau) = \frac{\tau_U^+(\theta_0)\tau_L^-(\theta_0)}{\tau_U^+(\theta_0) + \tau_L^-(\theta_0)},
    \label{eq:minimax-value-optimal-treatment-rule}
\end{equation}
and the optimal treatment rule is given by \cref{eq:optimal-treatment-rule}. Depending on the true intermediate parameter $\theta_0$, the future assignment given the optimal treatment rule may (not) be deterministic, and the corresponding minimax value may (not) be greater than zero. Following \textcite{Manski2009}, this study introduces a concept of error because of selecting an inferior treatment to interpret the regret value and rule. Specifically, refer to choosing treatment 0 when the ATE $\tau$ is positive as a Type 0 error. Conversely, refer to choosing treatment 1 when $\tau$ is negative as a Type 1 error. Further, recall that the welfare regret loss is given by $\tau(1-\delta)\cdot \ind\dbrace{\tau > 0} - \tau\delta\cdot \ind\dbrace{\tau \leq 0}$. The regret loss measures the loss from Type 0 or 1 errors: $\tau(1-\delta)$ corresponds to the loss from the Type 0 error, while $-\tau\delta$ corresponds to the loss given the Type 1 error.
\par
First, if $\tau_L(\theta_0) < \tau_U(\theta_0) \leq 0$, then $\kappa(\theta_0) = 0$, and hence, treatment 1 is assigned with probability 0. In this case, the ATE is certainly non-positive, and what matters is the Type 1 error. From the knowledge of the identification region $\tau(P_0)$, treatment 0 is never worse than treatment 1. Accordingly, the superior treatment can be chosen for sure, and the loss from the Type 1 error can be made zero. Similarly, if $0 \leq \tau_L(\theta_0) < \tau_U(\theta_0)$, then $\kappa(\theta_0) = 1$, and the minimax regret is again zero. This case is also interpreted similarly. Finally, if $\tau_L(\theta_0) < 0 < \tau_U(\theta_0)$, $\kappa(\theta_0) = \tau_U(\theta_0)/\dparen{\tau_U(\theta_0) - \tau_L(\theta_0)} \in (0,1)$, and the minimax regret value is greater than zero. In this case, positive and negative ATEs are possible, and hence, Type 0 and Type 1 errors matter. In particular, the worst-case loss from Type 0 error is $\tau_U(\theta_0)(1-\delta)$, and the worst-case loss from Type 1 error is $-\tau_L(\theta_0)\delta$. As $\delta$ increases from 0 to 1, the former loss decreases from $\tau_U(\theta_0)$ to $0$, while the latter loss increases from $0$ to $-\tau_L(\theta_0)$. As the welfare regret loss equally weighs the losses from both errors, the optimal treatment rule balances the worst-case losses from Type 0 and 1 errors; that is, $\tau_U(\theta_0)(1-\kappa(\theta_0)) = -\tau_L(\theta_0)\kappa(\theta_0)$.
\par
Under \cref{asum:bounds-of-identificaiton-region}, the optimal treatment rule $\kappa(\theta)$ inherits the properties of $\tau_L$ and $\tau_U$.
\begin{lemma}\label{lem:kappa-directionally-differentiable}
    If \cref{asum:bounds-of-identificaiton-region} is in force, the following hold.
    \begin{lemmaenum}
        \item\label{lemmaenum:directional-differentiability} The optimal treatment rule $\kappa: \Theta \to \mathbb{R}$ is Hadamard directionally differentiable at $\theta_0$ with the derivative given by
        \begin{equation}
            \kappa_{\theta_0}'(b) \coloneqq \frac{\tau_L^-(\theta_0)\cdot m_{\tau_U(\theta_0)}'\dparen*{\tau_{U,\theta_0}'(b)} - \tau_U^+(\theta_0)\cdot  m_{-\tau_L(\theta_0)}'\dparen*{-\tau_{L,\theta_0}'(b)}}{\dparen*{\tau_U^+(\theta_0) + \tau_L^-(\theta_0)}^2}
            \label{eq:kappa-directional-derivative}
        \end{equation}
         for all $b \in \mathbb{B}$. Here, $m_y':\mathbb{R}\to\mathbb{R}$ given by
         \begin{equation}
            m_{y}'(x) \coloneqq x \cdot \ind\dbrace{y > 0} + \max\dbrace{x,0}\cdot \ind\dbrace{y = 0} + 0 \cdot \ind\dbrace{y < 0}
            \label{eq:max-directional-derivative}
        \end{equation}
         denotes the Hadamard directional derivative of $x \mapsto \max\{x,0\}$ at $y \in \mathbb{R}$.
        \item\label{lemmaenum:lipschitz-continuity} The Hadamard directional derivative $\kappa_{\theta_0}':\mathbb{B} \to \mathbb{R}$ is Lipschitz continuous.
        \item\label{lemmaenum:translation-equivariance} For every $v \in \mathbb{R}$, there exists $\tilde{v} \in \mathbb{B}$ such that $\kappa_{\theta_0}'(b + \tilde{v}) = \kappa_{\theta_0}'(b) + v$ for all $b \in \mathbb{B}$.
    \end{lemmaenum}
\end{lemma}
Importantly, $\kappa(\theta)$ is also Hadamard directionally differentiable at $\theta_0$. The Hadamard directional derivative is derived by applying the chain rule for Hadamard directionally differentiable maps. Note that $\kappa(\theta)$ may not be differentiable at $\theta_0$ even if $\tau_L(\theta)$ and $\tau_U(\theta)$ are Hadamard differentiable at $\theta_0$. For instance, even if $\tau_U(\theta)$ is Hadamard differentiable at $\theta_0$, $\max\{\tau_U(\theta_0),0\}$ is not Hadamard differentiable when $\tau_U(\theta_0) = 0$; so is $\kappa(\theta_0)$. See also \pageref{ex:example2-third}
\begin{examplecontinued}
{ex:identification-region-with-empirical-evidence-alone}
\label{ex:example1-third}
    This example always have $\tau_L(\theta) \leq 0$ and $\tau_U(\theta) \geq 0$ for all $\theta \in \Theta$. Hence, the optimal treatment rule can be simplified as $\kappa(\theta_0) = \tau_U(\theta_0)/(y_U - y_L)$. As is clear from the simplified expression, $\kappa(\theta_0)$ is Hadamard differentiable with derivative $\kappa_{\theta_0}'(b) = \tau_{U,\theta_0}'(b)/(y_U - y_L)$.
\end{examplecontinued}
\begin{examplecontinued}
{ex:experimental-study-with-non-random-sampling}
\label{ex:example2-third}
    We separate the arguments into cases. If $\tau_L(\theta_0) < 0 < \tau_U(\theta_0)$, $\tau_U(\theta_0)$ and $\tau_L(\theta_0)$ are Hadamard differentiable at $\theta_0$, and so are $\tau_U^+(\theta_0)$ and $\tau_L^-(\theta_0)$. Hence, $\kappa(\theta_0)$ is also Hadamard differentiable, and $\kappa_{\theta_0}'(b)$ can be obtained from \cref{eq:kappa-directional-derivative,eq:max-directional-derivative}. If $\tau_L(\theta_0) = 0 < \tau_U(\theta_0)$, $\tau_L^-(\theta_0)$ is only Hadamard directionally differentiable at $\theta_0$, and so is $\kappa(\theta_0)$. In particular, \cref{eq:kappa-directional-derivative,eq:max-directional-derivative} imply 
    \begin{equation*}
        \kappa_{\theta_0}'(b) = \frac{1}{\tau_U(\theta_0)}\min\dbrace*{ \tau_{L,\theta_0}'(b),~0 }.
    \end{equation*}
    Finally, if $0 < \tau_L(\theta_0) < \tau_U(\theta_0)$, $\kappa_{\theta_0}'(b) = 0$. For other cases, it is possible to calculate $\kappa_{\theta_0}'(b)$ similarly.
\end{examplecontinued}

\section{Framework of Local Asymptotic Analysis}
\label{sec:framework-of-local-asymptotic-analysis}
It is often challenging to exactly solve the finite-sample problem \cref{eq:finite-sample-maximum-regret} for general data-generating processes. Instead, in the same spirit of \textcite{Hirano2009}, this study conducts the local asymptotic analysis and then derives asymptotically minimax STRs. This section discusses the framework of the local asymptotic analysis. Specifically, \cref{sec:local-statistical-experiments} setups local statistical experiments that are tailored to i.i.d. observations. \cref{sec:regularity-of-intermediate-parameter} gives the precise smoothness required for the intermediate parameter. Finally, \cref{sec:local-asymptotic-minimaxity} formally defines the notion of asymptotic minimax optimality employed in this study.
\subsection{Local Statistical Experiments}
\label{sec:local-statistical-experiments}
The study considers a sequence of \emph{local} statistical experiments in the following asymptotic analysis. Although the original statistical experiment involves all of the possible dgps in $\mathcal{P}$, the fixed true dgp $P_0$ is learnable with arbitrary precision as the sample size approaches infinity. Thus, from the perspective of asymptotic analysis, the asymptotic analysis, considering dgps that are too distant from $P_0$ may be meaningless. Instead, the local asymptotic analysis focuses only on dgps that are statistically difficult to distinguish from the true dgp $P_0$ to approximate the finite-sample situation well. The next concrete construction of the sequence of local experiments closely follows \textcite{Hirano2009,Vaart1991a}.
\par
Given the fixed true dgp $P_0 \in \mathcal{P}$, let $\mathcal{P}(P_0)$ be a collection of paths $(0, \epsilon) \ni t \mapsto P_{t} \in \mathcal{P}$ such that 
\begin{equation}
    \int \dbrack*{ \frac{dP_{t}^{1/2} - dP_0^{1/2}}{t} - \frac{1}{2}hdP_0^{1/2} }^2 \to 0 \quad \text{as} \quad t \downarrow 0,
    \label{eq:DQM}
\end{equation}
for some measurable function $h: \mathcal{Z} \to \mathbb{R}$, which is often called the \emph{score function}. It is possible to view each path $t \mapsto P_t$ as a one-dimensional parametric submodel $\dbrace{P_t : t \in (0,\epsilon)} \subset \mathcal{P}$, in which $t$ is treated as the only parameter. Gathering all score functions associated with paths in $\mathcal{P}(P_0)$ induces the \emph{tangent set} $T(P_0)$, which is the collection of score functions. It is known that $\int h dP_0 = 0$ and $\int h^2 dP_0 < \infty$ for any score function satisfying \cref{eq:DQM}. Thus, regarding $P_0$-almost surely equal score functions as an equivalence class, $T(P_0)$ can be viewed as a subset of $L_2(P_0)$, which is the separable Hilbert space with the inner product $\dangle{h,g} \coloneqq \int hg dP_0$ and norm $\norm{h}_{2,P_0} \coloneqq \sqrt{\dangle{h,h}}$. To simplify the exposition, this study posits that the tangent set is closed with respect to addition and scalar multiplication.\footnote{The result of this paper can be extended to the case where the tangent set is a convex cone with some appropriate modifications of proofs. See, e.g., \textcite{Vaart1989,Ponomarev2022}.} 
\begin{assumption}\label{asum:tangent-set}
    The tangent set $T(P_0)$ is a linear subspace of the separable Hilbert space $L_2(P_0)$. 
\end{assumption}
We introduce some notations to be used hereafter. For each path $t\mapsto P_t$ with score function $h \in T(P_0)$, let $P_{1/\sqrt{n},h} \in \mathcal{P}$ denote the value of the path evaluated at $t = 1/\sqrt{n}$. Moreover, let $\mathbb{P}_{n,h} = P_{1/\sqrt{n},h}^n$ and let $\mathbb{P}_{n,0} = P_0^n$. Correspondingly, the expectation with respect to $\mathbb{P}_{n,h}$ is denoted by $\expect_{n,h}[\cdot]$. Further, for every path with score function $h$, write $\theta_n(h) = \theta(P_{1/\sqrt{n},h})$.
\par
An important consequence of the property \cref{eq:DQM} is the local asymptotic normality of the statistical experiments $\mathcal{E}_n = (\mathcal{Z}^n,\mathcal{B}(\mathcal{Z}^n),\mathbb{P}_{n,h}:h \in T(P_0))$. Specifically, for each $h \in T(P_0)$, the log-likelihood ratio admits a linear expansion such that
\begin{equation}
    \log\diff{\mathbb{P}_{n,h}}{\mathbb{P}_{n,0}} = \frac{1}{\sqrt{n}}\sum_{i=1}^n h(Z_i) - \frac{1}{2}\norm{h}_{2,P_0}^2 + o_{p}(1),
    \label{eq:loglikelihood-expansion}
\end{equation}
where $\Delta_{n,h} = \sum_{i=1}^n h(Z_i)/\sqrt{n}$ converges weakly to $N(0,\norm{h}_{2,P_0}^2)$ under $\mathbb{P}_{n,0}$. This local asymptotic normality further implies the convergence of experiments $\mathcal{E}_n$ to a certain limit experiment $\mathcal{E}$. Under \cref{asum:tangent-set}, a useful characterization of $\mathcal{E}$ is available. \cref{asum:tangent-set} implies that $\overline{T(P_0)}$ itself is one of separable Hilbert spaces. Then, there exists a complete orthonormal basis $\{h_1,h_2,\cdots\} \subset T(P_0)$, and every $h \in T(P_0)$ can be written as $h = \sum_{j=1}^\infty \dangle{h,h_j}h_j$. In this case, the following convergence of experiments holds:
\begin{equation*}
    \mathcal{E}_n \to \mathcal{E} = \dparen*{\mathbb{R}^\infty, \mathcal{B}(\mathbb{R}^\infty), N_\infty((\dangle{h,h_j}),\mathbf{1}): h \in T(P_0)}.
\end{equation*}
The latter experiment corresponds to observe a single draw $\Xi = (\Xi_1,\Xi_2,\cdots) \in \mathbb{R}^\infty$ such that $\Xi_1,\Xi_2,\cdots$ are independent and $\Xi_j \sim N(\dangle{h,h_j},1)$ for each $j \in \mathbb{N}$. For ease of notation, write $\mathbb{P}_h$ for $N_\infty((\dangle{h,h_j}),\mathbf{1})$ for every $h \in T(P_0)$. Accordingly, the expectation with respect to $\mathbb{P}_h$ is denoted by $\expect_{h}[\cdot]$. As one expects, this limit experiment $\mathcal{E}$ is more tractable relative to the original experiment $\mathcal{E}_n$. Finally, the convergence of experiments culminates in the following asymptotic representation theorem.
\begin{lemma}[{\cite[Theorem~3.1]{Vaart1991a}}]\label{lem:asymptotic-representation-theorem}
    Suppose a sequence of experiments $\mathcal{E}_n$ converges to a dominated experiment $\mathcal{E}$. Let $T_n:\mathcal{Z}^n \to \mathbb{D}$ be a statistic with values in a metric space $\mathbb{D}$. Suppose the sequence $T_n \overset{h}{\leadsto} Q_h$ for every $h \in T(P_0)$ and that there exists a complete separable subset $\mathbb{D}_0 \subset \mathbb{D}$ such that $Q_h(\mathbb{D}_0) = 1$ for every $h$. Then there exists a randomized statistic $T$ in $\mathcal{E}$ such that $T_n \overset{h}{\leadsto} T$ for every $h$.
\end{lemma}
In the statement of \cref{lem:asymptotic-representation-theorem}, the randomized statistic $T$ refers to a measurable map $T:\mathbb{R}^\infty \times [0,1] \ni (\Xi,U) \mapsto T(\Xi,U) \in \mathbb{D}$, where $\Xi \sim \mathbb{P}_h$ and $U \sim \mathrm{Unif}([0,1])$ independently of $\Xi$. Moreover, $\overset{h}{\leadsto}$ means the weak convergence under $\mathbb{P}_{n,h}$ for $h \in T(P_0)$. \cref{lem:asymptotic-representation-theorem} allows for guessing a statistic with a certain asymptotic property by the following procedures; (i) focus on the limit experiment $\mathcal{E}$; (ii) obtain the statistic $T$ with the desired property in $\mathcal{E}$; (iii) construct a sequence of statistics $T_n$ that converges weakly to $T$. 
\subsection{Regularity of Intermediate Parameter}
\label{sec:regularity-of-intermediate-parameter}
For the intermediate parameter $\theta:\mathcal{P} \to \mathbb{B}$ introduced in \cref{sec:identification-region-ATE}, assume its pathwise differentiability.
\begin{assumption}\label{asum:theta-differentiable}
    The parameter $\theta:\mathcal{P} \to \mathbb{B}$ is pathwise differentiable at $P_0$ relative to the tangent set $T(P_0)$. That is, there exists a continuous linear map $\Dot{\theta}_0:T(P_0) \to \mathbb{B}$ such that
    \begin{equation*}
        \frac{\theta(P_{t}) - \theta(P_0)}{t} \to \Dot{\theta}_0(h) \quad\text{as}\quad t \downarrow 0
    \end{equation*}
    for each path $t \mapsto P_t$ in $\mathcal{P}(P_0)$ with score function $h \in T(P_0)$.
\end{assumption}
The pathwise differentiability especially implies $\sqrt{n}\dparen{\theta_n(h) - \theta_0} \to \Dot{\theta}_0(h)$ as $n \to \infty$. The pathwise derivative $\Dot{\theta}_0(h)$ must be linear maps. When $\mathbb{B}$ is the Euclidean space, the Riesz representation theorem gives a useful characterization of the pathwise derivative. Specifically, as $\overline{T(P_0)}$ is a Hilbert space, the theorem implies that there exists $\tilde{\theta}_0 = (\tilde{\theta}_{0,1},\cdots,\tilde{\theta}_{0,k})^\top \in \overline{T(P_0)}^k$ such that $\Dot{\theta}_{0,j}(h) = \dangle{\tilde{\theta}_{0,j},h}$ for all $h \in \overline{T(P_0)}$, where $\Dot{\theta}_{0,j}$ denotes the $j$-th element of $\Dot{\theta}_0$. The function $\tilde{\theta}_0$, thus constructed, is often called \emph{efficient influence function}. More generally, if $\mathbb{B}$ is an arbitrary Banach space, an adjoint map can characterize the pathwise derivative. In particular, there exists $\Dot{\theta}_0^*:\mathbb{B}^* \to \overline{T(P_0)}$ determined by $\dangle{\Dot{\theta}_0^*b^*, h} = b^*(\Dot{\theta}_0(h))$ for every $h \in T(P_0)$ and $b^* \in \mathbb{B}^*$.
\par
The pathwise differentiability assumption is often employed in the study of efficient estimators \parencite[see, e.g.,][]{Vaart1998,Vaart1996}. Among the results obtained in the study, one of the most important is the convolution theorem that characterizes the limit law of regular estimators. An estimator $\theta_n:\mathcal{Z}^n \to \Theta \subset \mathbb{B}$ is said to be \emph{regular} at $P_0$ for estimating $\theta_0$ if 
\begin{equation*}
    \sqrt{n}\dparen*{\theta_n - \theta_n(h)} \overset{h}{\leadsto} G \quad\text{for every}\quad h \in T(P_0),
\end{equation*}
where $G$ is a tight random element in $\mathbb{B}$ such that its law does not depend on $h$. In words, the limit law of the standardized statistics $\sqrt{n}\dparen{\theta_n - \theta_n(h)}$ are not affected by small perturbations of the dgp. Given the pathwise differentiability of $\theta(P)$, the convolution theorem \parencite[e.g.,][Theorem~3.11.2]{Vaart1996} ensures that the limit law of any regular estimator $\theta_n$ can be expressed as in
\begin{equation*}
    \mathcal{L}(G) = \mathcal{L}(\mathbb{G}_0 + W),
\end{equation*}
where $\mathbb{G}_0$ and $W$ are independent, tight, Borel measurable random elements in $\mathbb{B}$ such that $b^*\mathbb{G}_0 \sim N(0,\norm{\Dot{\theta}_0^*b^*}_{2,P_0}^2)$ for $b^* \in \mathbb{B}^*$, and the support of $\mathbb{G}_0$ is the closure of $\Dot{\theta}_0(T(P_0))$. The random element $W$ can be interpreted as an independent noise. Since adding an independent noise solely increases the variance, the noise injection is usually undesirable. Thus, the \emph{best} regular estimator $\hat{\theta}_n:\mathcal{Z}^n \to \Theta$ is defined as the regular estimator whose limit law equals $\mathbb{G}_0$. That is, the limit law of best regular estimators is maximally concentrated around zero. Hence, the best regularity is often considered one of the desirable properties for an estimator. In parametric models, typical examples of the best regular estimator include the maximum likelihood estimator and Bayesian posterior mean \parencite[see, e.g.,][]{Ibragimov1981}.
\par
The availability of a best regular estimator $\hat{\theta}_n$ for $\theta_0$ immediately provides a best regular estimator for $f(\theta_0)$ as far as $f: \Theta \to \mathbb{R}$ is an arbitrary Hadamard differentiable map. In this case, the plug-in estimator (i.e., $f(\hat{\theta}_n)$) is a best regular estimator for $f(\theta_0)$. Indeed, it is easy to show that the limit law of any regular estimator $f_n$ for $f(\theta_0)$ can be written as 
\begin{equation*}
    \sqrt{n}\dparen*{ f_n - f(\theta_n(h)) } \overset{h}{\leadsto} f_{\theta_0}'(\mathbb{G}_0) + \tilde{W} \quad \text{for every} \quad h \in T(P_0)
\end{equation*}
\parencite[see, e.g.,][Theorem~3.11.2]{Vaart1996}. Again, $\tilde{W}$ is a noise term independent of $\mathbb{G}_0$, and hence the limit law of best regular estimator for $f(\theta_0)$ is $f_{\theta_0}'(\mathbb{G}_0)$. Using the standard delta method, $f(\hat{\theta}_n)$ is the best regular. 
\par
When $f$ is not Hadamard differentiable but directionally differentiable, the plug-in estimator is no longer the best regular estimator. On the contrary, if $f$ is not Hadamard differentiable, there is no best regular estimator for $f(\theta_0)$ \parencite{Hirano2012,Vaart1991}. Therefore, in the current setting, there exists no estimators for $\tau_L(\theta_0)$, $\tau_U(\theta_0)$, and $\kappa(\theta_0)$ that are efficient in terms of best regularity.
\begin{examplecontinued}{ex:identification-region-with-empirical-evidence-alone}
\label{ex:example1-fourth}
    In this example, the tangent set is equal to $L_2^0(P_0) \coloneqq \dbrace{ h \in L_2(P_0) : \int hdP_0 = 0 }$, and $\theta(P_t)$ is pathwise differentiable at $P_0$ relative to this tangent set. The best regular estimator for $\theta_0$ is
    \begin{equation*}
        \hat{\theta}_n
        = \dparen*{ \hat{\mu}_{1,n},\hat{\mu}_{0,n},\hat{p}_n }
        = \dparen*{ \frac{\sum_{i=1}^n Y_iD_i}{\sum_{i=1}^n D_i},~\frac{\sum_{i=1}^nY_i(1-D_i)}{\sum_{i=1}^n (1-D_i)},~\frac{\sum_{i=1}^n D_i}{n} },
    \end{equation*}
    and the limit law of $\sqrt{n}\dparen{\hat{\theta}_n - \theta_0}$ is the mean-zero normal distribution whose covariance matrix is the diagonal matrix with its diagonal elements given by $\sigma_{Y|D=1}^2/p_0$, $\sigma_{Y|D=0}^2/(1-p_0)$, and $p_0(1-p_0)$, where $\sigma_{Y|D=d}^2 = \expect_{P_0}[\dparen{Y - \mu_{d,0}}^2|D=d]$ for $d \in \{0,1\}$.
\end{examplecontinued}
\begin{examplecontinued}{ex:experimental-study-with-non-random-sampling}
\label{ex:example2-fourth}
    As in \cref{ex:identification-region-with-empirical-evidence-alone}, the tangent set is $L_2^0(P_0)$, and $\theta(P_t)$ is pathwise differentiable relative to this tangent set.
    The best regular estimator is 
    \begin{equation*}
        \hat{\theta}_n = (\hat{\mu}_{1,n},\hat{\mu}_{0,n},\hat{p}_n) = \dparen*{ \frac{\sum_{i=1}^n Y_iD_iS_i}{\sum_{i=1}^nD_iS_i},~\frac{\sum_{i=1}^nY_i(1-D_i)S_i}{\sum_{i=1}^n (1-D_i)S_i},\frac{\sum_{i=1}^n S_i}{n} },
    \end{equation*}
    and $\mathbb{G}_0$ is distributed according to the mean-zero normal distribution whose covariance matrix is the diagonal matrix with its diagonal elements given by $\sigma_{Y|D=1,S=1}^2/(\pi_0p_0)$, $\sigma_{Y|D=0,S=1}^2/((1-\pi_0)p_0)$, and $p_0(1-p_0)$, where $\sigma_{Y|D=d,S=1}^2 = \expect_{P_0}[(Y-\mu_{d,0})^2|D=d,S=1]$ for $d \in \{0,1\}$ and $\pi_0 = P_0(D=1|S=1)$.
\end{examplecontinued}
\subsection{Local Asymptotic Minimaxity}
\label{sec:local-asymptotic-minimaxity}
Here, we precisely define our asymptotic minimax criterion (i.e., the LAM). First, define the criterion for a general risk function. For an STR $\delta_n$ and a dgp $P \in \mathcal{P}$, let $R_n(\delta_n,P) \geq 0$ be the \emph{risk function} that quantifies the loss from adopting the STR $\delta_n$ when the dgp is $P$. The (local) asymptotic maximum risk of $\delta_n$ with respect to $R_n$ is defined by
\begin{equation}
    \adjustlimits{\sup}_{I \subset T(P_0)}{\liminf}_{n \to \infty}\sup_{h \in I} R_n(\delta_n,P_{1/\sqrt{n},h}),
    \label{eq:local-asymptotic-maximum-risk}
\end{equation}
where the first supremum is taken with respect to arbitrary finite subsets $I$ of $T(P_0)$. This notion of local asymptotic maximum risk has been adopted in \textcite{Vaart1996,Vaart1998,Vaart1991a,Hirano2009,Fang2016,Ponomarev2022}. For the motivation of this notion, refer to \textcite{Fang2016}. The STR $\delta_n$ is said to be \emph{locally asymptotically minimax (LAM)} with respect to the risk function $R_n$ if it minimizes \cref{eq:local-asymptotic-maximum-risk}. Note that the LAM property depends on the choice of the risk function.
\par
Now, we specify the risk function employed. For an STR $\delta_n$ and $h \in T(P_0)$, its \emph{identifiable maximum risk (IMR)} is given by\footnote{This name originally comes from \textcite{Song2014a}.}
\begin{equation}
    \imr_n(\delta_n,h) 
    \coloneqq \sup_{\tau \in \tau\dparen*{P_{1/\sqrt{n},h}}} \expect_{n,h*}\dbrack*{ L \dparen*{ \delta_n,\tau } }.
    \label{eq:identifiable-maximum-risk}
\end{equation}
where $\expect_{n,h*}[\cdot]$ denotes the inner integral with respect to $\mathbb{P}_{n,h}$. The right-hand side of \cref{eq:identifiable-maximum-risk} measures the worst-case welfare regret incurred by not knowing the ATE even with the knowledge of the dgp $P_{1/\sqrt{n},h}$, and it has the same structure with that of \cref{eq:finite-sample-maximum-regret}. Note that the IMR can be expressed as 
\begin{equation}
    \imr_n(\delta_n,h) 
    = \max\dbrace[\Big]{ \tau_L^-\dparen*{ \theta_n(h) }\expect_{n,h*}[\delta_n],~ \tau_U^+\dparen*{ \theta_n(h) } \dparen*{1 - \expect_{n,h*}[\delta_n] } }.
    \label{eq:identifiable-maximum-risk-expansion}
\end{equation}
The LAM based on the IMR is unsatisfactory. It is easy to show that any consistent estimator for $\kappa(\theta_0)$ is LAM with respect to the IMR. This includes not only the \emph{plug-in STR}; that is, 
\begin{equation}
    \delta_n^{\plugin} \coloneqq \kappa\dparen{\hat{\theta}_n},
    \label{eq:plug-in-str}
\end{equation}
but also the STR, $\kappa(\hat{\theta}_n + a_n)$, with arbitrary adjustment term $a \in \mathbb{B}$ such that $a = o_p(1)$. That is, it is not possible to rank consistent estimators.
\par
To overcome the drawback of the asymptotic minimaxity based on the IMR, it is necessary to employ other risk functions. The choice here is the regret concerning the IMR: for an STR $\delta_n$ and $h \in T(P_0)$, its \emph{regret of identifiable maximum risk (RIMR)} is given by
\begin{equation}
    \mathrm{RIMR}_n(\delta_n,h) \coloneqq \mathrm{IMR}_n(\delta_n,h) - \min_{\delta_n} \mathrm{IMR}_n(\delta_n,h).
    \label{eq:identifiable-maximum-risk-regret}
\end{equation}
This measures the loss in terms of the IMR because of not knowing the true state $h$. The RIMR is scaled by $\sqrt{n}$ in the analysis. A similar approach has also been adopted in \citet{Christensen2022} and \citet{Song2014}. From \cref{eq:identifiable-maximum-risk-expansion}, observe that
\begin{equation*}
    \min_{\delta_n} \imr_n(\delta_n,h) = \frac{\tau_U^+(\theta_n(h))\tau_L^-(\theta_n(h))}{\tau_U^+(\theta_n(h)) + \tau_L^-(\theta_n(h))},
\end{equation*}
where the minimum is attained by $\delta_n = \kappa(\theta_n(h))$. Substituting the above equation and \cref{eq:identifiable-maximum-risk-expansion} in \cref{eq:identifiable-maximum-risk-regret}, the RIMR can be expressed as 
\begin{align}
    \begin{aligned}
        \mathrm{RIMR}_n(\delta_n,h)
        =\max\dbrace*{
        \begin{aligned}
            &\tau_L^-(\theta_n(h))\expect_{n,h*}\dbrack*{\delta_n - \kappa(\theta_n(h))},\\
            -&\tau_U^+(\theta_n(h))\expect_{n,h*}\dbrack*{\delta_n - \kappa(\theta_n(h))}
        \end{aligned}
        }.
    \end{aligned}
    \label{eq:rimr-expansion}
\end{align}
\par
For the following technical reason, it is challenging to deal directly with $\mathrm{RIMR}_{n}(\delta_n,h)$ in the asymptotic analysis, whence this study employs a modified version. In the process of obtaining a lower bound of the asymptotic maximum risk for a broad class of statistics, it is often necessary to evaluate the asymptotic lower bound of $\expect_{n,h*}[\ell(\sqrt{n}(\delta_n-\kappa(\theta_n(h)))]$ for some function $\ell:\mathbb{R}\to\mathbb{R}$. When one is interested in the point estimation of $\kappa(\theta_0)$, $\ell$ is usually chosen to be a subconvex loss function, such as squared loss $\ell(x) = x^2$. As the subconvex loss functions are bounded from below and satisfy lower semicontinuity, one can apply the Portmanteau lemma to obtain the desired bound. In the current problem, however, $\ell$ is the identity function that is not bounded; hence, one cannot employ the usual strategy. To circumvent this challenge, this study modifies $\rimr_{n}(\delta_n,h)$ by replacing $\ell$ with its truncated version. Specifically, it employs $\ell_M(x) \coloneqq \max\{-M,\min\{M,x\}\}$ for $M \in \mathbb{N}$ instead of $\ell$. Correspondingly, the modified $\rimr_n(\delta_n,h)$ becomes
\begin{align*}
    \begin{aligned}
        \mathrm{RIMR}_{M,n}(\delta_n,h)
        \coloneqq \max\dbrace*{
        \begin{aligned}
            &\tau_L^-(\theta_n(h))\expect_{n,h*}\dbrack*{\ell_M\dparen*{\sqrt{n}\dparen*{\delta_n - \kappa(\theta_n(h))}}},\\
            -&\tau_U^+(\theta_n(h))\expect_{n,h*}\dbrack*{\ell_M\dparen*{\sqrt{n}\dparen*{\delta_n - \kappa(\theta_n(h))}}}
        \end{aligned}
        }.
    \end{aligned}
\end{align*} 

\section{Asymptotically Minimax Optimal Statistical Treatment Rules}
\label{sec:asymptotically-minimax-optimal-statistical-treatment-rules}
\subsection{Lower Bound of Asymptotic Maximum RIMR}
We first obtain the lower bound of the asymptotic maximum RIMR for a broad class of STRs in \cref{thm:rimr-lower-bound}.
\begin{theorem}\label{thm:rimr-lower-bound}
    Suppose \cref{asum:bounds-of-identificaiton-region,asum:tangent-set,asum:theta-differentiable} hold. Let $\delta_n$ be an arbitrary STR such that $\sqrt{n}\dparen{\delta_n - \kappa(\theta_0)}$ is asymptotically tight and asymptotically measurable under $P_0^n$. Then, there exists $M_0 \in \mathbb{N}$ such that for every $M \geq M_0$,
    \begin{align}
        \MoveEqLeft
        \adjustlimits{\sup}_{I}{\liminf}_{n \to \infty}\sup_{h \in I} \sqrt{n}\mathrm{RIMR}_{M,n}(\delta_n,h)\nonumber\\
        \geq {} & \adjustlimits{\inf}_{w \in \mathbb{B}}{\sup}_{h \in T(P_0)} \max\dbrace*{
        \begin{aligned}
            &\tau_L^-(\theta_0)\expect\dbrack*{ \kappa_{\theta_0}'\dparen*{\mathbb{G}_0 + w + \Dot{\theta}_0(h) } - \kappa_{\theta_0}'\dparen*{ \Dot{\theta}_0(h) } },\\
            -&\tau_U^+(\theta_0)\expect\dbrack*{ \kappa_{\theta_0}'\dparen*{\mathbb{G}_0 + w + \Dot{\theta}_0(h) } - \kappa_{\theta_0}'\dparen*{ \Dot{\theta}_0(h) } }
        \end{aligned}
        }.
        \label{eq:rimr-lower-bound}
    \end{align}
\end{theorem}
\cref{thm:rimr-lower-bound} gives the lower bound of the asymptotic maximum RIMR. This lower bound holds for a broad class of STRs: the only requirement for STRs is the asymptotic tightness and measurability of the standardized statistic $\sqrt{n}(\delta_n - \kappa(\theta_0))$. To see the benefit of this lower bound, it is important to notice that the integrand in the right-hand side of \cref{eq:rimr-lower-bound} is the weak limit of a certain class of statistics. Specifically, noting that 
\begin{equation*}
    \sqrt{n}\dparen*{\hat{\theta}_n + w/\sqrt{n} - \theta_0} \overset{h}{\leadsto} \mathbb{G}_0 + w + \Dot{\theta}_0(h)
\end{equation*}
for any $w$ and $h$, apply the generalized delta method to obtain
\begin{align*}
    \sqrt{n}\dparen*{\kappa\dparen*{ \hat{\theta}_n + \frac{w}{\sqrt{n}} } - \kappa\dparen*{\theta_n(h)} } 
    &= \sqrt{n}\dparen*{ \kappa\dparen*{ \hat{\theta}_n + \frac{w}{\sqrt{n}} } - \kappa(\theta_0) } - \sqrt{n}\dparen*{ \kappa(\theta_n(h)) - \kappa(\theta_0) }\\
    &\overset{h}{\leadsto} \kappa_{\theta_0}'\dparen*{\mathbb{G}_0 + w + \Dot{\theta}_0(h) } - \kappa_{\theta_0}'\dparen*{ \Dot{\theta}_0(h) }.
\end{align*}
This implies that if there exists an adjustment term $w_{\theta_0}^* \in \mathbb{B}$ that solves the minimization in the right-hand side of \cref{eq:rimr-lower-bound}, the LAM STR can be constructed by
\begin{equation}
    \kappa\dparen*{\hat{\theta}_n + \frac{w_{\theta_0}^*}{\sqrt{n}}}.
    \label{eq:infeasible-LAM-STR}
\end{equation}
If $w_{\theta_0}^* \neq 0$, this implies that the plug-in STR is not LAM. Instead, it is essential to find an optimal adjustment term $w_{\theta_0}^*$ to obtain the LAM STR.
\par
A special case is when $\kappa(\theta)$ is Hadamard differentiable at $\theta_0$. In this case, the optimal adjustment term $w_{\theta_0}^*$ always equals zero. Indeed, the full differentiability implies the linearity of $\kappa_{\theta_0}'$, and hence
\begin{align*}
    &\sup_{h \in T(P_0)} \max\dbrace*{
    \begin{aligned}
        &\tau_L^-(\theta_0)\expect\dbrack*{ \kappa_{\theta_0}'\dparen*{\mathbb{G}_0 + w + \Dot{\theta}_0(h) } - \kappa_{\theta_0}'\dparen*{ \Dot{\theta}_0(h) } },\\
        -&\tau_U^+(\theta_0)\expect\dbrack*{ \kappa_{\theta_0}'\dparen*{\mathbb{G}_0 + w + \Dot{\theta}_0(h) } - \kappa_{\theta_0}'\dparen*{ \Dot{\theta}_0(h) } }
    \end{aligned}
    }\\
    &= 
    \max\dbrace*{ \tau_L^-(\theta_0)\dparen{ \expect\dbrack*{ \kappa_{\theta_0}'(\mathbb{G}_0) } + \kappa_{\theta_0}'(w) }, -\tau_U^+(\theta_0)\dparen{ \expect\dbrack*{ \kappa_{\theta_0}'(\mathbb{G}_0) } + \kappa_{\theta_0}'(w) } }.
\end{align*}
Moreover, the linearity of $\kappa_{\theta_0}'$ implies that $\kappa_{\theta_0}'$ is an element of the dual space $\mathbb{B}^*$, where $\kappa_{\theta_0}'(\mathbb{G}_0) \sim N(0, \norm{\Dot{\theta}_0^*\kappa_{\theta_0}'}_{2,P_0}^2)$. Thus, noting that the right-hand side of \cref{eq:rimr-lower-bound} is not smaller than zero, the choice, $w_{\theta_0}^* = 0$, is optimal.
\begin{remark}\label{rem:rimr-lower-bound-support-expression}
    The lower bound in \cref{thm:rimr-lower-bound} has an equivalent expression. Specifically, let $s = \Dot{\theta}_0(h) \in \mathbb{B}$ for $h \in T(P_0)$. Then, it is possible to replace the supremum for $h \in T(P_0)$ with the supremum for $s \in \overline{\Dot{\theta}_0(T(P_0))}$, as the corresponding objective function is continuous in $s$. Therefore, the right-hand side of \cref{eq:rimr-lower-bound} equals
    \begin{align}
        \begin{aligned}
            \adjustlimits{\inf}_{w\in\mathbb{B}}{\sup}_{s \in S(\mathbb{G}_0)} \max\dbrace*{
            \begin{aligned}
                &\tau_L^-(\theta_0)\expect\dbrack*{ \kappa_{\theta_0}'\dparen*{ \mathbb{G}_0 + w + s } - \kappa_{\theta_0}'(s) },\\
                -&\tau_U^+(\theta_0)\expect\dbrack*{ \kappa_{\theta_0}'\dparen*{ \mathbb{G}_0 + w + s} - \kappa_{\theta_0}'(s) }
            \end{aligned}
            },
        \end{aligned}
        \label{eq:rimr-lower-bound-support-expression} 
    \end{align}
    where $S(\mathbb{G}_0)$ denotes the support of $\mathbb{G}_0$. Especially when $\mathbb{B} = \mathbb{R}^k$, the support is equal to the linear span of the column vectors of the efficiency bound of estimators for $\theta_0$, $\Sigma_{\theta_0} = \expect_{P_0}[\tilde{\theta}_0^{~}\tilde{\theta}_0^\top]$. Moreover, if $\Sigma_{\theta_0}$ is nonsigular, then it holds that $S(\mathbb{G}_0) = \mathbb{R}^k$. This expression will be used for constructing the optimal STR that is LAM with respect to RIMR in \cref{sec:lam-str-with-respect-to-rimr}.
\end{remark}
\begin{remark}
    \textcite{Ponomarev2022,Song2014a,Fang2016} have derived the lower bounds of the asymptotic maximum mean squared error (MSE) to develop the LAM point estimators for parameters characterized by nondifferentiable functionals. When one views STRs as estimators for $\kappa(\theta)$, \textcite{Ponomarev2022}'s result, combined with the conditions in \cref{thm:rimr-lower-bound}, implies that the asymptotic maximum MSE is bounded from below as in 
    \begin{align}
        \MoveEqLeft
        \adjustlimits\sup_I\liminf_{n\to\infty}\sup_{h \in I} \mathbb{E}_{n,h*}\dbrack*{\dparen*{\sqrt{n}(\delta_n - \kappa(\theta_n(h)))}^2}\nonumber\\
        \geq {} & \adjustlimits\inf_{w \in \mathbb{B}}\sup_{s \in S(\mathbb{G}_0)}\mathbb{E}\dbrack*{ \dparen*{ \kappa_{\theta_0}'\dparen*{ \mathbb{G}_0 + w + s } - \kappa_{\theta_0}'\dparen*{ s } }^2 }.
        \label{eq:mse-lower-bound}
    \end{align}
    This lower bound suggests that the LAM point estimator also takes the form of $\kappa(\hat{\theta}_n + w_{\theta_0}^\dagger)$, where $w_{\theta_0}^\dagger$ is the optimal adjustment term that solves the minimization problem in the right-hand side of the preceding inequality. As the objective functions in \cref{eq:rimr-lower-bound,eq:mse-lower-bound} differ, the optimal adjustment terms, $w_{\theta_0}^*$ and $w_{\theta_0}^\dagger$, generally differ as well. This is exemplified using \cref{ex:experimental-study-with-non-random-sampling} below.
\end{remark}
\begin{examplecontinued}{ex:identification-region-with-empirical-evidence-alone}
    As $\kappa(\theta)$ is Hadamard differentiable at $\theta_0$, as long as $\theta_0$ is an interior point, the above discussion implies that the lower bound in \cref{eq:rimr-lower-bound} can always be made zero by choosing $w_{\theta_0}^* = 0$. Thus, the plug-in STR is LAM in this example.
\end{examplecontinued}
\begin{examplecontinued}{ex:experimental-study-with-non-random-sampling}
    \label{ex:example2-fifth}
    As noted above, if neither $\tau_L(\theta_0) = 0$ nor $\tau_U(\theta_0) = 0$, then $\kappa(\theta)$ is Hadamard differentiable at $\theta_0$. Thus, the above argument implies that the optimal adjustment term $w_{\theta_0}^*$ equals zero. If $\tau_L(\theta_0) = 0 < \tau_U(\theta_0)$, $\kappa(\theta)$ is Hadamard directionally differentiable at $\theta_0$ with derivative $\kappa_{\theta_0}'(b) = \min\dbrace{\tau_{L,\theta_0}'(b),0}/\tau_U(\theta_0)$. With the linearity of $\tau_{L,\theta_0}'$, \cref{rem:rimr-lower-bound-support-expression} implies that the lower bound is the infimum of 
    \begin{equation*}
        \sup_{s\in\mathbb{R}^3}\max\dbrace*{0, \min\dbrace*{\tau_{L,\theta_0}'(s),0} - \expect\dbrack*{ \min\dbrace*{\tau_{L,\theta_0}'(\mathbb{G}_0) + \tau_{L,\theta_0}'(w) + \tau_{L,\theta_0}'(s), 0  }  } }.
    \end{equation*}
    The direct argument based on the closed-form expression of $\expect\dbrack{ \min\dbrace{\tau_{L,\theta_0}'(\mathbb{G}_0 + w + s), 0  }  }$ shows that the supremum in the preceding display always attained at $s$ such that $\tau_{L,\theta_0}'(s) = 0$. That is, the object in the preceding display is equal to 
    \begin{equation*}
        \max\dbrace*{0,-\mathbb{E}\dbrack*{\min\dbrace*{\tau_{L,\theta_0}'(\mathbb{G}_0) + \tau_{L,\theta_0}'(w),0}}}.
    \end{equation*}
    Then, it is clear that the above object monotonically approaches zero from above by moving $w$ such that $\tau_{L,\theta_0}'(w)\to\infty$. Hence, the lower bound remains zero, although the optimal adjustment term $w_{\theta_0}^*$ does not exist. This argument implies that the plug-in STR $\kappa(\hat{\theta}_n)$ is not LAM STR. The plug-in STR is asymptotically inferior to the STR $\kappa(\hat{\theta}_n + w /\sqrt{n})$ for any $w$ satisfying $\tau_{L,\theta_0}'(w) > 0$.\par
    The LAM point estimator that minimizes the asymptotic maximum MSE is obtained by solving 
    \begin{equation*}
        \adjustlimits\inf_{w \in \mathbb{R}^3}\sup_{s \in \mathbb{R}^3} \mathbb{E}\dbrack*{ \dparen*{ \min\dbrace*{ \tau_{L,\theta_0}'(\mathbb{G}_0 + w + s), 0 } - \min\dbrace*{ \tau_{L,\theta_0}'(s), 0 } }^2 }.
    \end{equation*}
    The direct calculation shows that the optimal adjustment term that solves the above infimum is zero (i.e., $w_{\theta_0}^\dagger = 0$). Therefore, the LAM point estimator is the plug-in estimator $\kappa(\hat{\theta}_n)$. The argument in the preceding paragraph implies that the LAM point estimator is also inferior to the STR $\kappa(\hat{\theta}_n + w/\sqrt{n})$ in terms of the RIMR when $\tau_{L,\theta_0}'(w) > 0$. 
\end{examplecontinued}
\subsection{Feasible Construction of LAM STR}
\label{sec:lam-str-with-respect-to-rimr}
If the value $\theta_0$ of the intermediate parameter at the true data generating process $P_0$ were known, one could solve the minimax problem in \cref{eq:rimr-lower-bound-support-expression} to obtain the optimal adjustment term $w_{\theta_0}^*$. The value $\theta_0$ is unknown in practice; hence, one cannot solve the minimax problem because of the objects that depend on $\theta_0$. Specifically, five unknown objects exist given the unknown identity of $\theta_0$. First, the limit law $\mathbb{G}_0$ of the best regular estimator $\hat{\theta}_n$ is unknown. Correspondingly, the support $S(\mathbb{G}_0)$ of $\mathbb{G}_0$ is not known as well. The Hadamard directional derivatives $\tau_{L,\theta_0}'$ and $\tau_{U,\theta_0}'$ are also unknown, and so is $\kappa_{\theta_0}'$. Finally, $\tau_L^-(\theta_0)$ and $\tau_U^+(\theta_0)$ are unknown quantities.
\par
To alleviate this challenge, this study appropriately estimates those unknown objects and solves the sample analog of \cref{eq:rimr-lower-bound-support-expression}, where the unknown objects are replaced with the estimators, to obtain the estimate $\hat{w}_n \in \mathbb{B}$ for $w_{\theta_0}^*$. If $\hat{w}_n$ is a consistent estimator for $w_{\theta_0}^*$, then the STR given by $\kappa(\hat{\theta}_n + \hat{w}_n/\sqrt{n})$ is LAM. To ensure the consistency of $\hat{w}_n$, the estimators of unknown quantities must satisfy the appropriate conditions. In what follows, the study states the required conditions, defines our proposed STR formally, and then shows that the STR is LAM with respect to RIMR.
\par
First, estimate the law of $\mathbb{G}_0$ by that of the estimator $(\mathbf{Z}_n,\mathbf{W}_n) \mapsto \hat{\mathbb{G}}_n \in \mathbb{B}$, where $\mathbf{W}_n$ is independent of $\mathbf{Z}_n$. This estimator is particularly constructed via the bootstrap or simulation. For the bootstrap, let $\mathbf{W}_n = (W_{1},\dots,W_{n})$ be the nonparametric bootstrap weights independent of $\mathbf{Z}_n$, and let $(\mathbf{Z}_n,\mathbf{W}_n) \mapsto \hat{\theta}_n^*$ be the bootstrap version of the best regular estimator $\hat{\theta}_n$. Then, use $\hat{\mathbb{G}}_n = \sqrt{n}(\hat{\theta}_n^* - \hat{\theta}_n)$ to approximate the law of $\mathbb{G}_0$. For the simulation, presume the law of $\mathbb{G}_0$ belongs to a known class of distributions. For example, $\mathbb{G}_0 \sim N(0,\Sigma_{\theta_0})$ when $\mathbb{B} = \mathbb{R}^{d_\theta}$. Then, given a consistent estimator $\hat{\Sigma}_n$ for $\Sigma_{\theta_0}$, use $ \hat{\mathbb{G}}_n = \mathbf{W}_n \sim N(0,\hat{\Sigma}_n)$ to approximate the law of $\mathbb{G}_0$. In general, the estimator $\hat{\mathbb{G}}_n$ must satisfy \cref{asum:on-bootstrap} below. Given this assumption, $\hat{\mathbb{G}}_n \leadsto \mathbb{G}_0$ conditionally on $Z^n$; and hence, it is reasonable to approximate the law of $\mathbb{G}_0$ by $\hat{\mathbb{G}}_n$ \parencite[][Theorems~3.6.1 and 3.9.11]{Vaart1996}.
\begin{assumption}\label{asum:on-bootstrap}
    The estimator $(\mathbf{Z}_n,\mathbf{W}_n) \mapsto \hat{\mathbb{G}}_n \in \mathbb{B}$ with $\mathbf{W}_n$ independent of $\mathbf{Z}_n$ has the following properties: (i) $\sup_{f\in \mathrm{BL}_1(\mathbb{B})} \abs{ \expect \dbrack[]{f\dparen{\hat{\mathbb{G}}_n}|\mathbf{Z}_n} - \expect \dbrack[]{f(\mathbb{G}_0)} } = o_p(1)$, where $\mathrm{BL}_1(\mathbb{B})$ is the space of functions $f:\mathbb{B} \to \mathbb{R}$ such that $\sup_{b \in \mathbb{B}}\abs*{f(b)} \leq 1$ and $\abs*{f(b_1) - f(b_2)} \leq \normB{b_1 - b_2}$ for all $b_1,b_2 \in \mathbb{B}$; (ii) $\hat{\mathbb{G}}_n$ is asymptotically measurable jointly in $(\mathbf{Z}_n,\mathbf{W}_n)$.
\end{assumption}
Second, approximate the support $S(\mathbb{G}_0)$ by an expanding sequence of compact sets. The rationale of this strategy is based on the tightness of $\mathbb{G}_0$. Specifically, $\mathbb{G}_0$ has a tight Borel law if and only if its support is a $\sigma$-compact set (i.e., a countable union of compact sets). Thus, the support of $\mathbb{G}_0$ can be approximated by an expanding sequence of compact sets $S_n \subset \mathbb{B}$. For the property of this set sequence, assume the following.
\begin{assumption}\label{asum:support-G0-compact-sieve}
    There exists an expanding sequence of compact sets $S_n \subset \mathbb{B}$ such that (i) for every $s \in S(\mathbb{G}_0)$ and for any $\epsilon > 0$, there is $s_n \in S_n$ such that $\norm{s_n - s}_{\mathbb{B}} \leq \epsilon$  for sufficiently large $n$; (ii) $\sup_{s \in S_n} \normB{s} = o(\sqrt{n})$. The sequence $\{S_n\}$ are known or estimated by $\{\hat{S}_n\}$ satisfying $d_H(\hat{S}_n,S_n) = o_p(1)$, where $d_H(\hat{S}_n,S_n)$ denotes the Hausdorff distance between the sets $\hat{S}_n$ and $S_n$.
\end{assumption}
As stated in \cref{asum:support-G0-compact-sieve}, the expanding sequence $\{S_n\}$ is not necessary to be known ex-ante. Rather, it is sufficient to have an estimator $\hat{S}_n$ whose Hausdorff distance from $S_n$ goes to zero in probability. A typical candidate of $\hat{S}_n$ and $S_n$ for the case when $\mathbb{B} = \mathbb{R}^{d_\theta}$ is constructed as follows. Let $\hat{\Sigma}_n$ be a $\sqrt{n}$-consistent estimator for the efficiency bound $\Sigma_{\theta_0}$ of estimators for $\theta_0$. Moreover, let $\hat{\sigma}_j$ and $\sigma_j$ be the $j$-th column vector of $\hat{\Sigma}_n$ and $\Sigma_{\theta_0}$, respectively. Then, the following expanding sequence of compact sets usually fulfills the requirement:
\begin{equation*}
    S_n = \dbrace*{\sum_{j=1}^k a_j \sigma_j : \norm{a} \leq \lambda_n}
    \quad\text{and}\quad
    \hat{S}_n = \dbrace*{ \sum_{j=1}^k a_j \hat{\sigma}_j : \norm{a} \leq \lambda_n},
\end{equation*}
where $\lambda_n = o(\sqrt{n})$. Furthermore, in some examples such as \cref{ex:identification-region-with-empirical-evidence-alone,ex:experimental-study-with-non-random-sampling}, $\Sigma_{\theta_0} = \mathbb{R}^{d_\theta}$. Then, it suffices to take as $S_n$ the closed ball with center $0$ and radius $\lambda_n$. For the choice of $S_n$ and $\hat{S}_n$ when $\mathbb{B}$ is a Banach space, see Section~5.2.2 in \textcite{Ponomarev2022}.
\par
Third, this study postulates that the estimators $\hat{\tau}_{U,n}'$ and $\hat{\tau}_{L,n}'$ for $\tau_{U,\theta_0}'$ and $\tau_{L,\theta_0}'$ satisfying the following assumption are available.
\begin{assumption}\label{asum:bounds-derivative-estimator}
    For each $j \in \{U,L\}$, the estimator $\hat{\tau}_{j,n}':\mathbb{B} \to \mathbb{R}$ is a function of $\mathbf{Z}_n$ and has the following properties: (i) for every $\delta > 0$, $\sup_{b \in K_n^\delta} \abs{ \hat{\tau}_{j,n}'(b) - \tau_{j,\theta_0}'(b) } = o_p(1)$ for arbitrary expanding sequence of compact sets $K_n \subset \mathbb{B}$ with $\sup_{b \in K_n}\normB{b} = o(\sqrt{n})$; (ii) there exists an asymptotically tight random variable $C_{\hat{\tau}_{j,n}'} \in \mathbb{R}$ such that $\abs{\hat{\tau}_{j,n}'(b_1) - \hat{\tau}_{j,n}'(b_2)} \leq C_{\hat{\tau}_{j,n}'}\norm{b_1-b_2}_{\mathbb{B}}$ holds outer almost surely for every $(b_1,b_2) \in \mathbb{B}\times\mathbb{B}$.
\end{assumption}
\cref{asum:bounds-derivative-estimator}.(i) requires some uniform consistency of $\hat{\tau}_{j,n}'$. In many applications, estimators with stronger properties such as $\hat{\tau}_{j,n}(b) = \tau_{j,\theta_0}'(b)$ with probability approaching 1 are often available. See Remarks 3.3 and 3.4 and Section~S.4 in the supplementary material of \textcite{Fang2019} for more details. Usually, one can construct such estimators from the analytical expression of $\tau_{j,\theta_0}'$ or numerical method proposed in \textcite{Hong2018}.
\par
It is now appropriate to define the proposed STR. Given the estimators $\hat{\tau}_{U,n}'$ and $\hat{\tau}_{L,n}'$ for $\tau_{U,\theta_0}'$ and $\tau_{L,\theta_0}'$, define the estimator $\hat{\kappa}_n'$ for $\kappa_{\theta_0}'$ by 
\begin{equation}
    \hat{\kappa}_n'(b) \coloneqq \frac{\tau_L^-(\hat{\theta}_n)\cdot \hat{m}_{\tau_U(\hat{\theta}_n),n}'\dparen*{\hat{\tau}_{U,n}'(b)} - \tau_U^+(\hat{\theta}_n) \cdot \hat{m}_{-\tau_L(\hat{\theta}_n),n}\dparen*{-\hat{\tau}_{L,n}'(b)}}{\dparen[\big]{\tau_U^+(\hat{\theta}_n) + \tau_L^-(\hat{\theta}_n)}^2}
    \label{eq:kappa-hat}
\end{equation}
for $b \in \mathbb{B}$. Here, $\hat{m}_{y,n}$ is the estimator of the directional derivative $m_y'$ of $\max\{x,0\}$ at $y \in \mathbb{R}$, and, for any $x \in \mathbb{R}$, it is given by
\begin{equation*}
    \hat{m}_{y,n}(r) = x\cdot \ind\dbrace{y > \epsilon_n} + \max\dbrace{x,0}\cdot \ind\dbrace*{ \abs*{y} \leq \epsilon_n } + 0\cdot \ind\dbrace{y < -\epsilon_n}
\end{equation*}
with $\epsilon_n > 0$ such that $\epsilon_n \downarrow 0$ and $\sqrt{n}\epsilon_n \uparrow \infty$. Let $\{K_m\}$ be an expanding sequence of compact sets in $\mathbb{B}$. For a sufficiently large constant $M > 0$, let 
\begin{equation}
    \begin{aligned}
        \hat{w}_{n} \in \adjustlimits{\argmin}_{w \in K_m}{\sup}_{s \in \hat{S}_n} \max\dbrace*{
        \begin{aligned}
            &\tau_L^-(\hat{\theta}_n)\expect\dbrack*{ \ell_M\dparen*{ \hat{\kappa}_n'\dparen*{ \hat{\mathbb{G}}_n + w + s } - \hat{\kappa}_n'(s)} \middle| \mathbf{Z}_n},\\
            -&\tau_U^+(\hat{\theta}_n)\expect\dbrack*{ \ell_M\dparen*{ \hat{\kappa}_n'\dparen*{ \hat{\mathbb{G}}_n + w + s } - \hat{\kappa}_n'(s)} \middle| \mathbf{Z}_n }
        \end{aligned}        
        }.
    \end{aligned}
    \label{eq:optimal-adjustment-term}
\end{equation}
Then, define the STR by
\begin{equation}
    \delta_n^{\text{LAM}} \coloneqq \kappa\dparen*{ \hat{\theta}_n + \frac{\hat{w}_n}{\sqrt{n}} },
    \label{eq:LAM-STR}
\end{equation}
where $\hat{\theta}_n$ is the best regular estimator for $\theta_0$. The following theorem gives an upper bound of local asymptotic maximum RIMR of $\delta_n^{\text{LAM}}$.
\begin{theorem}\label{thm:lam-with-respect-to-rimr}
    Suppose \cref{asum:bounds-of-identificaiton-region,asum:tangent-set,asum:theta-differentiable,asum:on-bootstrap,asum:support-G0-compact-sieve,asum:bounds-derivative-estimator} hold. Then,
    \begin{align}
        \begin{aligned}
            \MoveEqLeft
            \adjustlimits\limsup_{m \to \infty}\limsup_{M \to \infty}\adjustlimits\sup_{I}\liminf_{n\to\infty}\sup_{h \in I} \sqrt{n}\rimr_{M,n}(\delta_n^{\lam},h)\\
            \leq {} & \adjustlimits{\inf}_{w \in \mathbb{B}}{\sup}_{h \in T(P_0)} \max\dbrace*{
            \begin{aligned}
                &\tau_L^-(\theta_0)\expect\dbrack*{ \kappa_{\theta_0}'\dparen*{\mathbb{G}_0 + w + \Dot{\theta}_0(h) } - \kappa_{\theta_0}'\dparen*{ \Dot{\theta}_0(h) }},\\
                -&\tau_U^+(\theta_0)\expect\dbrack*{ \kappa_{\theta_0}'\dparen*{\mathbb{G}_0 + w + \Dot{\theta}_0(h) } - \kappa_{\theta_0}'\dparen*{ \Dot{\theta}_0(h) }}
            \end{aligned}
            }.
        \end{aligned}
        \label{eq:rimr-upper-bound-LAM-STR}
    \end{align}
\end{theorem}
The right-hand sides of \cref{eq:rimr-lower-bound} and \cref{eq:rimr-upper-bound-LAM-STR} coincide. Therefore, combining \cref{thm:rimr-lower-bound,thm:lam-with-respect-to-rimr}, $\delta_n^{\lam}$ is LAM STR.

\section{Simulation Study}
This simulation study examines the finite-sample performance of the proposed LAM STR in the context of \cref{ex:experimental-study-with-non-random-sampling}. In the setup in \cref{ex:experimental-study-with-non-random-sampling}, consider a binary outcome (i.e., $\mathcal{Y} = \{0,1\}$) and set the true parameter to
\begin{equation*}
    \theta_0 = (\mu_{1,0},\mu_{0,0},p_0) = \dparen*{ \frac{2}{3}, \frac{1}{3}, \frac{3}{4} }.
\end{equation*}
In this case, $\tau_L(\theta_0) = 0$ and $\tau_U(\theta_0) = 1/2$. Although the boundary functions $\tau_L(\theta)$ and $\tau_U(\theta)$ are differentiable at $\theta_0$, the optimal treatment rule $\kappa(\theta)$ is only directionally differentiable at $\theta_0$. Focus on a specific set of perturbed data generating processes characterized by
\begin{equation*}
    \theta_n(h) = \dparen*{ \frac{2}{3} + \frac{h}{\sqrt{n}}, ~ \frac{1}{3} + \frac{h}{2\sqrt{n}}, ~ \frac{3}{4} + \frac{h}{\sqrt{n}} }
\end{equation*}
for $h \in I \coloneqq \{-2,-1.95,-1.90,\dots,2\}$. When $h = 0$, $\theta_n(h)$ corresponds to the true parameter value $\theta_0$. As $|h|$ gets larger, $\theta_n(h)$ deviates more from the nondifferentiable point. The goal here is to compare the maximum RIMR,  
\begin{equation*}
    \sup_{h \in I} \rimr_n(\delta_n,h),
\end{equation*}
between the proposed LAM STR and plug-in STR.
\par
This study describes details for constructing the LAM STR. Given a sample $\{(S_iY_i,S_iD_i,S_i)\}_{i=1}^n$, $\theta$ is estimated by the best regular estimator $\hat{\theta}_n$ described in \cref{ex:experimental-study-with-non-random-sampling} on page \pageref{ex:example2-fourth}. This estimator satisfies $\sqrt{n}(\hat{\theta}_n - \theta_n(h)) \overset{h}{\leadsto} \mathbb{G}_0 \sim N(0,\Sigma_{\theta_0})$ for any $h$. Given a $\sqrt{n}$ consistent estimator $\hat{\Sigma}_n$ for $\Sigma_{\theta_0}$, estimate the law of $\mathbb{G}_0$ by drawing $G_l,l=1,\dots,L$ independently from $N(0,\hat{\Sigma}_n)$, where $L = 1000$. The Hadamard derivatives, $\tau_{L,\theta_0}'(b)$ and $\tau_{U,\theta_0}'(b)$, are estimated by $\hat{\tau}_{L,n}'(b) = \dparen{ \frac{\partial \tau_L}{\partial \theta}|_{\theta = \hat{\theta}_n} }^\intercal b$ and $\hat{\tau}_{U,n}'(b) = \dparen{ \frac{\partial \tau_U}{\partial \theta}|_{\theta = \hat{\theta}_n} }^\intercal b$, respectively. In the estimator for the directional derivative of $\max\{x,0\}$, set $\epsilon_n = n^{1/3}$. This study sets the truncation constant $M$ for the loss function to be $1000$. As the efficiency bound $\Sigma_{\theta_0}$ is nonsingular, the support of $\mathbb{G}_0$ corresponds to $\mathbb{R}^3$. It implies that the support can be approximated by expanding cubes. Hence, set $\hat{S}_n = S_n = [-\lambda_n,\lambda_n]^3$, where $\lambda_n = n^{1/3}$. The compact set $K_m$ over which the optimal adjustment term is searched is set to $[-2,2]^3$. Given these specifications, solve the next minimization problem to obtain the data-dependent adjustment term $\hat{w}_n$:
\begin{equation*}
    \adjustlimits{\inf}_{w \in K_m}{\sup}_{s \in \hat{S}_n} \max\dbrace*{
    \begin{aligned}
        &\tau_L^-(\hat{\theta}_n)\frac{1}{L}\sum_{l=1}^L \ell_M\dparen*{ \hat{\kappa}_n'(G_l + w + s) - \hat{\kappa}_n'(s) },\\
        -&\tau_U^+(\hat{\theta}_n)\frac{1}{L}\sum_{l=1}^L \ell_M\dparen*{ \hat{\kappa}_n'(G_l + w + s) - \hat{\kappa}_n'(s) }
    \end{aligned}        
    }.
\end{equation*}
\par
For each $h \in H$, draw $J = 5000$ independent samples of data $\mathbf{Z}_n^{(j)},j = 1,\dots,J$. Each observation $Z_i^{(j)}$ in data $\mathbf{Z}_n^{(j)} = (Z_1^{(j)},\dots,Z_{n}^{(j)})$ comprises $(S_iY_i,S_iD_i,S_i)$ drawn from the data generating process characterized by $\theta_n(h)$ and $\Pr(D=1|S=1) = 1/2$. Based on this data, obtain the best regular estimate $\hat{\theta}_n^{(j)}$ and data-dependent adjustment term $\hat{w}_n^{(j)}$. Given the collection, $\hat{\theta}_n^{(1)},\dots,\hat{\theta}_n^{(J)}$ and $\hat{w}_n^{(1)},\dots,\hat{w}_n^{(J)}$, estimate $\mathbb{E}_{n,h}[\delta_n^{\lam}]$ by $J^{-1}\sum_{j=1}^J \kappa(\hat{\theta}_n^{(j)} + \hat{w}_n^{(j)}/\sqrt{n})$. Finally, the RIMR is calculated following equation \cref{eq:rimr-expansion}. Similarly, calculate the RIMR of the plug-in STR by ignoring the data-dependent adjustment term. The size $n$ of a sample is set to 300 or 1000.
\par
\cref{fig:rimr-comparison} plots the relation between the local parameter $h$ and RIMR by the size of a sample. For both STRs, the RIMR is maximized at the dgp such that the optimal treatment rule $\kappa(\theta)$ is nondifferentiable (i.e., at $h = 0$). Importantly, the maximized RIMR of the LAM STR is lower than that of the plug-in STR. The LAM STR is designed to minimize the (asymptotic) maximum RIMR, and thus this result is as expected. Conversely when $h \neq 0$, $\kappa(\theta)$ is differentiable at $\theta_n(h)$. Nevertheless, the RIMR is relatively large in a neighborhood of zero. The LAM STR outperforms the plug-in STR in such a neighborhood. 
    
\begin{figure}[t]
     \centering
     \caption{Comparison of the Plug-in and LAM STRs}
     \begin{subfigure}[b]{0.48\textwidth}
         \centering
         \caption{$n = 300$}
         \includegraphics[width=\textwidth]{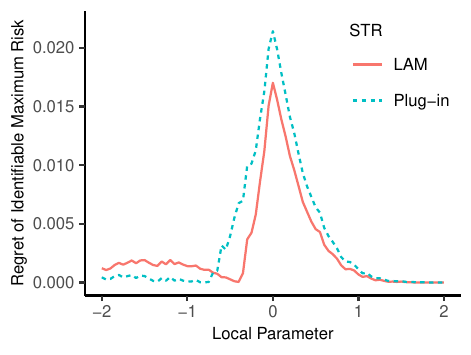}
         \label{fig:rimr-comparison-300}
     \end{subfigure}
     \hfill
     \begin{subfigure}[b]{0.48\textwidth}
         \centering
         \caption{$n = 1000$}
         \includegraphics[width=\textwidth]{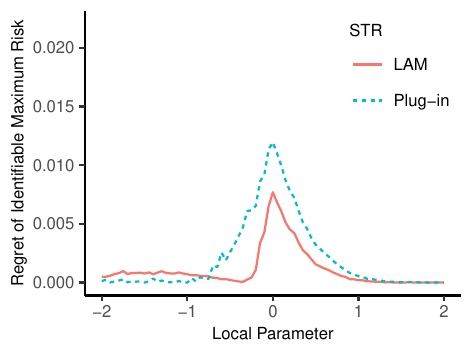}
         \label{fig:rimr-comparison-1000}
     \end{subfigure}        
    \label{fig:rimr-comparison}
\end{figure}

\section{Conclusion}\label{sec:conclusion}

This study analyzed the statistical treatment choice problem under the partial identification of the ATE. Solving the finite-sample problem exactly is difficult for general data generating processes. Instead, this study conducted the local asymptotic analysis and derived the LAM STRs. Specifically, the LAM STR takes the form of
\begin{equation*}
    \kappa\dparen*{ \hat{\theta}_n + \frac{w_{\theta_0}^*}{\sqrt{n}} },
\end{equation*}
where $\hat{\theta}_n$ is the efficient estimator for the true value $\theta_0$ of parameters, $w_{\theta_0}^*$ is the optimal adjustment term that can depend on $\theta_0$, and $\kappa$ is the function that outputs the optimal treatment rule given a value of the parameter. The adjustment term $w_{\theta_0}^*$ generally differs from zero. This implies that the LAM STR differs from the plug-in STR $\kappa(\hat{\theta}_n)$. In practice, the true parameter value $\theta_0$ is unknown; hence, so is the optimal adjustment term $w_{\theta_0}^*$. To alleviate this challenge, the study developed the data-dependent adjustment term $\hat{w}_n$, showing that the STR $\kappa(\hat{\theta}_n + \hat{w}_n/\sqrt{n})$ remained LAM.

\appendix

\appendix

\section{Proofs of Results}
\label{sec:proofs-of-results-in-texts}

\begin{proof}[Proof of \cref{lem:kappa-directionally-differentiable}]
    (i) Observe that 
    \begin{equation*}
        \kappa(\theta) = \chi \circ \phi \circ \tau_B(\theta),
    \end{equation*}
    where
    \begin{align*}
        &\tau_B:\Theta \ni \theta \mapsto (\tau_U(\theta),-\tau_L(\theta)) \in \mathbb{R}^2,\\
        &\phi: \mathbb{R}^2 \ni (x_1,x_2) \mapsto  (\max\{x_1,0\},\max\{x_2,0\}) \in \mathbb{R}^2,\\
        &\chi: \mathbb{R}^2 \ni (x_1,x_2) \mapsto \frac{x_1}{x_1 + x_2} \in \mathbb{R}.
    \end{align*}
    Given the chain rule for Hadamard directional differentiable functions, it suffices to confirm that $\tau_B$, $\phi$, and $\chi$ are Hadamard directionally differentiable.
    \par
    First, $\tau_B$ is Hadamard directionally differentiable at $\theta_0$ with derivative $\tau_{B,\theta_0}'(b) = (\tau_{U,\theta_0}'(b),-\tau_{L,\theta_0}'(b))$ for $b \in \mathbb{B}$. To this end, let $\{b_n\} \subset \mathbb{B}$ and $\{t_n\} \subset \mathbb{R}_{+}$ be sequences such that $t_n \downarrow 0$, $b_n \to b$ with $b \in \mathbb{B}$ as $n \to \infty$ and $\theta_0 + t_nb_n \in \Theta$ for all $n$. Noting that
    \begin{align*}
        &\norm*{ \frac{\tau_B(\theta_0 + t_nb_n) - \tau_B(\theta_0)}{t_n} - \tau_{B,\theta_0}'(b) }\\
        \leq {} & \abs*{ \frac{ \tau_U\dparen*{\theta_0 + t_nb_{n}} - \tau_U(\theta_0) }{t_n}  - \tau_{U,\theta_0}'(b) } + \abs*{ \frac{ \tau_L\dparen*{\theta_0 + t_nb_{n}} - \tau_L(\theta_0) }{t_n}  - \tau_{L,\theta_0}'(b) },
    \end{align*}
    \cref{assumptionenum:directional-differentiability} implies the convergence of the upper bound to zero, which implies the Hadamard directional differentiability of $\tau_B$.
    \par
    Next, the map $\phi$ is Hadamard directionally differentiable at $y = (y_1,y_2) = \tau_B(\theta_0)$ with derivative $\phi_y'(x) = (m_{y_1}'(x_1), m_{y_2}'(x_2))$. For any sequences $x_n = (x_{n,1},x_{n,2}) \in \mathbb{R}^2$ and $t_n > 0$ such that $t_n \downarrow 0$ and $x_n \to x = (x_1,x_2) \in \mathbb{R}^2$, 
    \begin{align*}
        &\norm*{ \frac{\phi(y + t_nx_n) - \phi(y)}{t_n} - \phi_y'(x)}\\
        \leq {} & \abs*{ \frac{ \max\{y_1 + t_nx_{n,1},0\} - \max\{y_1,0\} }{t_n}  - m_{y_1}'(x_1) } + \abs*{ \frac{ \max\{y_2 + t_nx_{n,2},0\} - \max\{y_2,0\} }{t_n}  - m_{y_2}'(x_2) }.
    \end{align*}
    As $\max\{x,0\}$ is Hadamard directionally differentiable at $y_1$ and $y_2$, the right-hand side of the above inequality approaches zero as $n \to \infty$. This implies the directional differentiability of $\phi$.
    \par
    Third, the map $\chi$ is Hadamard directionally differentiable at $y = (y_1,y_2) = \phi\circ \tau_B(\theta_0)$. Note that $\phi\circ\tau_B(\theta_0) \neq 0$ by $\tau_L(\theta_0) < \tau_U(\theta_0)$, where $\chi\circ\phi\circ\tau_B(\theta_0)$ is well-defined. Then, $\chi$ is obviously Hadamard differentiable at $y$ with its derivative given by
    \begin{equation*}
        \chi_y'(x) = \frac{y_2}{(y_1+y_2)^2}x_1 - \frac{y_1}{(y_1+y_2)^2}x_2 \quad \text{for every} \quad x = (x_1,x_2) \in \mathbb{R}^2.
    \end{equation*}
    \par
    Combining the observations so far and applying the chain rule \parencite[][Proposition~3.6]{Shapiro1990}, $\kappa(\theta)$ is Hadamard directionally differentiable at $\theta_0$ with derivative given by
    \begin{equation*}
        \kappa_{\theta_0}'(h) = \chi_{\phi\circ\tau_B(\theta_0)}'\circ \phi_{\tau_B(\theta_0)}' \circ \tau_{B,\theta_0}'(h)
    \end{equation*}
    for every $h \in \mathbb{B}$.
    \par
    (ii) The Lipschitz continuity of $\kappa_{\theta_0}'$ can be confirmed directly. Indeed, for any pair $(h_1,h_2) \in \mathbb{B}^2$, 
    \begin{align*}
        \abs*{\kappa_{\theta_0}'(h_1) - \kappa_{\theta_0}'(h_2)}
        &\leq \frac{\tau_L^-(\theta_0)}{\dparen[\big]{\tau_U^+(\theta_0) + \tau_L^-(\theta_0)}^2}\abs*{\tau_{U,\theta_0}'(h_1) - \tau_{U,\theta_0}'(h_2)}\\
        &\quad + \frac{\tau_U^+(\theta_0)}{\dparen[\big]{\tau_U^+(\theta_0) + \tau_L^-(\theta_0)}^2}\abs*{\tau_{L,\theta_0}'(h_1) - \tau_{L,\theta_0}'(h_2)}\\
        &\leq  \frac{\tau_L^-(\theta_0)\cdot C_{\tau_{U,\theta_0}'} + \tau_U^+(\theta_0)\cdot C_{\tau_{L,\theta_0}'}}{\dparen[\big]{\tau_U^+(\theta_0) + \tau_L^-(\theta_0)}^2}\normB{h_1 - h_2},
    \end{align*}
    where the last inequality follows from \cref{assumptionenum:lipschitz-continuity}.
    \par
    (iii) For every $v \in \mathbb{R}$, let $\bar{v} = \dparen{ \tau_U^+(\theta_0) + \tau_L^-(\theta_0) }^{-1}v$. Noting that the Hadamard directional derivative of $\max\{x,0\}$ is translation equivariant, the study obtains
    \begin{align*}
        \kappa_{\theta_0}'\dparen*{ b } + v
        = {} & \frac{ \tau_L^-(\theta_0) }{ \dparen*{ \tau_U^+(\theta_0) + \tau_L^-(\theta_0) }^2 }\dbrace*{ m_{\tau_U(\theta_0)}'\dparen*{ \tau_{U,\theta_0}'\dparen*{ b } } + \bar{v} }\\
        &- \frac{ \tau_U^+(\theta_0) }{ \dparen*{ \tau_U^+(\theta_0) + \tau_L^-(\theta_0) }^2 }\dbrace*{ m_{-\tau_L(\theta_0)}'\dparen*{ -\tau_{L,\theta_0}'\dparen*{ b } } - \bar{v} }\\
        = {} & \frac{ \tau_L^-(\theta_0) }{ \dparen*{ \tau_U^+(\theta_0) + \tau_L^-(\theta_0) }^2 }\dbrace*{ m_{\tau_U(\theta_0)}'\dparen*{ \tau_{U,\theta_0}'\dparen*{ b } + \bar{v} } }\\
        & - \frac{ \tau_U^+(\theta_0) }{ \dparen*{ \tau_U^+(\theta_0) + \tau_L^-(\theta_0) }^2 }\dbrace*{ m_{-\tau_L(\theta_0)}'\dparen*{ - \dparen*{\tau_{L,\theta_0}'\dparen*{ b } + \bar{v} } } }\\
        = {} & \frac{ \tau_L^-(\theta_0) }{ \dparen*{ \tau_U^+(\theta_0) + \tau_L^-(\theta_0) }^2 }\dbrace*{ m_{\tau_U(\theta_0)}'\dparen*{ \tau_{U,\theta_0}'\dparen*{ b + \tilde{v} } } }\\
        &- \frac{ \tau_U^+(\theta_0) }{ \dparen*{ \tau_U^+(\theta_0) + \tau_L^-(\theta_0) }^2 }\dbrace*{ m_{-\tau_L(\theta_0)}'\dparen*{ - \tau_{L,\theta_0}'\dparen*{ b + \tilde{v} } } }\\
        = {} & \kappa_{\theta_0}'\dparen*{ b + \tilde{v} },
    \end{align*}
    where the third equality follows from \cref{assumptionenum:translation-equivariance}.
\end{proof}
\begin{proof}[Proof of \cref{thm:rimr-lower-bound}]
Regard the set $\mathcal{I} = \{I \subset T(P_0): I \text{ is finite}\}$ as the directed set with its partial order defined by set inclusion. There exists a subnet $\{n_I\}_{I \in \mathcal{I}}$ of $\{n\}$ such that
\begin{equation*}
    R_M \coloneqq \adjustlimits{\sup}_{I}{\liminf}_{n \to \infty}\sup_{h \in I} \sqrt{n}\mathrm{RIMR}_{M,n}(\delta_n,h) = \adjustlimits{\lim}_{I}{\sup}_{h\in I}\sqrt{n}\mathrm{RIMR}_{M,n_I}(\delta_{n_I},h)
\end{equation*}
Indeed, one can construct such a subnet by choosing $n_I \in \mathbb{N}$ such that $n_I \geq |I|$ and
\begin{equation*}
    \abs*{ \sup_{h \in I} \rimr_{M,n_I}(\delta_{n_I},h) - \adjustlimits{\liminf}_{n \to \infty}{\sup}_{h \in I} \rimr_{M,n} (\delta_n,h) } < \frac{1}{|I|}
\end{equation*}
for any $I \in \mathcal{I}$. By \cref{lem:rimr-asymptotic-representation}, for any $h \in T(P_0)$, there is a further subnet $\{n_{J}\} \subset \{n_{I}\}_{I \in \mathcal{I}}$ such that
\begin{equation*}
    \sqrt{n_{J}}\dparen[\big]{\delta_{n_{J}} - \kappa(\theta_{n_{J}}(h+\bm{a}^\intercal \bm{h}))} \overset{h+\bm{a}^\intercal \bm{h}}{\leadsto} L_{h+\bm{a}^\intercal \bm{h}}
\end{equation*}
for every $\bm{a} \in \mathbb{R}^m$, where $\bm{h} = (h_1,\dots,h_m)$ is the first $m$ elements of a complete orthonormal basis of $\overline{T(P_0)}$. Then, one can bound $R_M$ from below as follows:
\begin{align}
    R_M
    & \geq \liminf_{J} \mathrm{RIMR}_{M,n_{J}}(\delta_{n_{J}},h + \bm{a}^\intercal \bm{h})\nonumber\\
    & \geq \max\dbrace*{
    \begin{aligned}
        & \tau_L^-(\theta_0)\liminf_{J}\expect_{n_{J},h+\bm{a}^\intercal \bm{h}*}\dbrack*{\ell_M\dparen*{\sqrt{n_{J}}\dparen*{\delta_{n_{J}} - \kappa\dparen*{\theta_{n_{J}}(h+\bm{a}^\intercal \bm{h})}}}},\\
        - &\tau_U^+(\theta_0)\limsup_{J}\expect_{n_{J},h+\bm{a}^\intercal \bm{h}*}\dbrack*{\ell_M\dparen[\Big]{\sqrt{n_{J}}\dparen*{\delta_{n_{J}} - \kappa\dparen*{\theta_{n_{J}}(h+\bm{a}^\intercal \bm{h})}}}}
    \end{aligned}
    }\nonumber\\
    &= \max\dbrace*{\tau_L^-(\theta_0)\int \ell_M\mathrm{d}L_{h+\bm{a}^\intercal \bm{h}},~-\tau_U^+(\theta_0)\int \ell_M\mathrm{d}L_{h+\bm{a}^\intercal \bm{h}}},
    \label{eq:rimr-lower-bound-ineq1}
\end{align}
where the last equality follows from the definition of the weak convergence and asymptotic measurability of $\sqrt{n_{J}}(\delta_{n_{J}} - \kappa(\theta_{n_{J}}(h + \bm{a}^\intercal \bm{h})))$ that is implied by its weak convergence to a tight Borel measure \parencite[][Lemma~1.3.8]{Vaart1996}. Note that this inequality holds for every $\bm{a} \in \mathbb{R}^m$. Thus, integrating both sides of \cref{eq:rimr-lower-bound-ineq1} with respect to $\bm{a} \sim N(0,\lambda^{-1}I_m)$  gives
\begin{align}
    R_M
    \geq & \int \max\dbrace*{\tau_L^-(\theta_0)\int \ell_M\mathrm{d}L_{h+\bm{a}^\intercal \bm{h}},~-\tau_U^+(\theta_0)\int \ell_M\mathrm{d}L_{h+\bm{a}^\intercal \bm{h}}} N(0,\lambda^{-1}I_m)(\mathrm{d}\bm{a})\nonumber\\
    \geq & \max\dbrace*{
    \begin{aligned}
        &\tau_L^-(\theta_0)\int \dparen*{\int \ell_M\mathrm{d}L_{h+\bm{a}^\intercal \bm{h}}}N(0,\lambda^{-1}I_m)(\mathrm{d}\bm{a}),\\
        -&\tau_U^+(\theta_0)\int \dparen*{\int \ell_M\mathrm{d}L_{h+\bm{a}^\intercal \bm{h}}}N(0,\lambda^{-1}I_m)(\mathrm{d}\bm{a}).
    \end{aligned}
    }
    \label{eq:rimr-lower-bound-ineq2}
\end{align}
\cref{lem:rimr-asymptotic-representation} posits that the law of $\int L_{h+\bm{a}^\intercal \bm{h}}N(0,\lambda^{-1}I_m)(\mathrm{d}\bm{a})$ is equal to the law of 
\begin{equation*}
    V_{\lambda,m} - \kappa_{\theta_0}'\dparen*{\mathbb{G}_{\lambda,m} + W_{\lambda,m} + \Dot{\theta}_0(h)} + \kappa_{\theta_0}'\dparen*{\Dot{\theta}_0(h)},    
\end{equation*}
where $\mathbb{G}_{\lambda,m} = \sum_{j=1}^m G_j\Dot{\theta}_0(h_j)/(1+\lambda)^{1/2}$ for $(G_1,\cdots,G_m) \sim N(0,I_m)$ and $(V_{\lambda,m},W_{\lambda,m}) \in \mathbb{R}\times\mathbb{B}$ is a tight random element independent of $G_{\lambda,m}$. Hence, the right-hand side of \cref{eq:rimr-lower-bound-ineq2} can be further bounded from below by 
\begin{align*}
    \MoveEqLeft
    \max\dbrace*{
        \begin{aligned}
            &\tau_L^-(\theta_0)\expect_h\dbrack*{\ell_M\dparen*{V_{\lambda,m} - \kappa_{\theta_0}'\dparen*{\mathbb{G}_{\lambda,m} + W_{\lambda,m} + \Dot{\theta}_0(h)} + \kappa_{\theta_0}'\dparen*{\Dot{\theta}_0(h)}}},\\
            -&\tau_U^+(\theta_0)\expect_h\dbrack*{\ell_M\dparen*{V_{\lambda,m} - \kappa_{\theta_0}'\dparen*{\mathbb{G}_{\lambda,m} + W_{\lambda,m} + \Dot{\theta}_0(h)} + \kappa_{\theta_0}'\dparen*{\Dot{\theta}_0(h)}}}
        \end{aligned}
    }\\
    = {} & \max\dbrace*{
        \begin{aligned}
            &\tau_L^-(\theta_0)\int \expect\dbrack*{\ell_M\dparen*{v - \kappa_{\theta_0}'\dparen*{\mathbb{G}_{\lambda,m} + w + \Dot{\theta}_0(h)} + \kappa_{\theta_0}'\dparen*{\Dot{\theta}_0(h)}}}Q_{h,\lambda,m}(\mathrm{d}(v,w)),\\
            -&\tau_U^+(\theta_0)\int \expect\dbrack*{\ell_M\dparen*{v - \kappa_{\theta_0}'\dparen*{\mathbb{G}_{\lambda,m} + w + \Dot{\theta}_0(h)} + \kappa_{\theta_0}'\dparen*{\Dot{\theta}_0(h)}}}Q_{h,\lambda,m}(\mathrm{d}(v,w)),
        \end{aligned}
    }
\end{align*}
where $Q_{h,\lambda,m}$ denotes the joint distribution of $(V_{\lambda,m},W_{\lambda,m})$ under $h \in T(P_0)$, and the last equality follows from the independence of $\mathbb{G}_{\lambda,m}$ and $(V_{\lambda,m},W_{\lambda,m})$. Note that the arguments thus far hold for each $h \in T(P_0)$. In particular, the inequalities hold for a countable dense subset $T_0 = \{t_1,t_2,\cdots\}$ of $T(P_0)$. Letting
\begin{equation*}
    \rho_{M,h,\lambda,m}(v,w) \coloneqq \expect\dbrack*{\ell_M\dparen*{v - \kappa_{\theta_0}'\dparen*{\mathbb{G}_{\lambda,m} + w + \Dot{\theta}_0(h)} + \kappa_{\theta_0}'\dparen*{\Dot{\theta}_0(h)}}},
\end{equation*}
the study obtains
\begin{align}
    \begin{split}
        R_M \geq & \max_{j \leq J}\max\dbrace*{
        \begin{aligned}
            &\tau_L^-(\theta_0) \int \rho_{M,t_j,\lambda,m}(v,w)Q_{t_j,\lambda,m}(\mathrm{d}(v,w)),\\
            -&\tau_U^+(\theta_0) \int \rho_{M,t_j,\lambda,m}(v,w)Q_{t_j,\lambda,m}(\mathrm{d}(v,w))
        \end{aligned}
        }
    \end{split}
    \label{eq:rimr-lower-bound-ineq3}
\end{align}
for all $J \in \mathbb{N}$.
\par
Next, apply the purification result in \textcite{Feinberg2006} to bound the right-hand side of \cref{eq:rimr-lower-bound-ineq3} from below by the function without integrals. To this end, first, define the state space, action space, decision rule, loss function, and risk function. For clarity, make the notation and terminology correspond as closely as possible to that of \textcite{Feinberg2006}. Let the state space $(X,\mathcal{X})$ be $([0,J],\mathcal{B}([0,J]))$, where $J \in \mathbb{N}$. The action space is specified as follows. Fix $\epsilon \in (0,1)$. Noting that $Q_{t_j,\lambda,m}$ is a tight Borel probability measure, it is possible to find a compact set $A_j \subset \mathbb{R}\times\mathbb{B}$ such that $Q_{t_j,\lambda,m}(A_j) \geq 1 - \epsilon$. Then, $A = \bigcup_{j=1}^J A_j$ is compact, which implies the completeness and separability of $A$. Define the action space as $(A,\mathcal{A})$, where $\mathcal{A}$ is the Borel $\sigma$-field obtained by taking the intersection of $A$ and each element of $\mathcal{B}(\mathbb{R}\times\mathbb{B})$.\footnote{In defining the action space, the study cannot use $(\mathbb{R}\times\mathbb{B},\mathcal{B}(\mathbb{R}\times\mathbb{B}))$, as $\mathbb{R}\times\mathbb{B}$ need not be separable. To circumvent this challenge, restrict the action space to the compact set constructed as in the proof of Theorem~1 in \textcite{Ponomarev2022}.} For each $j = 1,\dots,J$, define the Borel probability measure $\widetilde{Q}_{t_j,\lambda,m}$ on $(A,\mathcal{A})$ by 
\begin{equation*}
    \widetilde{Q}_{t_j,\lambda,m}(B) \coloneqq \frac{Q_{t_j,\lambda,m}(B \cap A)}{Q_{t_j,\lambda,m}(A)} \quad \text{for every} \quad B \in \mathcal{A}.
\end{equation*}
With this probability measure, 
\begin{align}
    \int \rho_{M,t_j,\lambda,m}\mathrm{d}Q_{t_j,\lambda,m} = Q_{t_j,\lambda,m}(A) \int \rho_{M,t_j,\lambda,m}\mathrm{d}\widetilde{Q}_{t_j,\lambda,m} + \int_{\mathbb{R}\times\mathbb{B}\setminus A} \rho_{M,t_j,\lambda,m}\mathrm{d}Q_{t_j,\lambda,m}.
    \label{eq:purification-equation1}
\end{align}
Let $\mu_k$ be the uniform distribution over the interval $[k-1,k]$ for $k = 1,\cdots,J$, and let $\Omega = \{\mu_1,\dots,\mu_J\}$ be the possible distributions over the state space $(X,\mathcal{X})$. Moreover, let $\pi$ be the decision rule defined by
\begin{equation*}
    \pi(B|x) \coloneqq \sum_{j=1}^J\widetilde{Q}_{t_j,\lambda,m}(B)\ind\{x \in [j-1,j]\}
\end{equation*}
for every $B \in \mathcal{A}$ and $x \in \mathcal{X}$. Define the loss by the vector function
\begin{equation*}
    \rho(\mu,x,a) = (\rho_1(\mu,x,a),\dots,\rho_J(\mu,x,a)) = (\rho_{M,t_1,\lambda,m}(a),\dots,\rho_{M,t_J,\lambda,m}(a))
\end{equation*}
for every $\mu \in \Omega$, $x \in X$, and $a = (v,w) \in A$. For the decision rule $\pi$ and $\mu_k \in \Omega$, define the risk vector by
\begin{equation*}
    \mathcal{R}(\mu_k,\pi) = (\mathcal{R}_1(\mu_k,\pi),\dots,\mathcal{R}_J(\mu_k,\pi)),
\end{equation*}
where 
\begin{equation*}
    \mathcal{R}_j(\mu_k,\pi) \coloneqq \int_X\int_A\rho_j(\mu_k,x,a)\pi(\mathrm{d}a|x)\mu_k(\mathrm{d}x) \quad \text{for every} \quad j = 1,\cdots,J.
\end{equation*}
Thus, $\mathcal{R}_j(\mu_j,\pi) = \int \rho_{M,t_j,\lambda,m} \mathrm{d}\widetilde{Q}_{t_j,\lambda,m}$ by construction.
\par
Now, Theorem~1 in \textcite{Feinberg2006} implies that there exists a measurable map $(v,w):X \to A$ such that
\begin{equation*}
    \mathcal{R}_j(\mu_k,\pi) = \int_X \rho_j(\mu_k,x,(v(x),w(x)))\mu_k(\mathrm{d}x) =  \int_{k-1}^k \rho_{M,t_j,\lambda,m}(v(x),w(x))\mathrm{d}x
\end{equation*}
for all $j = 1,\dots,J$ and $k = 1,\dots,J$.
Using this result and \cref{eq:purification-equation1}, it is possible to bound the right-hand side of \cref{eq:rimr-lower-bound-ineq3} from below by
\begin{align*}
    &\max_{j \leq J}\max\dbrace*{
    \begin{aligned}
        &\tau_L^-(\theta_0)Q_{t_j,\lambda,m}(A)\mathcal{R}_j(\mu_j,\pi) + \tau_L^-(\theta_0)\int_{\mathbb{R}\times\mathbb{B}\setminus A} \rho_{M,t_j,\lambda,m}\mathrm{d}Q_{t_j,\lambda,m},\\
        -&\tau_U^+(\theta_0)Q_{t_j,\lambda,m}(A)\mathcal{R}_j(\mu_j,\pi) - \tau_U^+(\theta_0)\int_{\mathbb{R}\times\mathbb{B}\setminus A} \rho_{M,t_j,\lambda,m}\mathrm{d}Q_{t_j,\lambda,m}
    \end{aligned}
    }\\
    &\geq \max_{j \leq J}\max\dbrace*{
    \begin{aligned}
        &\tau_L^-(\theta_0)Q_{t_j,\lambda,m}(A)\mathcal{R}_j(\mu_j,\pi) - M\tau_L^-(\theta_0)Q_{t_j,\lambda,m}(\mathbb{R}\times\mathbb{B}\setminus A),\\
        -&\tau_U^+(\theta_0)Q_{t_j,\lambda,m}(A)\mathcal{R}_j(\mu_j,\pi) - M\tau_U^+(\theta_0)Q_{t_j,\lambda,m}(\mathbb{R}\times\mathbb{B}\setminus A)
    \end{aligned}
    }\\
    &\geq \max_{j \leq J}\max\dbrace*{
        Q_{t_j,\lambda,m}(A)
        \max\dbrace*{
            \begin{aligned}
                &\tau_L^-(\theta_0)\mathcal{R}_j(\mu_j,\pi),\\
                -&\tau_U^+(\theta_0)\mathcal{R}_j(\mu_j,\pi)
            \end{aligned}
        }
        - M(y_U - y_L)Q_{t_j,\lambda,m}(\mathbb{R}\times\mathbb{B}\setminus A)
    }\\
    &\geq (1-\epsilon)\max_{j \leq J}\max\dbrace*{
    \begin{aligned}
        &\tau_L^-(\theta_0)\mathcal{R}_j(\mu_j,\pi),\\
        -&\tau_U^+(\theta_0)\mathcal{R}_j(\mu_j,\pi)
    \end{aligned}
    } - \epsilon M(y_U - y_L) \\
    &= (1-\epsilon)\max_{j \leq J}\max\dbrace*{
    \begin{aligned}
        &\tau_L^-(\theta_0)\int_{j-1}^j \rho_{M,t_j,\lambda,m}(v(x),w(x))\mathrm{d}x,\\
        -&\tau_U^+(\theta_0)\int_{j-1}^j \rho_{M,t_j,\lambda,m}(v(x),w(x))\mathrm{d}x
    \end{aligned}
    } - \epsilon M(y_U - y_L) \\
    &\geq (1-\epsilon)\adjustlimits{\inf}_{(v,w)\in\mathbb{R}\times\mathbb{B}}{\max}_{j \leq J}\max\dbrace*{
    \begin{aligned}
        &\tau_L^-(\theta_0)\int_{j-1}^j \rho_{M,t_j,\lambda,m}(v,w)\mathrm{d}x,\\
        -&\tau_U^+(\theta_0)\int_{j-1}^j \rho_{M,t_j,\lambda,m}(v,w)\mathrm{d}x
    \end{aligned}
    } - \epsilon M(y_U - y_L)\\
    &\geq (1-\epsilon)\adjustlimits{\inf}_{(v,w)\in\mathbb{R}\times\mathbb{B}}{\max}_{j \leq J}\max\dbrace*{
    \begin{aligned}
        &\tau_L^-(\theta_0) \rho_{M,t_j,\lambda,m}(v,w),\\
        -&\tau_U^+(\theta_0) \rho_{M,t_j,\lambda,m}(v,w)
    \end{aligned}
    } - \epsilon M(y_U - y_L),
\end{align*}
where the first inequality follows from $-M \leq \ell_M \leq M$, the second inequality follows from $y_L - y_U \leq \tau_j(\theta_0) \leq y_U - y_L$ for $j \in \{L,U\}$, the third inequality follows from $Q_{t_j,\lambda,m}(A) \geq Q_{t_j,\lambda,m}(A_j) \geq 1-\epsilon$ for all $j=1,\dots,J$, the first equality follows from the purification result, the fourth inequality follows as the range of $(v,w):X \to A$ is contained in $\mathbb{R}\times\mathbb{B}$, and the last inequality follows as the integrand does not depend on $x$. Moreover, the right-hand side of the above inequalities holds for every $\epsilon$. Letting $\epsilon \downarrow 0$, 
\begin{align}
    R_M \geq \adjustlimits{\inf}_{(v,w)\in\mathbb{R}\times\mathbb{B}}{\max}_{j \leq J}\max\dbrace*{ \tau_L^-(\theta_0) \rho_{M,t_j,\lambda,m}(v,w), -\tau_U^+(\theta_0) \rho_{M,t_j,\lambda,m}(v,w) }.
    \label{eq:rimr-lower-bound-ineq4}
\end{align}
\par
Recall that $\rho_{M,h,\lambda,m}(v,w)$ is given by 
\begin{equation*}
    \expect\dbrack*{\ell_M\dparen*{v - \kappa_{\theta_0}'\dparen*{\mathbb{G}_{\lambda,m} + w + \Dot{\theta}_0(h)} + \kappa_{\theta_0}'\dparen*{\Dot{\theta}_0(h)}}},
\end{equation*}
where $\mathbb{G}_{\lambda,m} = \sum_{j=1}^m G_j\Dot{\theta}_0(h_j)/(1+\lambda)^{1/2}$ with $(G_1,\cdots,G_m) \sim N(0,I_m)$. As $\lambda \downarrow 0$, $\mathbb{G}_{\lambda,m} \leadsto \mathbb{G}_m = \sum_{j=1}^m G_j\Dot{\theta}_0(h_j)$, and, for any $b^* \in \mathbb{B}^*$, the variable $b^*\mathbb{G}_m = \sum_{j=1}^m G_j b^*\Dot{\theta}_0(h_j)$ is normally distributed with $\expect [b^*\mathbb{G}_m] = 0$ and $\expect [(b^*\mathbb{G}_m)^2] = \sum_{j=1}^m \dangle{ \Dot{\theta}_0^*b^*, h_j }^2$. As $\{h_1,h_2,\cdots\}$ is a complete orthonormal basis, the Parseval's identity implies $\expect [b^*\mathbb{G}_m^2] \to \norm{\Dot{\theta}_0^*b^*}_{2,P_0}^2$ as $m \to \infty$, which followed by $b^*\mathbb{G}_m \leadsto N(0,\norm{\Dot{\theta}_0^*b^*}_{2,P_0}^2)$. Moreover, as demonstrated in the proof of Theorem~3.11.2 in \textcite{Vaart1996}, $\mathbb{G}_m,m \in \mathbb{N}$ is uniformly tight, where it has a subsequence $\{m_k\}_k$ such that $\mathbb{G}_{m_k} \leadsto \mathbb{G}_0$ as $k \to \infty$. As the subsequence $b^*\mathbb{G}_{m_k}$ also weakly converges to the same weak limit, $\mathbb{G}_0$ satisfies $b^*\mathbb{G}_0 \sim N(0,\norm{\Dot{\theta}_0^*b^*}_{2,P_0}^2)$ for every $b^* \in \mathbb{B}^*$. Moreover, the law of $\mathbb{G}_0$ is concentrated on the closure of $\Dot{\theta}_0(T(P_0))$.
\par
Define the map $f_{M}(\cdot,v,w,s):\mathbb{B}\to\mathbb{R}$ by $f_{M}(g,v,w,s) \coloneqq \ell_M(v - \kappa_{\theta_0}'(g + w + s) + \kappa_{\theta_0}'(s))$ for every $(v,w,s) \in \mathbb{R}\times\mathbb{B}\times\mathbb{B}$. \cref{lem:lipschitzness-of-bounded-loss} shows that $f_M$ is bounded and Lipschitz continuous with its Lipschitz constant not depending on the choice of $(v,w,s)$. Hence, the collection of functions $\dbrace*{f_M(g,v,w,s) : (v,w,s) \in \mathbb{R}\times\mathbb{B}\times\mathbb{B}}$ is bounded and uniformly equicontinuous. Then, Theorem~1.12.1 in \textcite{Vaart1996} implies that 
\begin{equation*}
    \sup_{(v,w,s) \in \mathbb{R}\times\mathbb{B}\times\mathbb{B}} \abs*{\expect\dbrack*{f_M(\mathbb{G}_{\lambda,m},v,w,s)} - \expect\dbrack*{f_M(\mathbb{G}_m,v,w,s)}} \to 0 \quad \text{as} \quad \lambda \downarrow 0.
\end{equation*}
Note that \cref{eq:rimr-lower-bound-ineq4} holds for arbitrary $\lambda$ and $m$. The preceding uniform convergence allows for interchanging limit and infimum; hence, letting $\lambda \downarrow 0$ on both sides of \cref{eq:rimr-lower-bound-ineq4} gives
\begin{align*}
    R_M
    &\geq \adjustlimits{\lim}_{\lambda \downarrow 0}{\inf}_{(v,w)\in\mathbb{R}\times\mathbb{B}}\max_{j \leq J} \max\dbrace*{
    \begin{aligned}
        &\tau_L^-(\theta_0)\expect\dbrack*{f_M\dparen*{\mathbb{G}_{\lambda,m},v,w,\Dot{\theta}_0(t_j)}},\\
        -&\tau_U^+(\theta_0)\expect\dbrack*{f_M\dparen*{\mathbb{G}_{\lambda,m},v,w,\Dot{\theta}_0(t_j)}}
    \end{aligned}
    }\nonumber\\
    & = \adjustlimits{\inf}_{(v,w)\in\mathbb{R}\times\mathbb{B}}{\max}_{j \leq J} \max\dbrace*{
    \begin{aligned}
        &\tau_L^-(\theta_0)\lim_{\lambda \downarrow 0}\expect\dbrack*{f_M\dparen*{\mathbb{G}_{\lambda,m},v,w,\Dot{\theta}_0(t_j)}},\\
        -&\tau_U^+(\theta_0)\lim_{\lambda \downarrow 0}\expect\dbrack*{f_M\dparen*{\mathbb{G}_{\lambda,m},v,w,\Dot{\theta}_0(t_j)}}
    \end{aligned}
    }\nonumber\\
    & = \adjustlimits{\inf}_{(v,w)\in\mathbb{R}\times\mathbb{B}}{\max}_{j \leq J} \max\dbrace*{
    \begin{aligned}
        &\tau_L^-(\theta_0)\expect\dbrack*{f_M\dparen*{\mathbb{G}_m,v,w,\Dot{\theta}_0(t_j)}},\\
        -&\tau_U^+(\theta_0)\expect\dbrack*{f_M\dparen*{\mathbb{G}_m,v,w,\Dot{\theta}_0(t_j)}}
    \end{aligned}
    }.
\end{align*}
A similar argument further shows that 
\begin{equation*}
    \sup_{(v,w,s) \in \mathbb{R}\times\mathbb{B}\times\mathbb{B}} \abs*{\expect\dbrack*{f_M(\mathbb{G}_{m_k},v,w,s)} - \expect\dbrack*{f_M(\mathbb{G}_0,v,w,s)}} \to 0 \quad \text{as} \quad k \to \infty,
\end{equation*}
and 
\begin{align}
    R_M
    \geq \adjustlimits{\inf}_{(v,w)\in\mathbb{R}\times\mathbb{B}}{\max}_{j \leq J} \max\dbrace*{
    \begin{aligned}
        &\tau_L^-(\theta_0)\expect\dbrack*{f_M\dparen*{\mathbb{G}_0,v,w,\Dot{\theta}_0(t_j)}},\\
        -&\tau_U^+(\theta_0)\expect\dbrack*{f_M\dparen*{\mathbb{G}_0,v,w,\Dot{\theta}_0(t_j)}}
    \end{aligned}
    }.
    \label{eq:rimr-lower-bound-ineq5}
\end{align}
\par
Write $x = (v,w)$ and $\mathcal{X} = \mathbb{R}\times\mathbb{B}$ for ease of notation. Further, let  
\begin{align*}
    F_{M,J}(x) &\coloneqq \max_{j \leq J} \max\dbrace*{
    \begin{aligned}
        &\tau_L^-(\theta_0)\expect\dbrack*{f_M\dparen*{\mathbb{G}_0,v,w,\Dot{\theta}_0(t_j)}},\\
        -&\tau_U^+(\theta_0)\expect\dbrack*{f_M\dparen*{\mathbb{G}_0,v,w,\Dot{\theta}_0(t_j)}}
    \end{aligned}},\\
    F_{M}(x) &\coloneqq \sup_{j \in \mathbb{N}} \max\dbrace*{
    \begin{aligned}
        &\tau_L^-(\theta_0)\expect\dbrack*{f_M\dparen*{\mathbb{G}_0,v,w,\Dot{\theta}_0(t_j)}},\\
        -&\tau_U^+(\theta_0)\expect\dbrack*{f_M\dparen*{\mathbb{G}_0,v,w,\Dot{\theta}_0(t_j)}}
    \end{aligned}},\\
    F(x) &\coloneqq \sup_{j \in \mathbb{N}} \max\dbrace*{
    \begin{aligned}
        &\tau_L^-(\theta_0)\expect\dbrack*{f\dparen*{\mathbb{G}_0,v,w,\Dot{\theta}_0(t_j)}},\\
        -&\tau_U^+(\theta_0)\expect\dbrack*{f\dparen*{\mathbb{G}_0,v,w,\Dot{\theta}_0(t_j)}}
    \end{aligned}}.
\end{align*}
Since $\mathcal{X}$ is completely regular, the Stone-\v{C}ech compactification gives the compact Hausdorff space $\beta \mathcal{X}$ such that $\mathcal{X}$ is a dense subspace of $\beta\mathcal{X}$, and every bounded continuous function $f:\mathcal{X}\to\mathbb{R}$ can be extended uniquely to a continuous function $f^\beta:\beta\mathcal{X}\to\mathbb{R}$. Thus, for each $J \in \mathbb{N}$, $F_{M,J}$ can be extended uniquely to $F_{M,J}^\beta$ \parencite[see, e.g.,][Theorem~38.2]{Munkres2013}. As implied by \cref{lem:lipschitzness-of-bounded-loss}, the sequence $\{F_{M,J}^\beta\}_{J \in \mathbb{N}}$ is equicontinuous and pointwise convergent on $\mathcal{X}$. These properties then imply that $\{F_{M,J}^\beta\}_{J \in \mathbb{N}}$ uniformly converges to a continuous function $\beta\mathcal{X}$. Specifically, the limit $\tilde{F}_M = \lim_{J \to \infty} F_{M,J}^\beta$ coincides with $F_M^\beta$, which is the unique extension of $F_M$ over $\beta\mathcal{X}$ because $\tilde{F}_M(x) = F_M(x)$ for all $x \in \mathcal{X}$ and $\tilde{F}_M$ is continuous over $\beta\mathcal{X}$. Then, letting $J \to \infty$ on both sides of \cref{eq:rimr-lower-bound-ineq5} yields
\begin{equation*}
    R_M \geq \adjustlimits{\lim}_{J \to \infty}{\inf}_{x \in \mathcal{X}} F_{M,J}(x) = \adjustlimits{\lim}_{J\to\infty}{\inf}_{x \in \beta\mathcal{X}} F_{M,J}^\beta(x) = \inf_{x \in \beta\mathcal{X}} F_M^\beta(x) = \inf_{x \in \mathcal{X}} F_M(x),
\end{equation*}
where the first and third equalities follow from the denseness of $\mathcal{X}$ in $\beta\mathcal{X}$ and the definition of the extension. Next, consider the sequence $\{F_M\}_{M \in \mathbb{N}}$. A similar argument as above implies that the sequence $\{F_M^\beta\}_{M \in \mathbb{N}}$ uniformly converges to $F^\beta$ over $\beta\mathcal{X}$. Taking limit inferior with respect to $M$ on the most left- and right-hand sides of the preceding display further gives
\begin{equation*}
    \liminf_{M\to\infty} R_M \geq \adjustlimits{\lim}_{M \to \infty}{\inf}_{x \in \mathcal{X}} F_{M}(x) = \adjustlimits{\lim}_{M\to\infty}{\inf}_{x \in \beta\mathcal{X}} F_{M}^\beta(x) = \inf_{x \in \beta\mathcal{X}} F^\beta(x) = \inf_{x \in \mathcal{X}} F(x).
\end{equation*}
The denseness of $T_0$ in $T(P_0)$ also implies
\begin{equation*}
    F(x) = \sup_{h \in T(P_0)} \max\dbrace*{
    \begin{aligned}
        &\tau_L^-(\theta_0)\expect\dbrack*{f\dparen*{\mathbb{G}_0,v,w,\Dot{\theta}_0(h)}},\\
        -&\tau_U^+(\theta_0)\expect\dbrack*{f\dparen*{\mathbb{G}_0,v,w,\Dot{\theta}_0(h)}}
    \end{aligned}}.
\end{equation*}
Combining the observations thus far yields
\begin{align}
    &\liminf_{M \to \infty} R_M\nonumber\\
    &\geq \adjustlimits\inf_{(v,w) \in \mathbb{R}\times\mathbb{B}}\sup_{h \in T(P_0)}\max\dbrace*{
    \begin{aligned}
        &\tau_L^-(\theta_0)\expect\dbrack*{v - \kappa_{\theta_0}'\dparen*{\mathbb{G}_0 + w + \Dot{\theta}_0(h)} + \kappa_{\theta_0}\dparen*{ \Dot{\theta}_0(h)}},\\
        -&\tau_U^+(\theta_0)\expect\dbrack*{v - \kappa_{\theta_0}'\dparen*{\mathbb{G}_0 + w + \Dot{\theta}_0(h)} + \kappa_{\theta_0}\dparen*{ \Dot{\theta}_0(h)}}.
    \end{aligned}
    }
    \label{eq:rimr-lower-bound-ineq6}
\end{align}
\par
Finally, the study confirms that the right-hand sides of \cref{eq:rimr-lower-bound,eq:rimr-lower-bound-ineq6} are the same. For every $(v,w) \in \mathbb{R}\times\mathbb{B}$, there exists $(\bar{v},w) \in \mathbb{R}\times\mathbb{B}$ such that
\begin{equation*}
    \expect\dbrack*{v - \kappa_{\theta_0}'\dparen*{\mathbb{G}_0 + w + \Dot{\theta}_0(h)} + \kappa_{\theta_0}\dparen*{ \Dot{\theta}_0(h)}}
    =
    \expect\dbrack*{ \kappa_{\theta_0}'\dparen*{\mathbb{G}_0 + w + \Dot{\theta}_0(h) } - \kappa_{\theta_0}'\dparen*{ \Dot{\theta}_0(h) } + \Bar{v} }.
\end{equation*}
for all $h \in T(P_0)$. \cref{lemmaenum:translation-equivariance} further implies the existence of $(\tilde{v},w) \in \mathbb{B}\times\mathbb{B}$ with 
\begin{equation*}
    \expect\dbrack*{ \kappa_{\theta_0}'\dparen*{\mathbb{G}_0 + w + \Dot{\theta}_0(h) } - \kappa_{\theta_0}'\dparen*{ \Dot{\theta}_0(h) } + \Bar{v} }
    = 
    \expect\dbrack*{ \kappa_{\theta_0}'\dparen*{\mathbb{G}_0 + w + \tilde{v} + \Dot{\theta}_0(h) } - \kappa_{\theta_0}'\dparen*{ \Dot{\theta}_0(h) } }
\end{equation*}
for all $h \in T(P_0)$. Hence, the infimum with respect to $(v,w) \in \mathbb{R}\times\mathbb{B}$ in \cref{eq:rimr-lower-bound-ineq6} and the infimum with respect to $w \in \mathbb{B}$ in \cref{eq:rimr-lower-bound} are equivalent.
\end{proof}
\begin{proof}[Proof of \cref{thm:lam-with-respect-to-rimr}]
Let 
\begin{align*}
    \hat{F}_{M,n}(w) &= 
    \sup_{s \in \hat{S}_n} \max\dbrace*{
    \begin{aligned}
        &\tau_L^-(\hat{\theta}_n)\expect\dbrack*{ \ell_M\dparen*{ \hat{\kappa}_n'\dparen*{ \hat{\mathbb{G}}_n + w + s } - \hat{\kappa}_n'(s) } \middle| \mathbf{Z}_n},\\
        -&\tau_U^+(\hat{\theta}_n)\expect\dbrack*{ \ell_M\dparen*{ \hat{\kappa}_n'\dparen*{ \hat{\mathbb{G}}_n + w + s } - \hat{\kappa}_n'(s) } \middle| \mathbf{Z}_n }
    \end{aligned}
    }\\
    F_{M}(w) &= 
    \sup_{s \in S(\mathbb{G}_0)} \max\dbrace*{
    \begin{aligned}
        &\tau_L^-(\theta_0)\expect\dbrack*{ \ell_M\dparen*{ \kappa_{\theta_0}'\dparen*{ \mathbb{G}_0 + w + s } - \kappa_{\theta_0}'(s) } },\\
        -&\tau_U^+(\theta_0)\expect\dbrack*{ \ell_M\dparen*{ \kappa_{\theta_0}'\dparen*{ \mathbb{G}_0 + w + s } - \kappa_{\theta_0}'(s) } }.
    \end{aligned}
    }.
\end{align*}
Then, $\hat{w}_n$ given in \cref{eq:optimal-adjustment-term} is a minimizer of $\hat{F}_{M,n}(w)$ over $K_m \subset \mathbb{B}$ by definition. \cref{lem:pw-consistency-adjustment-terms} implies that there exists a minimizer $w_{M,m}^*$ of $F_M(w)$ such that $\hat{w}_n$ converges in probability to $w_{M,m}^*$ for every $h \in T(P_0)$. Combining this consistency with the best regularity of $\hat{\theta}_n$ and \cref{asum:theta-differentiable}, 
\begin{align*}
    \sqrt{n}\dparen*{ \hat{\theta}_n + \frac{\hat{w}_n}{\sqrt{n}} - \theta_0 } 
    &= \sqrt{n}\dparen*{ \hat{\theta}_n - \theta_n(h) } + \hat{w}_n + \sqrt{n}\dparen*{ \theta_n(h) - \theta_0  }\\
    &\overset{h}{\leadsto} \mathbb{G}_0 + w_{M,m}^* + \Dot{\theta}_0(h)
\end{align*}
for every $h$. Then, the generalized delta method for Hadamard directionally differentiable functionals \parencite[e.g.,][Theorem~2.1]{Fang2019} gives
\begin{align*}
    \sqrt{n}\dparen*{ \delta_n^{\text{LAM}} - \kappa\dparen*{\theta_n(h)} }
    &= \sqrt{n}\dparen*{ \kappa\dparen*{ \hat{\theta}_n + \frac{\hat{w}_n}{\sqrt{n}} } - \kappa\dparen*{\theta_n(h)}}\\
    &= \sqrt{n}\dparen*{ \kappa\dparen*{ \hat{\theta}_n + \frac{\hat{w}_n}{\sqrt{n}} } - \kappa(\theta_0) } - \sqrt{n}\dparen*{ \kappa(\theta_n(h)) - \kappa(\theta_0) }\\
    &\overset{h}{\leadsto} \kappa_{\theta_0}'\dparen*{ \mathbb{G}_0 + w_{M,m}^* + \Dot{\theta}_0(h) } - \kappa_{\theta_0}'\dparen*{ \Dot{\theta}_0(h) }.
\end{align*}
Since $\ell_M$ is a continuous and bounded function, the definition of the weak convergence and the asymptotic measurability of $\sqrt{n}\dparen{\delta_n^{\text{LAM}} - \kappa(\theta_n(h))}$ yields
\begin{align*}
    \MoveEqLeft
    \adjustlimits{\lim}_{n \to \infty}{\sup}_{h \in I} \sqrt{n}\rimr_{M,n}(\delta_n^{\text{LAM}},h)\\
    = {} & \sup_{h \in I}\max\dbrace*{
    \begin{aligned}
        &\lim_{n \to \infty} \tau_L^-(\theta_n(h)) \lim_{n \to \infty} \expect_{n,h*}\dbrack*{ \ell_M\dparen*{ \sqrt{n}\dparen*{ \delta_n^{\text{LAM}} - \kappa(\theta_n(h)) } } },\\
        -&\lim_{n \to \infty}\tau_U^+(\theta_n(h)) \lim_{n \to \infty} \expect_{n,h*}\dbrack*{ \ell_M\dparen*{ \sqrt{n}\dparen*{ \delta_n^{\text{LAM}} - \kappa(\theta_n(h)) } } }
    \end{aligned}
    }\\
    = {} & \sup_{h \in I}\max\dbrace*{
    \begin{aligned}
        &\tau_L^-(\theta_0)\expect\dbrack*{ \ell_M\dparen*{ \kappa_{\theta_0}'\dparen*{ \mathbb{G}_0 + w_{M,m}^* + \Dot{\theta}_0(h) } - \kappa_{\theta_0}'\dparen*{ \Dot{\theta}_0(h) }} },\\
        -&\tau_U^+(\theta_0)\expect\dbrack*{ \ell_M\dparen*{ \kappa_{\theta_0}'\dparen*{ \mathbb{G}_0 + w_{M,m}^* + \Dot{\theta}_0(h) } - \kappa_{\theta_0}'\dparen*{ \Dot{\theta}_0(h) }} }
    \end{aligned}
    }
\end{align*}
for every finite subset $I$ of $T(P_0)$. Then, by taking the supremum with respect to $I \subset T(P_0)$ for both sides of the above equality, 
\begin{align*}
    \MoveEqLeft
    \adjustlimits{\sup}_{I}{\liminf}_{n \to \infty}\sup_{h \in I} \sqrt{n}\rimr_{M,n}(\delta_n^{\text{LAM}},h)\\
    = {} & \sup_{h \in T(P_0)} \max\dbrace*{
    \begin{aligned}
        &\tau_L^-(\theta_0)\expect\dbrack*{ \ell_M\dparen*{ \kappa_{\theta_0}'\dparen*{ \mathbb{G}_0 + w_{M,m}^* + \Dot{\theta}_0(h) } - \kappa_{\theta_0}'\dparen*{ \Dot{\theta}_0(h) }} },\\
        -&\tau_U^+(\theta_0)\expect\dbrack*{ \ell_M\dparen*{ \kappa_{\theta_0}'\dparen*{ \mathbb{G}_0 + w_{M,m}^* + \Dot{\theta}_0(h) } - \kappa_{\theta_0}'\dparen*{ \Dot{\theta}_0(h) }} }
    \end{aligned}
    }\\
    = {} & \inf_{(v,w) \in K_m} F_M(w).
\end{align*}
Let 
\begin{align*}
    F(w) = \sup_{s \in S(\mathbb{G}_0)} \max\dbrace*{
    \begin{aligned}
        &\tau_L^-(\theta_0)\expect\dbrack*{ \ell\dparen*{ \kappa_{\theta_0}'\dparen*{ \mathbb{G}_0 + w + s } - \kappa_{\theta_0}'(s)} },\\
        -&\tau_U^+(\theta_0)\expect\dbrack*{ \ell\dparen*{ \kappa_{\theta_0}'\dparen*{ \mathbb{G}_0 + w + s } - \kappa_{\theta_0}'(s)} }
    \end{aligned}
    }.
\end{align*}
By the Lipschitz continuity of $\ell_M$ given in \cref{lem:lipschitzness-of-bounded-loss}, $\{F_M\}_{M=1}^\infty$ is an equicontinuous sequence on $K_m$. Moreover, for each $w \in K_m$, $F_M(w)$ pointwisely converges to $F(w)$. These observations imply the uniform convergence of $F_M$ to $F$ over $K_m$, which allows for exchanging the limit and infimum to obtain
\begin{equation*}
    \adjustlimits{\lim}_{M \to \infty}{\inf}_{w\in K_m} F_M(w) = \inf_{w\in K_m} F(w).
\end{equation*}
Then, letting $m$ approach infinity on both sides of the preceding equality concludes the proof.
\end{proof}
\section{Auxiliary Results}\label{sec:auxiliary-results}
\begin{lemma}\label{lem:rimr-asymptotic-representation}
    Suppose that the conditions in \cref{thm:rimr-lower-bound} hold. Let $\bm{h} = \{h_1,\cdots,h_m\}$ be the first $m$ elements of a complete orthonormal basis of $\overline{T(P_0)}$. For any $h \in T(P_0)$ and for arbitrary subnet $\{n_{i}\}_{i \in I}$ of $\{n\}$, there exists a further subnet $\{n_{j}\}_{j \in J}$ such that 
    \begin{equation*}
        \sqrt{n_{j}}\dparen*{ \delta_{n_j} - \kappa(\theta_{n_j}\dparen*{ h + \bm{a}^\intercal \bm{h} }) } \overset{h + \bm{a}^\intercal \bm{h}}{\leadsto} L_{h + \bm{a}^\intercal \bm{h}} \quad \text{for every} \quad \bm{a} = (a_1,\cdots,a_m) \in \mathbb{R}^m,
    \end{equation*}
    where $L_{h + \bm{a}^\intercal \bm{h}}$ is a tight Borel probability measure on $\mathbb{R}$. Moreover, 
    \begin{equation*}
        \int L_{h + \bm{a}^\intercal \bm{h}}(B)N(0,\lambda^{-1}I_m)(\mathrm{d}\bm{a}) = 
        P_h\dparen*{V_{\lambda,m} - \kappa_{\theta_0}'\dparen*{ G_{\lambda,m} + W_{\lambda,m} + \Dot{\theta}_0(h) } + \kappa_{\theta_0}'\dparen*{ \Dot{\theta}_0(h) } \in B},
    \end{equation*}
    where $G_{\lambda,m} = \sum_{j=1}^m A_j\Dot{\theta}_0(h_j)/(1+\lambda)^{1/2}$ for $(A_1,\cdots,A_m) \sim N(0,I_m)$, and $(V_{\lambda,m},W_{\lambda,m}) \in \mathbb{R}\times \mathbb{B}$ is a tight random element independent of $G_{\lambda,m}$.
\end{lemma}

\begin{proof}[Proof of \cref{lem:rimr-asymptotic-representation}]
    Since $\sqrt{n}\dparen{\delta_n - \kappa(\theta_0)}$ is asymptotically tight and asymptotically measurable under $\mathbb{P}_{n,0}$, its subnet $\sqrt{n_{i}}\dparen{\delta_{n_{i}} - \kappa(\theta_0)}$ is also asymptotically tight and asymptotically measurable under $\mathbb{P}_{n_{i},0}$. Moreover, as \cref{eq:loglikelihood-expansion} implies that $\Delta_{n,h} = n^{-1/2}\sum_{i=1}^n h(Z_i) \leadsto \Delta_h \sim N(0,\norm{h}_{2,P_0}^2)$, its subnet $\Delta_{n_{i},h}$ also weakly converges to the same limit law under $\mathbb{P}_{n_{i},0}$, whence $\Delta_{n_{i},h}$ is asymptotically tight and asymptotically measurable \parencite[][Lemma~1.3.8]{Vaart1996}. Combining these observations results in the asymptotic tightness and asymptotic measurability of $(\sqrt{n_{i}}\dparen{\delta_{n_{i}} - \kappa(\theta_0)}, \Delta_{n_{i},h})$ under $\mathbb{P}_{n_{i},0}$ \parencite[][Lemmas~1.4.3~and~1.4.4]{Vaart1996}. Then, Prohorov's theorem ensures the existence of the further subnet $\{n_{i_2}\}_{i_2 \in I_2}$ such that $(\sqrt{n_{i_2}}\dparen{\delta_{n_{i_2}} - \kappa(\theta_0)}, \Delta_{n_{i_2},h})$ jointly converges in distribution to a tight Borel law in $\mathbb{R}\times\mathbb{R}$ under $\mathbb{P}_{n_{i_2},0}$. Moreover, the linear expansion in \cref{eq:loglikelihood-expansion} and the Slutsky's lemma also implies that $(\sqrt{n_{i_2}}\dparen{\delta_{n_{i_2}} - \kappa(\theta_0)}, \mathrm{d}\mathbb{P}_{n_{i_2},h}/\mathrm{d}\mathbb{P}_{n_{i_2},0})$ converges weakly to a tight Borel law in $\mathbb{R}\times\mathbb{R}$ under $\mathbb{P}_{n_{i_2},0}$. Le Cam's third lemma concludes that $\sqrt{n_{i_2}}\dparen{\delta_{n_{i_2}} - \kappa(\theta_0)}$ converges weakly to a tight Borel law in $\mathbb{R}$ under $\mathbb{P}_{n_{i_2},h}$ for every $h \in T(P_0)$, from which follows that the asymptotic tightness and asymptotic measurability of $\sqrt{n_{i_2}}\dparen{\delta_{n_{i_2}} - \kappa(\theta_0)}$ under $\mathbb{P}_{n_{i_2},h}$ for every $h$ \parencite[][Lemma~1.3.8]{Vaart1996}.
    \par
    Fix $h \in T(P_0)$ arbitrarily. By \cref{asum:theta-differentiable}, $\theta_n(h) = \theta_0 + \Dot{\theta}_0(h)/\sqrt{n} + o(1/\sqrt{n})$. Moreover, by \cref{lemmaenum:directional-differentiability}, $\sqrt{n}\dparen{\kappa(\theta_n(h)) - \kappa(\theta_0)} \to C \coloneqq \kappa_{\theta_0}'(\Dot{\theta}_0(h))$ as $n \to \infty$. Noting that 
    \begin{equation*}
        \sqrt{n_{i_2}}\dparen*{ \delta_{n_{i_2}} - \kappa\dparen*{\theta_{n_{i_2}}(h)}} = \sqrt{n_{i_2}}\dparen*{ \delta_{n_{i_2}} - \kappa(\theta_0) } - \sqrt{n_{i_2}} \dparen*{ \kappa\dparen*{ \theta_{n_{i_2}}(h) } - \kappa(\theta_0) },
    \end{equation*}
    $\sqrt{n_{i_2}}\dparen{\delta_{n_{i_2}} - \kappa(\theta_{n_{i_2}}(h))}$ is also asymptotically tight and asymptotically measurable under $\mathbb{P}_{n_{i_2},h}$ \parencite[][Lemma~A.6]{Ponomarev2022}. Further, Lemma~A.3 in \textcite{Ponomarev2022} gives for every $\bm{a} \in \mathbb{R}^m$,
    \begin{equation*}
        \diff{\mathbb{P}_{n,h+\bm{a}^\intercal\bm{h}}}{\mathbb{P}_{n,h}} \overset{h}{\leadsto} \exp\left(\bm{a}^\intercal \bm{\Delta} - \frac{1}{2}\bm{a}^\intercal \bm{a}\right),
    \end{equation*}
    where $\mathbf{\Delta} = (\Delta_{h_1},\cdots,\Delta_{h_m}) \sim N(0,I_m)$. This immediately implies that $\mathrm{d}\mathbb{P}_{n_{i_2},h+\bm{a}^\intercal\bm{h}}/\mathrm{d}\mathbb{P}_{n_{i_2},h}$ is also asymptotically tight and asymptotically measurable under $\mathbb{P}_{n_{i_2},h}$ \parencite[][Lemma~1.3.8]{Vaart1996}. Combining these observations with Lemmas~1.4.3 and 1.4.4 in \textcite{Vaart1996} and Prohorov's theorem ensures the existence of a further subnet $\{n_{i_3}\} \subset \{n_{i_2}\}$ such that 
    \begin{equation*}
        \displaystyle\begin{pmatrix}
            \sqrt{n_{i_3}}\dparen*{ \delta_{n_{i_3}} - \kappa\dparen*{ \theta_{n_{i_3}}(h) } }\\
            \mathrm{d}\mathbb{P}_{n_{i_3},h+\bm{a}^\intercal\bm{h}}/\mathrm{d}\mathbb{P}_{n_{i_3},h}
        \end{pmatrix}
        \overset{h}{\leadsto}
        \begin{pmatrix}
            V\\
            \exp\dparen*{ \bm{a}^\intercal \mathbf{\Delta} - \frac{1}{2}\bm{a}^\intercal \bm{a} }
        \end{pmatrix}
    \end{equation*}
    for any $\bm{a} \in \mathbb{R}^m$. Note that, in the above display, $V$ has a tight Borel law that depends on $h$.\par
    Next, observe that 
    \begin{align*}
        \MoveEqLeft
        \sqrt{n_{i_3}}\dparen*{ \delta_{n_{i_3}} - \kappa\dparen*{ \theta_{n_{i_3}}\dparen{ h+\bm{a}^\intercal\bm{h} } } }\\
        = {} & \sqrt{n_{i_3}}\dparen*{ \delta_{n_{i_3}} - \kappa\dparen*{ \theta_{n_{i_3}}(h) } } - \sqrt{n_{i_3}}\dparen*{ \kappa\dparen*{ \theta_{n_{i_3}}(h+\bm{a}^\intercal\bm{h}) } - \kappa(\theta_0) } + \sqrt{n_{i_3}} \dparen*{ \kappa\dparen*{ \theta_{n_{i_3}}(h) } - \kappa(\theta_0) }.
    \end{align*}
    For the second term of the right-hand side of the above equation, 
    \begin{equation*}
        \sqrt{n_{i_3}} \dparen*{ \kappa\dparen*{ \theta_{n_{i_3}}(h+\bm{a}^\intercal\bm{h}) } - \kappa(\theta_0) } \to \kappa_{\theta_0}'\dparen*{ \Dot{\theta}_0(h) + \bm{a}^\intercal \bm{\Dot{\theta}_0(\bm{h})} }.
    \end{equation*}
    where $\bm{\Dot{\theta}_0(h)} = (\Dot{\theta}_0(h_1),\dots,\Dot{\theta}_0(h_m))^\intercal$. Slutsky's lemma implies that 
    \begin{equation*}
        \begin{pmatrix}
            \sqrt{n_{i_3}}\dparen*{ \delta_{n_{i_3}} - \kappa\dparen*{ \theta_{n_{i_3}}\dparen*{ h+\bm{a}^\intercal\bm{h} } } }\\
            \mathrm{d}\mathbb{P}_{n_{i_3},h+\bm{a}^\intercal\bm{h}}/\mathrm{d}\mathbb{P}_{n_{i_3},h}
        \end{pmatrix}
        \overset{h}{\leadsto}
        \begin{pmatrix}
            V - \kappa_{\theta_0}'\dparen*{ \Dot{\theta}_0(h)  + \bm{a}^\intercal \bm{\Dot{\theta}_0(\bm{h})} } + C\\
            \exp\dparen*{ \bm{a}^\intercal \mathbf{\Delta} - \frac{1}{2}\bm{a}^\intercal \bm{a} }
        \end{pmatrix}.
    \end{equation*}
    Then, Le Cam's third lemma gives the first statement of the lemma, in which $L_{h + \bm{a}^\intercal \bm{h}}$ is defined by 
    \begin{align*}
        \MoveEqLeft
        L_{h+\bm{a}^\intercal\bm{h}}(B)
        \coloneqq \expect_h\dbrack*{\ind_B\dparen*{V - \kappa_{\theta_0}'\dparen*{\Dot{\theta}_0(h) + \bm{a}^\intercal \bm{\Dot{\theta}_0(h)} } + C}\exp\dparen*{\bm{a}^\intercal \mathbf{\Delta} - \frac{1}{2}\bm{a}^\intercal \bm{a}}}
    \end{align*}
    for each $B \in \mathcal{B}(\mathbb{R})$, where $1_B(x) = 1\{x \in B\}$.
    \par
    Next is the proof of the second statement. Integrating both sides of the preceding display with respect to $\bm{a} \sim N(0,\lambda^{-1}I_m)$ and applying the Fubini's theorem yields
    \begin{align*}
        &\int L_{h+\bm{a}^\intercal\bm{h}}(B)N(0,\lambda^{-1}I_m)(\text{d}\bm{a})\\
        &= \int \expect_h\Biggl[\ind_B\dparen*{V - \kappa_{\theta_0}'\dparen*{\Dot{\theta}_0(h) + \bm{a}^\intercal \bm{\Dot{\theta}_0(h)} } + C}
        \exp\dparen*{\bm{a}^\intercal \mathbf{\Delta} - \frac{1}{2}\bm{a}^\intercal \bm{a}}\frac{\lambda^{m/2}}{\sqrt{(2\pi)^m}}\exp\left(-\frac{\lambda}{2}\bm{a}^\intercal \bm{a}\right)\Biggr]\text{d}\bm{a}.
    \end{align*}
    Let $\bm{t}\coloneqq (1 + \lambda)^{1/2}\bm{a} - (1+\lambda)^{-1/2}\mathbf{\Delta}$. With this variable, 
    \begin{align*}
        &\exp\dparen*{\bm{a}^\intercal \mathbf{\Delta} - \frac{1}{2}\bm{a}^\intercal \bm{a}}\frac{\lambda^{m/2}}{\sqrt{(2\pi)^m}}\exp\left(-\frac{\lambda}{2}\bm{a}^\intercal \bm{a}\right)\\
        &= \lambda^{m/2}\exp\dparen*{ \frac{(1+\lambda)^{-1}}{2}\mathbf{\Delta}^\intercal \mathbf{\Delta} }\frac{1}{\sqrt{(2\pi)^m}}\exp\dparen*{ -\frac{\bm{t}^\intercal \bm{t}}{2} }.
    \end{align*}
    By the change of variables, $\int L_{h+\bm{a}^\intercal\bm{h}}(B)N(0,\lambda^{-1}I_m)(\text{d}\bm{a})$ equals
    \begin{align*}
        &\int \expect_h\Biggl[\ind_B\dparen*{V - \kappa_{\theta_0}'(h)\dparen*{\Dot{\theta}_0(h) + \frac{\bm{\Delta}^\intercal \bm{\Dot{\theta}_0(h)}}{1+\lambda} + \frac{\bm{t}^\intercal \bm{\Dot{\theta}_0(h)}}{(1+\lambda)^{1/2}} } +C} c_{\lambda,m}(\mathbf{\Delta})\Biggr]N(0,I_m)(\text{d}\bm{t})\\
        &= \expect_h\Biggl[\ind_B\dparen*{V - \kappa_{\theta_0}'(h)\dparen*{\Dot{\theta}_0(h) + \frac{\bm{\Delta}^\intercal \bm{\Dot{\theta}_0(h)}}{1+\lambda} + G_{\lambda,m} } +C} c_{\lambda,m}(\mathbf{\Delta})\Biggr],
    \end{align*}
    where $G_{\lambda,m} \coloneqq \bm{A}^\intercal \bm{\Dot{\theta}_0(h)}/(1+\lambda)^{1/2}$ for $\bm{A} \sim N(0,I_m)$ and
    \begin{equation*}
        c_{\lambda,m}(\bm{\Delta}) \coloneqq \dparen*{\frac{\lambda}{1+\lambda}}^{m/2}\exp\dparen*{ \frac{(1+\lambda)^{-1}}{2}\mathbf{\Delta}^\intercal \mathbf{\Delta} }.
    \end{equation*}
    Note that $(V,\bm{\Delta})$ and $G_{\lambda,m}$ are independent by construction. Observe that the set function $Q$ given by
    \begin{equation*}
        Q(A) = \expect_h\dbrack*{ \ind_A\dparen*{V,\frac{\bm{\Delta}^\intercal \bm{\Dot{\theta}_0(h)}}{(1+\lambda)^{1/2}}} c_{\lambda,m}(\mathbf{\Delta}) } \quad \text{for every} \quad A \in \mathcal{B}(\mathbb{R}\times\mathbb{B})
    \end{equation*}
    is a Borel probability measure. Indeed, $c_{\lambda,m}$ is non-negative, and it is easy to confirm that $Q(\mathbb{R}\times\mathbb{B}) = \expect_{h}[c_{\lambda,m}(\mathbf{\Delta})] = 1$. It follows that $0 \leq Q(A) \leq 1$ for every $A$. The countable additivity of $Q$ also follows from the monotone convergence theorem. Setting $X = (X_1,X_2) = (\bm{\Delta}^\intercal \bm{\Dot{\theta}_0(h)}/(1+\lambda)^{1/2},\bm{\Delta})$, $Y = V$, $Z = G_{\lambda,m}$,
    \begin{align*}
        f(X,Y,Z) = Y - \kappa_{\theta_0}'(h)\dparen*{\Dot{\theta}_0(h) + X_1 + Z } +C, \quad \text{and}\quad
        g(X) = c_{\lambda,m}(X_2).
    \end{align*}
    Apply Lemma~A.4 in \textcite{Ponomarev2022} to ensure the existence of a random element $(V_{\lambda,m},W_{\lambda,m}) \in \mathbb{R}\times\mathbb{B}$ independent of $G_{\lambda,m}$ such that 
    \begin{equation*}
        \int L_{h+\bm{a}^\intercal\bm{h}}(B)N(0,\lambda^{-1}I_m)(\text{d}\bm{a}) = \expect_h\dbrack*{\ind_B\dparen*{V_{\lambda,m} - \kappa_{\theta_0}'(h)\dparen*{\Dot{\theta}_0(h) + W_{\lambda,m} + G_{\lambda,m} } + C}}.
    \end{equation*}
    This concludes the proof.
\end{proof}

\begin{lemma}\label{lem:lipschitzness-of-bounded-loss}
    Regard $\mathbb{R}\times\mathbb{B}\times\mathbb{B}\times\mathbb{B}$ as the metric space with product metric,
    \begin{equation*}
        \abs*{ v - v' } + \norm{ w - w'}_{\mathbb{B}} + \norm{ s - s'}_{\mathbb{B}} + \norm{g - g'}_{\mathbb{B}}.
    \end{equation*}
    Then, the function $f_M$ defined in the proof of \cref{thm:rimr-lower-bound} is and Lipschitz continuous with Lipshitz constant not depending on $M$.
\end{lemma}
\begin{proof}[Proof of \cref{lem:lipschitzness-of-bounded-loss}]
    The boundedness of $f_M$ is obvious by definition. The Lipschitz continuity follows from the direct calculations:
    \begin{align*}
        \MoveEqLeft
        \abs*{ f_M(g,v,w,s) - f_M(g,v',w',s') }\\
        \leq {} & \abs*{ \dparen*{ v - \kappa_{\theta_0}'\dparen*{ g + w + s } + \kappa_{\theta_0}'(s) } - \dparen*{ v' - \kappa_{\theta_0}'\dparen*{ g' + w' + s' } + \kappa_{\theta_0}'(s') } }\\
        \leq {} & bounded \abs*{ v - v' } + C_{\kappa_{\theta_0}'}\norm{g + w + s - g' - w' - s'}_{\mathbb{B}} + C_{\kappa_{\theta_0}'}\norm{s - s'}_{\mathbb{B}}\\
        \leq {} & \abs*{ v - v' } + C_{\kappa_{\theta_0}'}\norm{g - g'}_{\mathbb{B}} + C_{\kappa_{\theta_0}'}\norm{w - w'}_{\mathbb{B}} + 2C_{\kappa_{\theta_0}'}\norm{s - s'}_{\mathbb{B}}\\
        \leq {} & \max\dbrace{1, 2C_{\kappa_{\theta_0}'}}\dparen*{ \abs*{ v - v' } + \norm{ w - w'}_{\mathbb{B}} + \norm{ s - s'}_{\mathbb{B}} + \norm{g - g'}_{\mathbb{B}} },
    \end{align*}
    where the second inequality follows from \cref{lem:kappa-directionally-differentiable}.
\end{proof}
\begin{lemma}\label{lem:kappa-hat-property}
    Suppose \cref{asum:bounds-of-identificaiton-region,asum:bounds-derivative-estimator} hold, and $\hat{\kappa}_n':\mathbb{B} \to \mathbb{R}$ is defined in \cref{eq:kappa-hat}. Then, for every $\delta > 0$, $\sup_{s \in K_n^\delta} \abs*{ \hat{\kappa}_n'(s) - \kappa_{\theta_0}'(s) } = o_p(1)$ under every $h \in T(P_0)$, where $\{K_n\}$ is an arbitrary expanding sequence of compact subsets of $\mathbb{B}$ such that it contains zero vector for all $n$ and that $\sup_{s \in K_n} \normB{s} = o(\sqrt{n})$. Moreover, there exists an asymptotically tight random $C_{\hat{\kappa}_n'} \in \mathbb{R}$ such that $\abs{\hat{\kappa}_n'(s_1) - \hat{\kappa}_n'(s_2)} \leq C_{\hat{\kappa}_n'}\norm{s_1-s_2}_{\mathbb{B}}$ holds outer almost surely for every $(s_1,s_2) \in \mathbb{B}\times\mathbb{B}$.
\end{lemma}
\begin{proof}[Proof of \cref{lem:kappa-hat-property}]
    For ease of notation, let 
    \begin{equation*}
        \eta_L(\theta) \coloneqq \frac{\tau_L^-(\theta)}{\dparen[\big]{\tau_U^+(\theta) + \tau_L^-(\theta)}^2},
        \quad \text{and} \quad
        \eta_U(\theta) \coloneqq \frac{\tau_U^+(\theta)}{\dparen[\big]{\tau_U^+(\theta) + \tau_L^-(\theta)}^2}.
    \end{equation*}
    First, prove the first claim. Fix $\delta > 0$ and observe that
    \begin{align}
        \begin{aligned}
            \sup_{s \in K_n^\delta}\abs*{\hat{\kappa}_n'(s) - \kappa_{\theta_0}'(s)} 
            & \leq  
            \sup_{s \in K_n^\delta}\abs*{ \eta_L\dparen{\hat{\theta}_n} \cdot \hat{m}_{\tau_U(\hat{\theta}_n),n}'\dparen[\big]{\hat{\tau}_{U,n}'(s)} - \eta_L\dparen{\theta_0} \cdot m_{\tau_U(\theta_0)}'(\tau_{U,\theta_0}'(s))}\\
            & \quad + \sup_{s \in K_n^\delta}\abs*{ \eta_U\dparen{\hat{\theta}_n} \cdot \hat{m}_{\tau_L(\hat{\theta}_n),n}'\dparen[\big]{\hat{\tau}_{L,n}'(s)} - \eta_U\dparen{\theta_0} \cdot m_{\tau_L(\theta_0)}'(\tau_{L,\theta_0}'(s))}.
        \end{aligned}
        \label{eq:pw-est-err-kappa-deriv}
    \end{align}
    Here, the first term of the right-hand side of the preceding inequality converges to zero in probability; the same property for the second term can be shown similarly. First, observe that
    \begin{align}
        \MoveEqLeft
        \sup_{s \in K_n^\delta}\abs*{\eta_L\dparen{\hat{\theta}_n}\cdot \hat{m}_{\tau_U(\hat{\theta}_n),n}'\dparen[\big]{\hat{\tau}_{U,n}'(s)} - \eta_L\dparen{\theta_0}\cdot m_{\tau_U(\theta_0)}'(\tau_{U,\theta_0}'(s))}\nonumber\\
        \begin{split}
            \leq {} & 
            \sup_{s \in K_n^\delta}\abs*{ \eta_L\dparen{\hat{\theta}_n} - \eta_L\dparen{\theta_0} } \cdot \abs*{ m_{\tau_U(\theta_0)}'\dparen[\big]{\tau_{U,\theta_0}'(s)}}\\
            &+ \sup_{s \in K_n^\delta}\eta_L\dparen{\hat{\theta}_n} \cdot \abs*{ \hat{m}_{\tau_U(\hat{\theta}_n),n}'\dparen*{\hat{\tau}_{U,n}'(s)} - \hat{m}_{\tau_U(\hat{\theta}_n),n}'\dparen*{ \tau_{U,\theta_0}'(s) } }\\
            & + \sup_{s \in K_n^\delta}\eta_L\dparen{\hat{\theta}_n} \cdot \abs*{ \hat{m}_{\tau_U(\hat{\theta}_n),n}'\dparen*{ \tau_{U,\theta_0}'(s) } - m_{\tau_U(\theta_0)}'\dparen*{\tau_{U,\theta_0}'(s)} }
        \end{split}\nonumber\\
        \begin{split}
            \leq {} & 
            \abs*{ \eta_L\dparen{\hat{\theta}_n} - \eta_L\dparen{\theta_0} } \cdot \sup_{s \in K_n^\delta}\abs*{ \tau_{U,\theta_0}'(s) }
            + \eta_L\dparen{\hat{\theta}_n} \cdot \sup_{s \in K_n^\delta}\abs*{ \hat{\tau}_{U,n}'(s) - \tau_{U,\theta_0}'(s) }\\
            & + \eta_L\dparen{\hat{\theta}_n} \cdot \sup_{s \in K_n^\delta}\abs*{ \hat{m}_{\tau_U(\hat{\theta}_n),n}'\dparen*{ \tau_{U,\theta_0}'(s) } - m_{\tau_U(\theta_0)}'\dparen*{\tau_{U,\theta_0}'(s)} },
        \end{split}
        \label{eq:pw-est-err-kappa-deriv-1stUB}
    \end{align}
    where the last inequality follows from the Lipschitz continuity of $\hat{m}_{x,n}$ and $\max\{x,0\} \leq \abs{x}$. Thus, it suffices to see that each term in \cref{eq:pw-est-err-kappa-deriv-1stUB} converges in probability to zero.
    For the first term of \cref{eq:pw-est-err-kappa-deriv-1stUB}, \cref{asum:bounds-of-identificaiton-region} ensures that $\eta_L\dparen{\theta}$ is Hadamard directionally differentiable at $\theta_0$. Then, the delta method for Hadamard directionally differentiable maps gives 
    \begin{equation*}
        \sqrt{n}\dparen{\eta_L\dparen{\hat{\theta}_n} - \eta_L\dparen{\theta_0}} \overset{0}{\leadsto} \eta_{L,\theta_0}'(\mathbb{G}_0) \quad \text{in} \quad \mathbb{R},
    \end{equation*}
    where $\eta_{L,\theta_0}'$ denotes the Hadamard directional derivative of $\eta_L$ at $\theta_0$. With the positive homogeneity of Hadamard directional derivatives, \cref{assumptionenum:lipschitz-continuity} implies for any $s \in K_n^{\delta}$
    \begin{equation*}
        \abs*{\tau_{U,\theta_0}'(s)}
        = \abs*{ \tau_{U,\theta_0}'(s) - \tau_{U,\theta_0}'(0) }
        \leq C_{\tau_{U,\theta_0}'}\normB{s} 
        \leq C_{\tau_{U,\theta_0}'}(\delta + \normB{\tilde{s}})
        \leq C_{\tau_{U,\theta_0}'}(\delta + \sup_{s \in K_n}\normB{s}),
    \end{equation*}
    where $\tilde{s}$ is an element of $K_n$ such that $\normB{s - \tilde{s}} < \delta$.
    Combining these observations with $\sup_{s \in K_n}\normB{s} = o(\sqrt{n})$ yields
    \begin{equation*}
        \abs*{\eta_L\dparen{\hat{\theta}_n} - \eta_L\dparen{\theta_0}}\cdot \sup_{s \in K_n^\delta} \abs*{\tau_{U,\theta_0}'(s)} = \abs*{\sqrt{n}\dparen*{\eta_L\dparen{\hat{\theta}_n} - \eta_L\dparen{\theta_0}}}\cdot \frac{C_{\tau_{U,\theta_0}'}(\delta + \sup_{s \in K_n}\normB{s})}{\sqrt{n}} = o_p(1).
    \end{equation*}
    The second term in \cref{eq:pw-est-err-kappa-deriv-1stUB} converges in probability to zero by \cref{asum:bounds-derivative-estimator}.
    Finally, the third term in \cref{eq:pw-est-err-kappa-deriv-1stUB} converges in probability to zero in the following way. Consider the case where $\tau_U(\theta_0) = 0$. If $\abs{\tau_U\dparen{\hat{\theta}_n}} \leq \epsilon_n$, then
    \begin{equation}
        \eta_L(\hat{\theta}_n) \cdot \abs*{ \hat{m}_{\tau_U(\hat{\theta}_n),n}'\dparen*{ \tau_{U,\theta_0}'(s) } - m_{\tau_U(\theta_0)}'\dparen*{\tau_{U,\theta_0}'(s)} } = 0 \quad \text{for every} \quad s \in K_n^{\delta}.
        \label{eq:lemma-est-err-0}
    \end{equation}
    Taking the contraposition gives, for any $\epsilon > 0$, 
    \begin{align*}
        \MoveEqLeft
        \limsup_{n\to\infty}P_{n,0}\dparen*{ \eta_L\dparen{\hat{\theta}_n} \cdot \sup_{s \in K_n^{\delta}} \abs*{ \hat{m}_{\tau_U(\hat{\theta}_n),n}'\dparen*{ \tau_{U,\theta_0}'(s) } - m_{\tau_U(\theta_0)}'\dparen*{\tau_{U,\theta_0}'(s)} } > \epsilon }\\
        &\leq \limsup_{n\to\infty}P_{n,0}\dparen*{ \abs{\tau_U\dparen{\hat{\theta}_n}} > \epsilon_n }\\
        &\leq \limsup_{n\to\infty}P_{n,0}\dparen*{ \sqrt{n}\abs{\tau_U(\hat{\theta}_n) - \tau_U(\theta_0)} > \sqrt{n}\epsilon_n } = 0,
    \end{align*}
    where the last equality follows from $\sqrt{n}\abs{\tau_U(\hat{\theta}_n) - \tau_U(\theta_0)} \overset{0}{\leadsto} \abs{\tau_{U,\theta_0}'(\mathbb{G}_0)}$ and $\sqrt{n}\epsilon_n \uparrow \infty$. Next, consider the case $\tau_U(\theta_0) \neq 0$. If $\abs{\tau_U(\hat{\theta}_n)} > \epsilon_n$, then \cref{eq:lemma-est-err-0} holds. Again, taking contraposition gives, for any $\epsilon > 0$,
    \begin{align*}
        \MoveEqLeft
        \limsup_{n\to\infty}P_{n,0}\dparen*{\eta_L\dparen{\hat{\theta}_n} \cdot \sup_{s \in K_n^{\delta}} \abs*{ \hat{m}_{\tau_U(\hat{\theta}_n),n}'\dparen*{ \tau_{U,\theta_0}'(s) } - m_{\tau_U(\theta_0)}'\dparen*{\tau_{U,\theta_0}'(s)} } > \epsilon }\\
        &\leq \limsup_{n\to\infty}P_{n,0}\dparen*{ \abs{\tau_U\dparen{\hat{\theta}_n}} \leq \epsilon_n } = 0,
    \end{align*}
    where the last equality follows from $\abs{\tau_U(\hat{\theta}_n)}$ converges in probability to $\abs{\tau_U(\theta_0)}$ and $\epsilon_n \downarrow 0$. Combining these observations implies the desired result.
    \par
    Next is the proof of the second claim. Observe that
    \begin{align*}
        \abs*{\hat{\kappa}_n'(s_1) - \hat{\kappa}_n'(s_2)}
            & \leq \eta_L\dparen{\hat{\theta}_n} \cdot \abs*{\hat{m}_{\tau_U(\hat{\theta}_n),n}'(\hat{\tau}_{U,n}'(s_1)) - \hat{m}_{\tau_U(\hat{\theta}_n),n}'(\hat{\tau}_{U,n}'(s_2))}\\
            & \quad  + \eta_U\dparen{\hat{\theta}_n} \cdot \abs*{\hat{m}_{\tau_L(\hat{\theta}_n),n}'(\hat{\tau}_{L,n}'(s_1)) - \hat{m}_{\tau_L(\hat{\theta}_n),n}'(\hat{\tau}_{L,n}'(s_2))}\\
        & \leq \eta_L\dparen{\hat{\theta}_n} \cdot \abs*{\hat{\tau}_{U,n}'(s_1) - \hat{\tau}_{U,n}'(s_2)}
        + \eta_U\dparen{\hat{\theta}_n} \cdot \abs*{\hat{\tau}_{L,n}'(s_1) - \hat{\tau}_{L,n}'(s_2)}\\
        & \leq \dparen*{ \eta_L\dparen{\hat{\theta}_n} \cdot C_{\hat{\tau}_{U,n}'}
        + \eta_U\dparen{\hat{\theta}_n} \cdot C_{\hat{\tau}_{L,n}'}}\cdot\norm{s_1-s_2}_{\mathbb{B}}
    \end{align*}
    holds outer almost surely, where the last inequality follows from \cref{asum:bounds-derivative-estimator}.(ii). Thus, it suffices to show that the coefficient of $\norm{s_1-s_2}_{\mathbb{B}}$ is asymptotically tight. Fix $\epsilon > 0$. By assumption, there exist $M_1,M_2 > 0$ such that
    \begin{equation*}
        \limsup_{n \to \infty} P_{n,0}\dparen*{C_{\hat{\tau}_U',n} \geq M_1 + \tilde{\delta}} < \frac{\epsilon}{4}
        \quad\text{and}\quad
        \limsup_{n \to \infty} P_{n,0}\dparen*{C_{\hat{\tau}_L',n} \geq M_2 + \tilde{\delta}} < \frac{\epsilon}{4}
    \end{equation*}
    for every $\tilde{\delta} > 0$. Similarly, $\eta_j\dparen{\hat{\theta}_n} \overset{P_{n,0}}{\to} \eta_j(\theta_0),j \in \{L,U\}$ implies the existence of $M_3,M_4>0$ satisfying
    \begin{equation*}
        \limsup_{n \to \infty} P_{n,0}\dparen*{\eta_L\dparen{\hat{\theta}_n} \geq M_3 + \tilde{\delta}} < \frac{\epsilon}{4}
        \quad\text{and}\quad
        \limsup_{n \to \infty} P_{n,0}\dparen*{\eta_U\dparen{\hat{\theta}_n} \geq M_4 + \tilde{\delta}} < \frac{\epsilon}{4}
    \end{equation*}
    for every $\tilde{\delta} > 0$. For each $\delta > 0$, choose $\tilde{\delta} > 0$ such that
    \begin{equation*}
        2\tilde{\delta}^2 + \dparen{M_1 + M_2 + M_3 + M_4}\tilde{\delta} < \delta.
    \end{equation*}
    If $\eta_L\dparen{\hat{\theta}_n} < M_3 + \tilde{\delta}$, $C_{\hat{\tau}_U',n} < M_1 + \tilde{\delta}$, $\eta_U\dparen{\hat{\theta}_n} < M_4 + \tilde{\delta}$, and $C_{\hat{\tau}_L',n} < M_2 + \tilde{\delta}$, then $\eta_L\dparen{\hat{\theta}_n} \cdot C_{\hat{\tau}_{U,n}'} + \eta_U\dparen{\hat{\theta}_n} \cdot C_{\hat{\tau}_{L,n}'} < M_1M_3 + M_2M_4 + \delta$ by the choice of $\tilde{\delta}$. Therefore,
    \begin{align*}
        \MoveEqLeft
        \limsup_{n \to \infty} P_{n,0}\dparen*{ \eta_L\dparen{\hat{\theta}_n} \cdot C_{\hat{\tau}_{U,n}'} + \eta_U\dparen{\hat{\theta}_n} \cdot C_{\hat{\tau}_{L,n}'} \geq M_1M_3 + M_2M_4 + \delta}\\
        \begin{split}
            \leq {} & \limsup_{n \to \infty} P_{n,0}\dparen*{C_{\hat{\tau}_U',n} \geq M_1 + \tilde{\delta}}
            + \limsup_{n \to \infty} P_{n,0}\dparen*{\eta_L\dparen{\hat{\theta}_n} \geq M_3 + \tilde{\delta}}\\
            & + \limsup_{n \to \infty} P_{n,0}\dparen*{C_{\hat{\tau}_L',n} \geq M_2 + \tilde{\delta}}
            + \limsup_{n \to \infty} P_{n,0}\dparen*{\eta_U\dparen{\hat{\theta}_n} \geq M_4 + \tilde{\delta}}
        \end{split}\\
        < {} & \epsilon,
    \end{align*}
    which concludes the proof.
\end{proof}

\begin{lemma}\label{lem:pw-consistency-adjustment-terms}
    Suppose \cref{asum:bounds-of-identificaiton-region,asum:bounds-derivative-estimator,asum:on-bootstrap,asum:support-G0-compact-sieve,asum:tangent-set,asum:theta-differentiable} hold. Fix $M \in \mathbb{N}$ and a compact set $K \subset \mathbb{B}$ arbitrarily. Let $\hat{\Psi}_{M,n}$ and $\Psi_{M,n}$ denote the set of minimizers of $\hat{F}_{M,n}(w)$ and $F_M(w)$ over $K$. Then, for each $\hat{w}_{M,n} \in \hat{\Psi}_{M,n}$, 
    \begin{equation*}
        \inf_{w^* \in \Psi_{M}} \normB{\hat{w}_{M,n} - w^*} = o_p(1)
    \end{equation*}
    under every $h \in T(P_0)$.
\end{lemma}

\begin{proof}[Proof of \cref{lem:pw-consistency-adjustment-terms}]
    Prepare some notations to simplify the notation. Let
    \begin{align*}
        \hat{l}_{M,n}(g,w,s) = \ell_M\dparen*{\hat{\kappa}_n'(g + w + s) - \hat{\kappa}_n'(s)}
        \text{ and }
        l_M(g,w,s) = \ell_M\dparen*{\kappa_{\theta_0}'(g + w + s) - \kappa_{\theta_0}'(s)}.
    \end{align*}
    Then, 
    \begin{align*}
        \hat{F}_{M,n}(w) 
        &= \sup_{s \in \hat{S}_n}\max\dbrace*{ \tau_L^-\dparen{\hat{\theta}_n} \expect \dbrack*{\hat{l}_{M,n}(\hat{\mathbb{G}}_n,w,s)\middle| \mathbf{Z}_n}, -\tau_U^+\dparen{\hat{\theta}_n} \expect\dbrack*{\hat{l}_{M,n}(\hat{\mathbb{G}}_n,w,s)\middle| \mathbf{Z}_n} },\\
        F_M(w) 
        &= \sup_{s \in S(\mathbb{G}_0)} \max\dbrace*{\tau_L^-\dparen{\theta_0} \expect \dbrack*{l(\mathbb{G}_0,w,s)}, -\tau_U^+\dparen{\theta_0} \expect\dbrack*{l(\mathbb{G}_0,w,s)} }.
    \end{align*}
    \par
    For a while, assume uniform consistency, 
    \begin{equation}
        \sup_{w \in K} \abs*{ \hat{F}_{M,n}(w) - F_M(w) } = o_p(1),\label{eq:uniform-consistency}
    \end{equation}
    under $h = 0$ and the identification condition,
    \begin{equation}
        \inf_{w \in K\setminus \Psi_M^\epsilon} F_M(w) > \inf_{w \in K} F_M(w).\label{eq:identificaiton-condition}
    \end{equation}
    for any $\epsilon > 0$. The satisfaction of these assumptions will be proved later. Fix $\epsilon > 0$ arbitrarily. For each $w \in K\setminus \Psi_M^\epsilon$, the identification condition \cref{eq:identificaiton-condition} implies that there exists $\delta >0$ such that 
    \begin{equation*}
        F_M(w) - F_M(w^*) > \delta,
    \end{equation*}
    where $w^* \in \Psi_M$. Then, it follows that 
    \begin{align*}
        &P_{n,0}\dparen*{ \inf_{w^*\in \Psi_M}\normB{\hat{w}_{M,n} - w^*} > \epsilon }\\
        &\leq  P_{n,0}\dparen*{ \hat{w}_{M,n} \in K\setminus \Psi_M^\epsilon }\\
        &\leq P_{n,0}\dparen*{ F_M(\hat{w}_{M,n}) - F_M(w^*) > \delta }\\
        &= P_{n,0}\dparen*{ F_M(\hat{w}_{M,n}) - \hat{F}_{M,n}(\hat{w}_{M,n}) + \hat{F}_{M,n}(\hat{w}_{M,n}) - F_M(w^*) > \delta }\\
        &\leq P_{n,0}\dparen*{ F_M(\hat{w}_{M,n}) - \hat{F}_{M,n}(\hat{w}_{M,n}) + \hat{F}_{M,n}(w^*) - F_M(w^*) > \delta }\\
        &\leq P_{n,0}\dparen*{ 2\sup_{w \in K} \abs*{ \hat{F}_{M,n}(w) - F_{M}(w) } > \delta },
    \end{align*}
    where the third inequality follows since $\hat{w}_{M,n}$ is the minimizer of $\hat{F}_{M,n}(w)$. The right-hand side of the last inequality converges to zero under $h=0$ by the uniform consistency \cref{eq:uniform-consistency}. Then, the desired result follows from Le Cam's first lemma.
    \par
    The remaining task is to prove the uniform consistency \cref{eq:uniform-consistency} and identification condition \cref{eq:identificaiton-condition}. To prove the former, decompose the difference as follows:
    \begin{equation} 
        \begin{aligned}
            \hat{F}_{M,n}(w) - F_M(w) 
            = {} &
            \dparen*{\hat{F}_{M,n}(w) - \hat{F}_{M,n}^{(1)}(w)} 
            + \dparen*{\hat{F}_{M,n}^{(1)}(w) - \hat{F}_{M,n}^{(2)}(w)}\\
            &+ \dparen*{\hat{F}_{M,n}^{(2)}(w) - \hat{F}_{M,n}^{(3)}(w)}
            + \dparen*{\hat{F}_{M,n}^{(3)}(w) - F_M(w)},
        \end{aligned}
        \label{eq:criterion-difference-decompose}
    \end{equation}
    where 
    \begin{align*}
        \hat{F}_{M,n}^{(1)}(w)
        &= \sup_{s \in S_n} \max\dbrace*{ \tau_L^-\dparen{\hat{\theta}_n} \expect \dbrack*{\hat{l}_{M,n}(\hat{\mathbb{G}}_n,w,s)| \mathbf{Z}_n}, -\tau_U^+\dparen{\hat{\theta}_n} \expect\dbrack*{\hat{l}_{M,n}(\hat{\mathbb{G}}_n,w,s)| \mathbf{Z}_n} },\\
        \hat{F}_{M,n}^{(2)}(w)
        &= \sup_{s \in S_n} \max\dbrace*{ \tau_L^-\dparen{\theta_0} \expect \dbrack*{l_M(\hat{\mathbb{G}}_n,w,s)| \mathbf{Z}_n}, -\tau_U^+\dparen{\theta_0} \expect\dbrack*{l_M(\hat{\mathbb{G}}_n,w,s)| \mathbf{Z}_n} },\\
        \hat{F}_{M,n}^{(3)}(w) 
        &= \sup_{s \in S_n} \max\dbrace*{ \tau_L^-\dparen{\theta_0} \expect \dbrack*{l_M(\mathbb{G}_0,w,s)}, -\tau_U^+\dparen{\theta_0} \expect\dbrack*{l_M(\mathbb{G}_0,w,s)} }.
    \end{align*}
    Then, it suffices to show that each difference in \cref{eq:criterion-difference-decompose} uniformly converges in probability to zero.
    \par
    For the first difference in \cref{eq:criterion-difference-decompose}, observe that
    \begin{align*}
        \MoveEqLeft
        \left|
        \max\dbrace*{
            \begin{aligned}
                &\tau_L^-\dparen{\hat{\theta}_n}\expect \dbrack*{\hat{l}_{M,n}(\hat{\mathbb{G}}_n,w,s_1)| \mathbf{Z}_n},\\
                -&\tau_U^+\dparen{\hat{\theta}_n} \expect\dbrack*{\hat{l}_{M,n}(\hat{\mathbb{G}}_n,w,s_1)| \mathbf{Z}_n}
            \end{aligned}
        }
        -
        \max\dbrace*{
            \begin{aligned}
                &\tau_L^-\dparen{\hat{\theta}_n} \expect \dbrack*{\hat{l}_{M,n}(\hat{\mathbb{G}}
                _n,w,s_2)| \mathbf{Z}_n},\\
                -&\tau_U^+\dparen{\hat{\theta}_n}\expect\dbrack*{\hat{l}_{M,n}(\hat{\mathbb{G}}_n,w,s_2)| \mathbf{Z}_n}
            \end{aligned}
        }
        \right|\\
        \leq {} & \dparen*{ \tau_U^+\dparen{\hat{\theta}_n} + \tau_L^-\dparen{\hat{\theta}_n} } \expect \dbrack*{\abs*{\hat{l}_{M,n}(\hat{\mathbb{G}}_n,w,s_1) - \hat{l}_{M,n}(\hat{\mathbb{G}}_n,w,s_2)} | \mathbf{Z}_n}\\
        \leq {} & \dparen*{ \tau_U^+\dparen{\hat{\theta}_n} + \tau_L^-\dparen{\hat{\theta}_n} } \mathbb{E}\dbrack*{ \abs*{ \hat{\kappa}_n'(\hat{\mathbb{G}}_n+w+s_1) - \hat{\kappa}_n'(\hat{\mathbb{G}}_n + w + s_2) } + \abs*{ \hat{\kappa}_n'(s_1) - \hat{\kappa}_n'(s_2) } \middle| \mathbf{Z}_n }\\
        \leq {} & 2\dparen*{ \tau_U^+\dparen{\hat{\theta}_n} + \tau_L^-\dparen{\hat{\theta}_n} }C_{\hat{\kappa}_n'}\norm{s_1 - s_2}_{\mathbb{B}},
    \end{align*}
    where the last inequality follows from \cref{lem:kappa-hat-property}, provided \cref{asum:bounds-of-identificaiton-region,asum:bounds-derivative-estimator}. Moreover, the same lemma ensures the asymptotic tightness of $C_{\hat{\kappa}_n'}$; hence, $2\dparen{ \tau_L^-\dparen{\hat{\theta}_n} + \tau_U^+\dparen{\hat{\theta}_n} }C_{\hat{\kappa}_n'}$ is also asymptotically tight. Therefore, with \cref{asum:support-G0-compact-sieve}, Lemma~A.10 in \citet{Ponomarev2022} concludes 
    \begin{equation*}
        \sup_{(w)\in K} \abs*{ \hat{F}_{M,n}(w) - \hat{F}_{M,n}^{(1)}(w) } = o_p(1).
    \end{equation*}
    \par
    Next, consider the second difference in \cref{eq:criterion-difference-decompose} and observe that 
    \begin{align*}
        \MoveEqLeft
        \sup_{w \in K} \abs*{ \hat{F}_{M,n}^{(2)}(w) - \hat{F}_{M,n}^{(3)}(w) }\\
        \begin{split}
            \leq {} & \sup_{w \in K} \sup_{s \in S_n}\abs*{ \tau_L^-\dparen{\hat{\theta}_n} \cdot \expect\dbrack*{ \hat{l}_{M,n}(\hat{\mathbb{G}}_n,w,s) | \mathbf{Z}_n} - \tau_L^-(\theta_0)\cdot\expect\dbrack*{l_M(\hat{\mathbb{G}}_n,w,s) | \mathbf{Z}_n } }\\
            &+ \sup_{w \in K} \sup_{s \in S_n} \abs*{ \tau_U^+\dparen{\hat{\theta}_n}\cdot\expect\dbrack*{ \hat{l}_{M,n}(\hat{\mathbb{G}}_n,w,s) | \mathbf{Z}_n } - \tau_U^+(\theta_0)\cdot\expect\dbrack*{l_M(\hat{\mathbb{G}}_n,w,s) | \mathbf{Z}_n } }
        \end{split}\\
        \begin{split}
            \leq {} & \abs*{\tau_L^-\dparen{\hat{\theta}_n} - \tau_L^-(\theta_0)} \cdot \sup_{w \in K} \sup_{s \in S_n} \abs*{\expect\dbrack*{ \hat{l}_{M,n}(\hat{\mathbb{G}}_n,w,s)| \mathbf{Z}_n }}\\
            & + \tau_L^-(\theta_0) \cdot \sup_{w \in K} \sup_{s \in S_n} \expect\dbrack*{ \abs*{ \hat{l}_{M,n}(\hat{\mathbb{G}}_n,w,s) - l_M(\hat{\mathbb{G}}_n,w,s) } | \mathbf{Z}_n }\\
            & + \abs*{ \tau_U^+(\hat{\theta}_n) - \tau_U^+(\theta_0) } \cdot \sup_{w \in K} \sup_{s \in S_n} \abs*{ \expect\dbrack*{ \hat{l}_{M,n}(\hat{\mathbb{G}}_n,w,s) | \mathbf{Z}_n } }\\
            & + \tau_U^+(\theta_0) \cdot \sup_{w \in K} \sup_{s \in S_n} \expect\dbrack*{ \abs*{ \hat{l}_{M,n}(\hat{\mathbb{G}}_n,w,s) - l_M(\hat{\mathbb{G}}_n,w,s) } | \mathbf{Z}_n }
        \end{split}\\
        \begin{split}
            \leq {} & \dparen*{ \tau_U^+(\theta_0) + \tau_L^-(\theta_0) } \cdot \sup_{w \in K} \sup_{s \in S_n} \expect\dbrack*{ \abs*{ \hat{l}_{M,n}(\hat{\mathbb{G}}_n,w,s) - l_M(\hat{\mathbb{G}}_n,w,s) } | \mathbf{Z}_n }\\
            & + M\cdot \dparen*{ \abs*{\tau_L^-\dparen{\hat{\theta}_n} - \tau_L^-(\theta_0)} + \abs*{ \tau_U^+(\hat{\theta}_n) - \tau_U^+(\theta_0) } },
        \end{split}
    \end{align*}
    where the last inequality follows, as $\abs{ \hat{l}_{M,n}(\hat{\mathbb{G}}_n,w,s) } \leq M$. Under \cref{asum:bounds-of-identificaiton-region}, $\abs{\tau_L^-\dparen{\hat{\theta}_n} - \tau_L^-(\theta_0)} = o_p(1)$ and $\abs{ \tau_U^+(\hat{\theta}_n) - \tau_U^+(\theta_0) } = o_p(1)$.
    \par
    In what follows, the study confirms that 
    \begin{equation*}
        \adjustlimits\sup_{w \in K}\sup_{s \in S_n} \expect\dbrack*{ \abs*{ \hat{l}_{M,n}(\hat{\mathbb{G}}_n,w,s) - l_M(\hat{\mathbb{G}}_n,w,s) } | \mathbf{Z}_n } = o_p(1).
    \end{equation*}
    Note, first, that, under \cref{asum:on-bootstrap}, Lemma~S.3.1 in \citet{Fang2019} ensures $\hat{\mathbb{G}}_n \overset{0}{\leadsto} \mathbb{G}_0$ unconditionally. As $\mathbb{G}_0$ is tight, Lemma~1.3.8 in \citet{Vaart1996} implies the asymptotic tightness of $\hat{\mathbb{G}}_n$: that is, for any $\epsilon > 0$ and $\eta > 0$, there exists a compact set $S \subset \mathbb{B}$ such that 
    \begin{equation*}
        \limsup_{n \to \infty} P_{n,0}(\hat{\mathbb{G}}_n \notin S^\delta) \leq \epsilon\eta \quad\text{for every}\quad\delta >0.
    \end{equation*}
    Then, Markov's inequality yields
    \begin{equation*}
        \limsup_{n\to\infty} P_{n,0}\dparen*{P_{n,0}\dparen*{\hat{\mathbb{G}}_n \in S^\delta | \mathbf{Z}_n} > \epsilon } \leq \limsup_{n\to\infty} \frac{P_{n,0}(\hat{\mathbb{G}}_n \notin S^\delta)}{\epsilon} \leq \eta,
    \end{equation*}
    wihch implies $P_{n,0}\dparen{\hat{\mathbb{G}}_n \notin S^\delta | \mathbf{Z}_n } = o_p(1)$. From this observation, 
    \begin{align*}
        \MoveEqLeft
        \adjustlimits\sup_{w \in K}\sup_{s \in S_n}\expect\dbrack*{ \abs*{ \hat{l}_{M,n}(\hat{\mathbb{G}}_n,w,s) - l_M(\hat{\mathbb{G}}_n,w,s) } | \mathbf{Z}_n }\\
        \leq {} & 2M\cdot P_{n,0}\dparen{\hat{\mathbb{G}}_n \notin S^\delta | \mathbf{Z}_n } + \sup_{w \in K}\sup_{g \in S^\delta}\sup_{s \in S_n} \abs*{ \hat{l}_{M,n}(g,w,s) - l_M(g,w,s) }\\
        \leq {} &  \adjustlimits\sup_{w \in K}\sup_{g \in S^\delta}\sup_{s \in S_n} \abs*{ \hat{\kappa}_n'(g + w + s) - \kappa_{\theta_0}'(g + w + s)} + \sup_{s \in S_n^\delta}\abs*{ \hat{\kappa}_n'(s) - \kappa_{\theta_0}'(s) } + o_p(1)\\
        \leq {} & 
        \sup_{\tilde{s} \in K_{n,\delta}} \abs*{ \hat{\kappa}_n'(\tilde{s}) - \kappa_{\theta_0}'(\tilde{s})}
        + \sup_{s \in S_n^\delta}\abs*{ \hat{\kappa}_n'(s) - \kappa_{\theta_0}'(s) } + o_p(1),
    \end{align*}
    where $K_{n,\delta} = \dbrace{ g + w + s : w \in K, g \in S^\delta, s \in S_n}$.
    \par
    Let $K_n = \dbrace{g + w + s : w \in K, g \in S, s\in S_n}$. Three properties of $K_n$ follows. First, $K_n$ is a compact set, as $K \times S\times S_n \subset \mathbb{B}\times\mathbb{B}\times\mathbb{B}$ is compact by Tychonoff's theorem and is the image of $K \times S\times S_n$ by the continuous function $(w,g,s) \mapsto w + g + s$. Second, $K_{n,\delta} \subset K_n^\delta$. To observe this, let $b \in K_{n,\delta}$. There exists $w \in K$, $g \in S^\delta$, and $s \in S_n$ such that $b = w + g + s$. As $g \in S^\delta$, there further exists $\tilde{g} \in S$ with $\normB{g - \tilde{g}} < \delta$. Then, $b \in K^\delta$ since $\tilde{b} = w + \tilde{g} + s \in K_n$ satisfies $\norm{b - \tilde{b}}_{\mathbb{B}} < \delta$. Third, $\sup_{b \in K_n} \normB{b} = o(\sqrt{n})$. This is implied by the following inequality:
    \begin{equation*}
        \frac{\sup_{b \in K_n} \normB{b}}{\sqrt{n}} \leq \frac{\sup_{w \in K,g\in S}\norm{g + w}_{\mathbb{B}}}{\sqrt{n}} + \frac{\sup_{s \in S_n} \normB{s}}{\sqrt{n}},
    \end{equation*}
    where the first term of the right-hand side converges to zero, as $\sup_{w \in K,g\in S}\norm{g + w}_{\mathbb{B}}$ is finite by the compactness of $K$ and $S$, and the second term converges to zero by \cref{asum:support-G0-compact-sieve}. Note there is also $\sup_{b \in K_n^\delta} \normB{b} = o(\sqrt{n})$ for every $\delta > 0$.
    \par
    Combining the observations thus far yields
    \begin{align*}
        \MoveEqLeft
        \adjustlimits\sup_{w \in K}\sup_{s \in S_n}\expect\dbrack*{ \abs*{ \hat{l}_{M,n}(\hat{\mathbb{G}}_n,w,s) - l_M(\hat{\mathbb{G}}_n,w,s) } | \mathbf{Z}_n }\\
        &\leq \sup_{\tilde{s} \in K_{n}^\delta} \abs*{ \hat{\kappa}_n'(\tilde{s}) - \kappa_{\theta_0}'(\tilde{s})}
        + \sup_{s \in S_n^\delta}\abs*{ \hat{\kappa}_n'(s) - \kappa_{\theta_0}'(s) } + o_p(1).
    \end{align*}
    The convergence in probability of the right-hand side of the preceding inequality is guaranteed by \cref{lem:kappa-hat-property}.
    \par
    For the third difference in \cref{eq:criterion-difference-decompose}, observe that 
    \begin{align*}
        \MoveEqLeft
        \sup_{w \in K} \abs*{ \hat{F}_{M,n}^{(2)}(w) - \hat{F}_{M,n}^{(3)}(w) }\\
        \leq {} & \dparen*{ \tau_L^-(\theta_0) + \tau_U^+(\theta_0) } \cdot \adjustlimits\sup_{w \in K} \sup_{s \in S_n} \abs*{ \expect\dbrack*{l_M(\hat{\mathbb{G}}_n,w,s) | \mathbf{Z}_n } - \expect\dbrack*{l_M(\mathbb{G}_0,w,s)} }.
    \end{align*}
    \cref{lem:kappa-directionally-differentiable} implies that $l_M$ is uniformly bounded by $M$ and uniformly Lipschitz continuous with constant $C_{\kappa_{\theta_0}'}$. Thus, if $c = \max\dbrace{M,C_{\kappa_{\theta_0}'}}$, the right-hand side of the last inequality is bounded from above by
    \begin{equation*}
        \dparen*{\tau_L^-(\theta_0) + \tau_U^+(\theta_0)}\cdot \sup_{f \in \mathrm{BL}_c(\mathbb{B})} \abs*{ \expect\dbrack*{f(\hat{\mathbb{G}}_n)| \mathbf{Z}_n} - \expect\dbrack*{f(\mathbb{G}_0)} },
    \end{equation*}
    which converges in probability to zero by \cref{asum:on-bootstrap}.
    \par
    Finally, consider the forth difference in \cref{eq:criterion-difference-decompose} and observe that $\hat{F}_{M,n}^{(3)}(w)$ monotonically converges to $F_M(w)$ for every $w \in K$. To show this, note, first, that, for every $\epsilon > 0$, there exists $s \in S(\mathbb{G}_0)$ such that 
    \begin{equation}
        F(w) - \frac{\epsilon}{2} < \max\dbrace*{\tau_L^-\dparen{\theta_0}\cdot \expect \dbrack*{l_M(\mathbb{G}_0,w,s)}, -\tau_U^+\dparen{\theta_0}\cdot \expect\dbrack*{l_M(\mathbb{G}_0,w,s)}}.
        \label{eq:uni-con-4th-diff-1st-ineq}
    \end{equation}
    Let $\tilde{\epsilon} = \epsilon\dparen{4C_{\kappa_{\theta_0}'}\dparen{\tau_L^-(\theta_0)+\tau_U^+(\theta_0)}}^{-1}$. By \cref{asum:support-G0-compact-sieve}, there exists $n_0 \in \mathbb{N}$ such that, for all $n \geq n_0$, there exists $s_n \in S_n$, where $\norm{s - s_n}_{\mathbb{B}} < \tilde{\epsilon}$. Then, for $n \geq n_0$, 
    \begin{align*}
        \MoveEqLeft
        \left|
            \max\dbrace*{
            \begin{aligned}
                &\tau_L^-\dparen{\theta_0} \expect \dbrack*{l_M(\mathbb{G}_0,w,s)},\\
                -&\tau_U^+\dparen{\theta_0} \expect\dbrack*{l_M(\mathbb{G}_0,w,s)}
            \end{aligned}
            } -
            \max\dbrace*{
            \begin{aligned}
                &\tau_L^-\dparen{\theta_0} \expect \dbrack*{l_M(\mathbb{G}_0,w,s_n)},\\
                -&\tau_U^+\dparen{\theta_0} \expect\dbrack*{l_M(\mathbb{G}_0,w,s_n)}
            \end{aligned}
            }
        \right|\\
        & \leq \dparen*{ \tau_L^-(\theta_0) + \tau_U^+(\theta_0) } \expect\dbrack*{ \abs*{ l_M(\mathbb{G}_0,w,s) - l_M(\mathbb{G}_0,w,s_n) } }\\
        &\leq 2C_{\kappa_{\theta_0}'}\dparen*{ \tau_L^-(\theta_0) + \tau_U^+(\theta_0) } \norm{s - s_n}_{\mathbb{B}}\\
        &< \frac{\epsilon}{2}.
    \end{align*}
    Combining this inequality with \cref{eq:uni-con-4th-diff-1st-ineq} implies that, for all $n \geq n_0$, 
    \begin{align*}
        F_M(w) - \epsilon 
        < \max\dbrace*{
        \begin{aligned}
            &\tau_L^-\dparen{\theta_0} \expect \dbrack*{l_M(\mathbb{G}_0,w,s_n)},\\
            -&\tau_U^+\dparen{\theta_0} \expect\dbrack*{l_M(\mathbb{G}_0,w,s_n)}
        \end{aligned}
        }
        \leq \hat{F}_{M,n}^{(3)}(w).
    \end{align*}
    Noting that $\hat{F}_{M,n}^{(3)}(w) \leq F_M(w)$ holds by definition, $\lim_{n \to \infty} \hat{F}_{M,n}^{(3)}(w) = F_M(w)$ for each $w \in K$. Finally, as $F_M(w)$ is (Lipschitz) continuous in $w$, Dini's theorem ensures 
    \begin{equation*}
        \sup_{w \in K} \abs*{ \hat{F}_{M,n}^{(3)}(w) - F_M(w) } = o(1).
    \end{equation*}
    Combining observations thus far ensures the satisfaction of \cref{eq:uniform-consistency}.
    \par
    Next is the proof of \cref{eq:identificaiton-condition}. Fix $\epsilon > 0$ arbitrarily and suppose that $\inf_{w \in K\setminus \Psi_M^\epsilon} F_M(w) = \inf_{w\in K} F_M(w)$. Then, it is possible to construct a sequence $\{w_n\} \subset K\setminus \Psi_M^\epsilon$ such that 
    \begin{equation*}
        \lim_{n \to \infty} F_M(w_n) = \inf_{w\in K} F_M(w).
    \end{equation*}
    As the closure of $K\setminus \Psi_M^\epsilon$ is compact, one can pick a subsequence $\{w_{n_j}\}$ such that $w_{n_j} \to w^* \in \overline{K\setminus \Psi_M^\epsilon}$ as $j \to \infty$. Then, the continuity of $F_M(w)$ in $w$ gives 
    \begin{equation*}
        F_M(w^*) = \lim_{j \to \infty} F_M(w_{n_j}) = \inf_{w \in K} F_M(w),
    \end{equation*}
    which implies $w^* \in \Psi_M$. However, $w^* \in \overline{K\setminus \Psi_M^\epsilon}$ implies one can construct a sequence $\{\tilde{w}_m\} \subset K\setminus \Psi_M^\epsilon$ such that $\tilde{w}_m \to w^*$ as $m \to \infty$. This implies
    \begin{equation*}
        \inf_{w \in \Psi_M}\normB{w^* - w} = \lim_{m \to \infty} \inf_{w \in \Psi_M} \normB{\tilde{w}_m - w} \geq \epsilon > 0,
    \end{equation*}
    which contradicts to $w^* \in \Psi_M$.  
\end{proof}

\section{Details of Examples}

\subsection{\texorpdfstring{\cref{ex:identification-region-with-empirical-evidence-alone}}{Example~1}}
\label{sec:example1-detail}
(On page \pageref{ex:example1-fourth}) The study first introduced the class $\mathcal{P}$ of possible data-generating processes. Suppose $\mathcal{Y}$ is a compact set of $\mathbb{R}$ equipped with the usual metric. Moreover, regard $\mathcal{D}=\{0,1\}$ as the metric space with the discrete metric. Then, $\mathcal{Z}_0 = \mathcal{Y}\times\mathcal{Y}\times\mathcal{D}$ to which the potential observation $(Y(1),Y(0),D)$ belongs is a complete and separable metric space for the corresponding product metric. Similarly, $\mathcal{Z} = \mathcal{Y}\times\mathcal{D}$ to which the observation $(Y,D)$ belongs is also a complete and separable metric space for the corresponding product metric. As the class of possible distributions of $(Y(1),Y(0),D)$, let $\mathcal{Q} = \dbrace{Q:Q \ll \nu}$, where $\nu$ is a finite measure on $(\mathcal{Z}_0,\mathcal{B}(\mathcal{Z}_0))$. Define the observation function as $\gamma(Y(1),Y(0),D) = \dparen{Y(1)D + Y(0)(1-D),D}$. Then, define the class $\mathcal{P}$ of possible dgps by $\mathcal{P} = \dbrace*{ Q\circ\gamma^{-1} : Q \in \mathcal{Q} }$.
\par
The class $\mathcal{P}$ equals $\{P:P \ll \mu\}$, where $\mu = \nu\circ \gamma^{-1}$. The inclusion, $\mathcal{P} \subset \{P:P\ll \mu\}$, readily follows since $Q\circ\gamma^{-1} \ll \mu$. Indeed, since $Q \ll \nu$, $\mu(A) = \nu(\gamma^{-1}(A)) = 0$ implies $Q\circ\gamma^{-1}(A) = Q(\gamma^{-1}(A)) = 0$. The study can confirm the converse as follows. Let $P \ll \mu$. There exists a density function $p:\mathcal{Z}\to[0,\infty)$ such that $P(B) = \int_B p(z)\mu(\mathrm{d}z)$ for all $B \in \mathcal{B}(\mathcal{Z})$. Define a non-negative measurable function $q:\mathcal{Z}_0\to[0,\infty)$ by $q(z_0) = p(\gamma(z_0))$ for all $z_0 \in \mathcal{Z}_0$, and define $Q(A) = \int_A q(z_0)\nu(\mathrm{d}z_0)$ for all $A \in \mathcal{B}(\mathcal{Z}_0)$. By construction, $Q \ll \nu$. Further, for $B \in \mathcal{B}(\mathcal{Z})$, it holds that
\begin{equation*}
    P(B) = \int_B p(z)\mu(\mathrm{d}z) = \int_{\gamma^{-1}(B)} p(\gamma(z_0))\nu(\mathrm{d}z_0) = \int_{\gamma^{-1}(B)} q(z_0)\nu(\mathrm{d}z_0) = Q(\gamma^{-1}(B)).
\end{equation*}
Thus, $P$ is an element of $\mathcal{P}$.
\par
Next, specify the tangent set $T(P_0)$ as in Example~25.16 in \textcite{Vaart1998}. Let $p_0$ be the density function of the true dgp $P_0 \in \mathcal{P}$. For each $h \in L_2^0(P_0) = \dbrace{h \in L_2(P_0): \int hdP_0 = 0}$, define the path $t \mapsto P_t$ by $P_t(B) = \int_B p_t(z)\mu(\mathrm{d}z)$, where $p_t:\mathcal{Z}\to[0,\infty)$ is the density function given by 
\begin{equation*}
    p_t(z) = \frac{\psi(th(z))p_0(z)}{\expect_{P_0}[\psi(th)]}.
\end{equation*}
Here, $\psi:\mathbb{R}\to(0,\infty)$ is a bounded and continuously differentiable function such that (i) the derivative $\psi'$ is bounded, (ii) $\psi'/\psi$ is bounded, and (iii) $\psi(0) = \psi'(0) = 1$. For instance, $\psi(x) = 2(1+e^{-2x})^{-1}$ meets these conditions. Note that $P_t$ is absolutely continuous with respect to $\mu$, where $P_t \in \mathcal{P}$. Moreover, the path $t \mapsto P_t$ satisfies \cref{eq:DQM} with score function $h$, which implies $L_2^0(P_0) \subset T(P_0)$. Conversely, since any score function satisfying \cref{eq:DQM} is necessarily an element of $L_2^0(P_0)$, $T(P_0) = L_2^0(P_0)$ \parencite[see, e.g.,][Lemma~25.14]{Vaart1998}.
\par
The intermediate parameter is pathwise differentiable at $P_0$ relative to the tangent set. To observe this, note that $\theta(P_t)$ can be expressed as in 
\begin{equation*}
    \theta(P_t) = \dparen*{ \frac{\expect_{P_t}[YD]}{P_t(D=1)},~\frac{\expect_{P_t}[Y(1-D)]}{P_t(D=0)},~P_t(D=1) }^\intercal.
\end{equation*}
Noting that 
\begin{equation*}
    \frac{\psi(th(z))}{\expect_{P_0}[\psi(th)]} - 1 = th(z) + o(t), 
\end{equation*}
then 
\begin{equation*}
    t^{-1}\dparen*{P_t(D=1) - P_0(D=1)} \to \expect_{P_0}[Dh(Z)].
\end{equation*}
Based on this convergence, 
\begin{align*}
    \MoveEqLeft
    t^{-1}\dparen*{ \frac{\expect_{P_t}[YD]}{P_t(D=1)} - \frac{\expect_{P_0}[YD]}{P_0(D=1)} }\\
    = {} & \frac{t^{-1}\dparen*{ \expect_{P_t}[YD] - \expect_{P_0}[YD] }}{P_0(D=1)} - \frac{\expect_{P_t}[YD]\cdot t^{-1}\dparen*{ P_t(D=1) - P_0(D=1) }}{P_t(D=1)P_0(D=1)}\\
    \to {} & \expect_{P_0}\dbrack*{ \frac{(Y - \mu_{1,0})D}{p_0}h(Z) }.
\end{align*}
Similarly,
\begin{align*}
    \MoveEqLeft
    t^{-1}\dparen*{ \frac{\expect_{P_t}[Y(1-D)]}{P_t(D=0)} - \frac{\expect_{P_0}[Y(1-D)]}{P_0(D=0)} } \to \expect_{P_0}\dbrack*{ \frac{(Y - \mu_{0,0})(1-D)}{1-p_0}h(Z) }.
\end{align*}
Therefore, 
\begin{equation*}
    \Dot{\theta}_0(h) = \int
    \dparen*{\frac{(y - \mu_{1,0})d}{p_0},\frac{(y - \mu_{0,0})(1-d)}{1-p_0},d}^\intercal h(z) P_0(\mathrm{d}z).
\end{equation*}
One can easily observe that 
\begin{equation*}
    \tilde{\theta}_0(Z) = \dparen*{ \frac{(Y - \mu_{1,0})D}{p_0}, \frac{(Y - \mu_{0,0})(1-D)}{1-p_0}, D - p_0 }^\intercal 
\end{equation*}
satisfies $\tilde{\theta}_0 \in (T(P_0))^3$ and $\Dot{\theta}_0(h) = \dangle{\tilde{\theta}_0,h}$, and, hence, it is the efficient influence function.
\par
The best regularity of 
\begin{equation*}
    \hat{\theta}_n
    = \dparen*{ \frac{\sum_{i=1}^n Y_iD_i}{\sum_{i=1}^n D_i},~\frac{\sum_{i=1}^nY_i(1-D_i)}{\sum_{i=1}^n (1-D_i)},~\frac{\sum_{i=1}^n D_i}{n} }^\intercal 
\end{equation*}
readily follows from Lemma~25.23 in \textcite{Vaart1998}, as 
\begin{align*}
    \MoveEqLeft
    \sqrt{n}(\hat{\theta}_n - \theta(P_0)) \\
    = {} & \dparen*{ \frac{n^{-1/2}\sum_{i=1}^n (Y_i-\mu_{1,0})D_i}{\sum_{i=1}^n D_i/n},\frac{n^{-1/2}\sum_{i=1}^n(Y_i-\mu_{0,0})(1-D_i)}{\sum_{i=1}^n (1-D_i)/n},n^{-1/2}\sum_{i=1}^n (D_i - p_0) }^\intercal\\
    = {} & \frac{1}{\sqrt{n}}\sum_{i=1}^n \tilde{\theta}_0(Z_i) + o_p(1),
\end{align*}
where the second equality follows $\sum_{i=1}^nD_i / n \overset{p}{\to} p_0$

\subsection{\texorpdfstring{\cref{ex:experimental-study-with-non-random-sampling}}{Example~2}}
\label{sec:example2-detail}
(On page \pageref{ex:example2-third}) The study began with the introduction of the possible data-generating processes $\mathcal{P}$. Define $\mathcal{Y}$ and $\mathcal{D}$ following \cref{sec:example1-detail}. For simplicity, assume $\mathcal{Y}$ contains zero. Moreover, think of $\mathcal{S}=\{0,1\}$ as the metric space with the discrete metric. Then, $\mathcal{Z}_0 = \mathcal{Y}\times\mathcal{Y}\times\mathcal{D}\times\mathcal{S}$ to which the potential observation $(Y(1),Y(0),D,S)$ belongs is a complete and separable metric space for the corresponding product metric. Similarly, $\mathcal{Z} = \mathcal{Y} \times \mathcal{D}\times \mathcal{S}$ to which the observation $(SY,SD,S)$ belongs is also a complete separable metric space for the corresponding product metric. It is convenient to think as if an experiment had also been conducted on the population with $S=0$. Thus, as the class of possible distributions of $(Y(1),Y(0),D,S)$, let $\mathcal{Q} = \dbrace{Q: Q \ll \nu,~ (Y(1),Y(0))\perp D|S }$. Here, $\nu$ is the product measure on $\mathcal{Z}_0$ given by $\nu = \nu_{\mathcal{Y}}\otimes\nu_{\mathcal{Y}}\otimes\nu_{\mathcal{D}}\otimes\nu_{\mathcal{S}}$, where $\nu_{\mathcal{Y}}$ is a finite measure on $\mathcal{Y}$, and $\nu_{\mathcal{D}}$ and $\nu_{\mathcal{S}}$ are the counting measures on $\{0,1\}$. Define the observation function as $\gamma(Y(1),Y(0),D,S) = \dparen{\tilde{Y},\tilde{D},S}$, where $\tilde{Y} = SY$, $\tilde{D} = SD$, and $Y = Y(1)D + Y(0)(1-D)$. Then, define the class $\mathcal{P}$ of possible dgps by $\mathcal{P} = \dbrace*{ Q\circ\gamma^{-1} : Q \in \mathcal{Q} }$.
\par
The class $\mathcal{P}$ equals $\{P:P\ll \mu\}$, where $\mu = \nu \circ \gamma^{-1}$. The inclusion, $\mathcal{P} \subset \{P: P \ll \mu\}$, can be shown in the same way as \cref{sec:example1-detail}. To show the converse inclusion, let $P \ll \mu$. There exists a density function $p:\mathcal{Z}\to[0,\infty)$ such that $P(B) = \int_B p(z)\mu(\mathrm{d}z)$ for all $B \in \mathcal{B}(\mathcal{Z})$. By Corollary 10.4.17 in \textcite{Bogachev2007}, there exist regular conditional measures $\mu(\cdot,\tilde{d},s)$ on $\mathcal{B}(\mathcal{Y})$ and $\mu(\cdot,s)$ on $\mathcal{B}(\mathcal{D})$, and 
\begin{equation*}
    \int_{\mathcal{Z}} p(z)\mu(\mathrm{d}z) 
    = \int_{\mathcal{S}} \int_{\mathcal{D}} \int_{\mathcal{Y}} p(z)\mu(\mathrm{d}\tilde{y},\tilde{d},s)\mu(\mathrm{d}\tilde{d},s)\mu_{S}(s),
\end{equation*}
where $\mu_{S}$ denotes the image of $\mu$ under the projection $(\tilde{y},\tilde{d},s) \mapsto s$. Then, letting
\begin{align*}
    &p_{\tilde{Y}}(\tilde{y}|\tilde{D} = \tilde{d},S = s) = \frac{p(z)}{p_{(\tilde{D},S)}(\tilde{d},s)}, \quad p_{(\tilde{D},S)}(\tilde{d},s) = \int_{\mathcal{Y}} p(z) \mu(\mathrm{d}\tilde{y},\tilde{d},s),\\
    &p_{\tilde{D}}(\tilde{d}|S = s) = \frac{p_{(\tilde{D},S)}(\tilde{d},s)}{p_S(s)},\quad\text{and}\quad p_{S}(s) = \int_{\mathcal{D}} p_{(\tilde{D},S)}(\tilde{d},s)\mu(\mathrm{d}\tilde{d},s),
\end{align*}
the density $p(z)$ can be factorized as in $p(z) = p_{\tilde{Y}}(\tilde{y}|\tilde{D} = \tilde{d}, S = s)p_{\tilde{D}}(\tilde{d}|S = s)p_S(s)$. Define 
\begin{equation*}
    q((y_1,y_0,d,s)) = p_{\tilde{Y}}(y_1|\tilde{D} = 1,S=1)p_{\tilde{Y}}(y_0|\tilde{D} = 0,S=1)p_{\tilde{D}}(d|S = 1)p_S(s),
\end{equation*}
and $Q(A) = \int_A q(z_0)\nu(\mathrm{d}z_0)$ for all $A \in \mathcal{B}(\mathcal{Z}_0)$. By construction, $Q \ll \nu$. Further, $(Y(1),Y(0))\perp D | S$ under $Q$ as follows:
\begin{align*}
    \MoveEqLeft
    Q\dparen{(Y(1),Y(0)) \in A_1, D \in A_2|S = s}\\
    = {} & \dparen*{ \int_{\mathcal{Y}\times\mathcal{Y}\times\mathcal{D}} q(z_0) \nu_{\mathcal{Y}}\otimes\nu_{\mathcal{Y}}\otimes\nu_{\mathcal{D}}(\mathrm{d}(y_1,y_0,d))}^{-1}\\
    & \times\int_{\mathcal{Z}_0} \Bigl(1\{(Y(1),Y(0)) \in A_1, D \in A_2\} p_{\tilde{Y}}(y_1|\tilde{D} = 1,S=1)p_{\tilde{Y}}(y_0|\tilde{D} = 0,S=1)\\
    & \qquad \times p_{\tilde{D}}(d|S = 1)\Bigr) \nu_{\mathcal{Y}}\otimes\nu_{\mathcal{Y}}\otimes\nu_{\mathcal{D}}(\mathrm{d}(y_1,y_0,d))\\
    = {} & \dparen*{ \int_{\mathcal{Y}\times\mathcal{Y}\times\mathcal{D}} q(z_0) \nu_{\mathcal{Y}}\otimes\nu_{\mathcal{Y}}\otimes\nu_{\mathcal{D}}(\mathrm{d}(y_1,y_0,d))}^{-1}\\
    & \int_{\mathcal{Z}_0} 1\{(Y(1),Y(0)) \in A_1\}p_{\tilde{Y}}(y_1|\tilde{D} = 1,S=1)p_{\tilde{Y}}(y_0|\tilde{D} = 0,S=1) \nu_{\mathcal{Y}}\otimes\nu_{\mathcal{Y}}(\mathrm{d}(y_1,y_0))\\
    & \times \int_{\mathcal{Z}_0} 1\{D \in A_2\} p_{\tilde{D}}(d|S = 1) \nu_{\mathcal{D}}(\mathrm{d}d)\\
    = {} & Q((Y(1),Y(0)) \in A_1|S = s)Q(D \in A_2|S=s).
\end{align*}
Combine these observations, to obtain $Q \in \mathcal{P}$.
\par
The tangent set $T(P_0)$ is equal to $L_2^0(P_0)$, which can be shown in the same way as \cref{sec:example1-detail}. Moreover, the intermediate parameter,
\begin{equation*}
    \theta(P_t) = \dparen*{ \frac{\expect_{P_t}[YDS]}{P_t(D = 1, S = 1)},~ \frac{\expect_{P_t}[Y(1-D)S]}{P_t(D = 0, S = 1)}, ~ P_t(S = 1) }^\intercal,
\end{equation*}
is also pathwise differentiable relative to this tangent set. In particular, a similar argument as that of \cref{sec:example1-detail} concludes that the efficient influence function is
\begin{align*}
    \tilde{\theta}_0(Z)
    = 
    \dparen*{ 
        \frac{(Y - \mu_{1,0})DS}{P_0(D=1,S=1)},
        \frac{(Y - \mu_{0,0})(1-D)S}{P_0(D=0,S=1)},  
        S - p_0
    }^\intercal.
\end{align*}
Again, Lemma~25.23 in \textcite{Vaart1998} gives the best regularity of 
\begin{equation*}
    \hat{\theta}_n = (\hat{\mu}_{1,n},\hat{\mu}_{0,n},\hat{p}_n) = \dparen*{ \frac{\sum_{i=1}^n Y_iD_iS_i}{\sum_{i=1}^nD_iS_i},~\frac{\sum_{i=1}^nY_i(1-D_i)S_i}{\sum_{i=1}^n (1-D_i)S_i},\frac{\sum_{i=1}^n S_i}{n} }^\intercal.
\end{equation*} 
\par
(On page \pageref{ex:example2-fifth}) To obtain the optimal adjustment term $w_{\theta_0}^*$ that minimizes the asymptotic maximum RIMR for the case $\tau_L(\theta_0) = 0 < \tau_U(\theta_0)$, it suffices to solve 
\begin{equation}
    \adjustlimits{\inf}_{w \in \mathbb{R}^3}{\sup}_{s \in \mathbb{R}^3} \max\dbrace*{ 0, F(\mathfrak{w},\mathfrak{s}) },
    \label{eq:example2-rimr-lower-bound}
\end{equation}
where $F(\mathfrak{w},\mathfrak{s}) = \min\dbrace*{ \mathfrak{s},0 } - \expect\dbrack*{ \min\dbrace*{ \tau_{L,\theta_0}'(\mathbb{G}_0) + \mathfrak{w} + \mathfrak{s}, 0 } }$, $\mathfrak{w} = \tau_{L,\theta_0}'(w)$, and $\mathfrak{s} = \tau_{L,\theta_0}'(s)$. As 
\begin{equation*}
    \tau_{L,\theta_0}'(\mathbb{G}_0) + \mathfrak{w} + \mathfrak{s} \sim N(\mathfrak{w}+\mathfrak{s},\sigma^2), \quad \sigma^2 = \dparen*{\frac{\partial \tau_L}{\partial \theta}\Bigr|_{\theta = \theta_0}}^\intercal\Sigma_{\theta_0}\dparen*{\frac{\partial \tau_L}{\partial \theta}\Bigr|_{\theta = \theta_0}},
\end{equation*}
the direct calculation yields
\begin{align*}
    \expect\dbrack*{ \min\dbrace*{ \tau_{L,\theta_0}'(\mathbb{G}_0) + \mathfrak{w} + \mathfrak{s}, 0 } }
    = -\frac{\sigma}{\sqrt{2\pi}}\exp\dparen*{ -\frac{(\mathfrak{w}+\mathfrak{s})^2}{2\sigma^2} } + (\mathfrak{w}+\mathfrak{s})\dparen*{ 1 - \Phi\dparen*{\frac{\mathfrak{w}+\mathfrak{s}}{\sigma}} }.
\end{align*}
Fix $w \in \mathbb{R}^3$. Then, the inner supremum in \cref{eq:example2-rimr-lower-bound} is equal to $\max\dbrace{ 0,\sup_{\mathfrak{s} \in \mathbb{R}}F(\mathfrak{w},\mathfrak{s}) }$, where
\begin{equation*}
    F(\mathfrak{w},\mathfrak{s}) = \min\dbrace*{\mathfrak{s},0} + \frac{\sigma}{\sqrt{2\pi}}\exp\dparen*{ -\frac{(\mathfrak{w}+\mathfrak{s})^2}{2\sigma^2} } - (\mathfrak{w}+\mathfrak{s})\dparen*{ 1 - \Phi\dparen*{\frac{\mathfrak{w}+\mathfrak{s}}{\sigma}}}.
\end{equation*}
When $\mathfrak{s} \in (0,\infty)$, $F(\mathfrak{w},\mathfrak{s})$ is decreasing in $\mathfrak{s}$, as $\frac{\partial F(\mathfrak{w},\mathfrak{s})}{\partial \mathfrak{s}} = \Phi\dparen*{ \frac{\mathfrak{w} + \mathfrak{s}}{\sigma}} - 1$. Conversely, when $\mathfrak{s} \in (-\infty,0)$, $F(\mathfrak{w},\mathfrak{s})$ is increasing, as $\frac{\partial F(\mathfrak{w},\mathfrak{s})}{\partial \mathfrak{s}} = \Phi\dparen*{ \frac{\mathfrak{w} + \mathfrak{s}}{\sigma}}$. Therefore, the supremum is attained at $\mathfrak{s} = 0$, with value 
\begin{equation*}
    F(\mathfrak{w},0) = \frac{\sigma}{\sqrt{2\pi}}\exp\dparen*{ -\frac{\mathfrak{w}^2}{2\sigma^2} } - \mathfrak{w}\dparen*{ 1 - \Phi\dparen*{\frac{\mathfrak{w}}{\sigma}}},
\end{equation*}
and \cref{eq:example2-rimr-lower-bound} equals $\inf_{\mathfrak{w} \in \mathbb{R}} \max\dbrace{ 0, F(\mathfrak{w},0) }$. Observe that 
\begin{equation*}
    \lim_{\mathfrak{w} \to -\infty} F(\mathfrak{w},0) = +\infty, \quad \lim_{\mathfrak{w} \to \infty} F(\mathfrak{w},0) = 0, \quad \text{and} \quad \frac{\partial F(\mathfrak{w},0)}{\partial \mathfrak{w}} = \Phi\dparen*{ \frac{\mathfrak{w}}{\sigma} } - 1.
\end{equation*}
By letting $\mathfrak{w} \to \infty$, the infimum is zero.
\par
The optimal adjustment term $w_{\theta_0}^\dagger$ that minimizes the asymptotic maximum MSE can be obtained by solving 
\begin{equation}
    \adjustlimits\inf_{\mathfrak{w} \in \mathbb{R}}\sup_{\mathfrak{s} \in \mathbb{R}} G(\mathfrak{w},\mathfrak{s})
    \label{eq:example2-mse-lower-bound}
\end{equation}
where $G(\mathfrak{w},\mathfrak{s}) =  \mathbb{E}\dbrack*{ \dparen*{ \min\dbrace*{ \tau_{L,\theta_0}(\mathbb{G}_0) + \mathfrak{w} + \mathfrak{s}, 0 } - \min\dbrace*{ \mathfrak{s}, 0 } }^2 }$. Fix $\mathfrak{w} \in \mathbb{R}$. When $\mathfrak{s} \geq 0$, 
\begin{equation*}
    G(\mathfrak{w},\mathfrak{s}) = \dparen*{ (\mathfrak{w} + \mathfrak{s})^2 + \sigma^2 }\Phi\dparen*{ -\frac{\mathfrak{w}+\mathfrak{s}}{\sigma} } - \frac{(\mathfrak{w}+\mathfrak{s})\sigma}{\sqrt{2\pi}}\exp\dparen*{ - \frac{(\mathfrak{w}+\mathfrak{s})^2}{2\sigma^2} }.
\end{equation*}
As $\frac{\partial^2 G(\mathfrak{w},\mathfrak{s})}{\partial \mathfrak{s}^2} = 2\Phi\dparen*{ -\frac{\mathfrak{w}+\mathfrak{s}}{\sigma} }$, $G(\mathfrak{w},\mathfrak{s})$ is convex in $\mathfrak{s} \geq 0$. Hence, $\sup_{\mathfrak{s} \geq 0} G(\mathfrak{w},\mathfrak{s})$ is attained at $\mathfrak{s} = 0$ or $\mathfrak{s} \to \infty$. Observing that $\lim_{\mathfrak{s} \to \infty} G(\mathfrak{w},\mathfrak{s}) = 0$ and $G(\mathfrak{w},\mathfrak{s}) \geq 0$, 
\begin{equation*}
    \sup_{\mathfrak{s} \geq 0} G(\mathfrak{w},\mathfrak{s}) = G(\mathfrak{w},0) = (\mathfrak{w}^2+\sigma^2)\Phi\dparen*{-\frac{\mathfrak{w}}{\sigma}} - \frac{\mathfrak{w}\sigma}{\sqrt{2\pi}}\exp\dparen*{ -\frac{\mathfrak{w}^2}{2\sigma^2} }.
\end{equation*}
When $\mathfrak{s} < 0$,
\begin{equation*}
    G(\mathfrak{w},\mathfrak{s}) = \frac{1}{\sqrt{2\pi\sigma^2}}\int_{-\infty}^{-\mathfrak{s}} x^2\exp\dparen*{ -\frac{(x-\mathfrak{w})^2}{2\sigma^2} }\mathrm{d}x + \mathfrak{s}^2\dparen*{ 1 - \Phi\dparen*{ -\frac{\mathfrak{w}+\mathfrak{s}}{\sigma} } }.
\end{equation*}
Since $\frac{\partial G(\mathfrak{w},\mathfrak{s})}{\partial \mathfrak{s}} = 2\mathfrak{s}\dparen*{ 1 - \Phi\dparen*{ - \frac{\mathfrak{w}+\mathfrak{s}}{\sigma}} }$, $G(\mathfrak{w},\mathfrak{s})$ is decreasing in $\mathfrak{s} < 0$. By letting $\mathfrak{s} \to -\infty$, 
\begin{equation*}
    \sup_{\mathfrak{s} < 0} G(\mathfrak{w},\mathfrak{s}) = \mathfrak{w}^2 + \sigma^2.
\end{equation*}
To observe which of $\sup_{\mathfrak{s} < 0} G(\mathfrak{w},\mathfrak{s})$ or $\sup_{\mathfrak{s} \geq 0} G(\mathfrak{w},\mathfrak{s})$ is larger, let
\begin{align*}
    \Delta G(\mathfrak{w}) 
    &= \sup_{\mathfrak{s} < 0} G(\mathfrak{w},\mathfrak{s}) - \sup_{\mathfrak{s} \geq 0} G(\mathfrak{w},\mathfrak{s})\\
    &= (\mathfrak{w}^2+\sigma^2)\dparen*{1 - \Phi\dparen*{-\frac{\mathfrak{w}}{\sigma}}} + \frac{\mathfrak{w}\sigma}{\sqrt{2\pi}}\exp\dparen*{ -\frac{\mathfrak{w}^2}{2\sigma^2} }.
\end{align*}
As
\begin{equation*}
    \lim_{\mathfrak{w} \to -\infty} \Delta G(\mathfrak{w}) = 0, \quad 
    \lim_{\mathfrak{w} \to -\infty} \frac{\partial \Delta G(\mathfrak{w})}{\partial \mathfrak{w}} = 0, \quad\text{and}\quad
    \frac{\partial^2 \Delta G(\mathfrak{w})}{\partial \mathfrak{w}^2} > 0,
\end{equation*}
$\Delta G(\mathfrak{w})$ is positive for any $\mathfrak{w}$. Hence, $\sup_{\mathfrak{s} \in \mathbb{R}} G(\mathfrak{w},\mathfrak{s}) = \mathfrak{w}^2 + \sigma^2$ for any $\mathfrak{w}$. Therefore, \cref{eq:example2-mse-lower-bound} is solved by $w_{\theta_0}^\dagger = 0$ with value equal to $\sigma^2$.

\printbibliography

\end{document}